\documentclass{LMCS}

\def\dOi{9(4:9)2013}
\lmcsheading%
{\dOi}
{1--42}
{}
{}
{Oct.~24, 2012}
{Oct.~31, 2013}
{}

\ACMCCS{[{\bf Theory of computation}]: Computational complexity and
  cryptography---Complexity classes\,/\,Complexity theory and logic;
  Formal languages and automata theory---Formalisms---Rewrite systems;
  Logic---Logic and verification}

\usepackage[utf8x]{inputenc}
\usepackage{xcolor}
\usepackage{amssymb}
\usepackage{amsmath}
\usepackage{amsthm}
\usepackage{graphicx}
\usepackage[english]{babel}
\usepackage[numbers,sort&compress]{natbib}   
\usepackage{tikz}
\usetikzlibrary{decorations,shapes.callouts,arrows,positioning}
\usepackage{slashed} 
\usepackage{url}
\usepackage[bookmarks=false]{hyperref}
\usepackage{xargs}
\usepackage{stmaryrd}
\usepackage{bbding}
\usepackage{xifthen}
\usepackage{enumitem}
\usepackage{pop}
\newtheorem{example}[thm]{Example}
\newtheorem{definition}[thm]{Definition}
\newtheorem{theorem}[thm]{Theorem}
\newtheorem{lemma}[thm]{Lemma}
\newtheorem{corollary}[thm]{Corollary}

\newtheorem{proposition}[thm]{Proposition}
\newtheorem*{remark}{Remark}

\begin{document}
\title[Polynomial Path Orders]{Polynomial Path Orders}

\author[M.~Avanzini]{Martin Avanzini}
\address{Institute of Computer Science\\University of Innsbruck\\ Austria}
\email{\{martin.avanzini,georg.moser\}@uibk.ac.at}

\author[G.~Moser]{Georg Moser}
\address{\vspace{-18 pt}}
\thanks{This work is partially supported by FWF (Austrian Science Fund) project I-603-N18}

\renewcommand{\labelitemi}{-}

\keywords{Term Rewriting, Complexity Analysis, Implicit Computational
Complexity, Automation}

\subjclass{F.4.1, F.4.2, F.1.3, D.2.4}

\begin{abstract}
This paper is concerned with the complexity analysis of
constructor term rewrite systems and its ramification in implicit
computational complexity.
We introduce a path order with multiset status, 
the \emph{polynomial path order} \POPSTAR, that is applicable
in two related, but distinct contexts.
On the one hand \POPSTAR\ induces polynomial innermost runtime complexity
and hence may serve as a syntactic, and fully automatable, method to analyse
the innermost runtime complexity of term rewrite systems.
On the other hand \POPSTAR\ provides an order-theoretic 
characterisation of the polytime computable functions: 
the polytime computable functions are exactly the functions
computable by an orthogonal constructor TRS compatible with \POPSTAR.
\end{abstract}

\maketitle

\section{Introduction}\label{s:intro}

In this paper we are concerned with the complexity analysis of
constructor term rewrite systems.
Since term rewrite systems (TRSs for short) underlie much of 
declarative programming, time complexity of functions defined by TRSs
is of particular interest.

In rewriting two notions of complexity have been widely studied. 
Hofbauer and Lautemann proposed to assess the complexity of a given TRS 
as the maximal length of derivation sequences. More precisely
the \emph{derivational complexity function} relates the maximal length of a 
derivation with the size of the starting term~\cite{HL89}. 
As an alternative Hirokawa and the second author proposed to study the 
\emph{runtime complexity function}~\cite{HM08}, which forms a variation
of the derivational complexity function. Instead of all possible derivations,
one studies only derivations with starting terms whose arguments are constructor terms
(aka \emph{basic terms}), see also~\cite{BCMT01}.
In the context of this paper, runtime complexity is the more natural notion. 
We emphasise that the runtime complexity of a rewrite system
forms a \emph{polynomially invariant} cost model~\cite{Boas:TCS:90}, cf.~Section~\ref{s:basics}.

To motivate our studies, we present a natural encoding of
the well-known satisfiability problem $\SAT$ of propositional logic 
as a TRS.\@ Given a propositional formula in conjunctive normal form, 
the TRS $\RSsat$ given below computes a satisfying assignment if it exists.
Note that $\RSsat$ is not confluent, i.e., the computation is performed 
nondeterministically. 
The rewrite system $\RSsat$ thus encodes 
the function problem \emph{$\FSAT$} associated with the satisfiability
problem. 
$\FSAT$ is complete for the class of \emph{function problems over $\NP$} 
($\FNP$ for short). See Section~\ref{s:basics} where $\FNP$ is formally defined. 
As corollary to the polynomial invariance of
the runtime complexity of rewrite systems, we obtain that the runtime complexity
of $\RSsat$ is expected to be polynomial.

\begin{example}
\label{ex:rssat}
Consider the following (non-confluent) TRS $\RSsat$:%
\footnote{This is a slight variant of Example \textsf{TCT\_12/sat.xml} in the current 
Termination Problem Database (TPDB) version 8.0.}
\begin{alignat*}{4}
\rlabel{RSsat:neg:Z} && \mneg(+x) & \to -x 
& ~
\rlabel{RSsat:neq:O} && \mneg(-x) & \to +x 
\\[2mm]
\rlabel{RSsat:eq:ZO} && \meq(\mZ(x),\mO(y)) & \to \mfalse 
& ~
\rlabel{RSsat:eq:ZZ} && \meq(\mZ(x),\mZ(y)) & \to \meq(x,y) 
\\
\rlabel{RSsat:eq:OZ} && \meq(\mO(x),\mZ(y)) & \to \mfalse 
& ~
\rlabel{RSsat:eq:OO} && \meq(\mO(x),\mO(y)) & \to \meq(x,y) 
\\
\rlabel{RSsat:eq:np} && \meq(-x,+y) & \to \mfalse 
& ~
\rlabel{RSsat:eq:nn} && \meq(-x,-y) & \to \meq(x,y) 
\\
\rlabel{RSsat:eq:pn} && \meq(+x,-y) & \to \mfalse 
& ~
\rlabel{RSsat:eq:pp} && \meq(+x,-y) & \to \meq(x,y) 
\\
\rlabel{RSsat:eq:ee} && \meq(\varepsilon,\varepsilon) & \to \mtrue 
\\[2mm]
\rlabel{RSsat:if:t} && \mif(\mtrue,t,e) & \to t 
& ~
\rlabel{RSsat:if:f} && \mif(\mfalse,t,e) & \to e 
\\[2mm]
\rlabel{RSsat:verify:b} && \verify(\nil) & \to \mtrue 
& ~
\rlabel{RSsat:verify:r} && \verify(l \cons ls) & \to \mif(\member(\mneg(l),ls), \mfalse, \verify(ls))
\\[2mm]
\rlabel{RSsat:member:b} && \member(x,\nil) & \to \mfalse
& ~
\rlabel{RSsat:member:r} && \member(x,y \cons ys) & \to \mif(\meq(x,y), \mtrue, \member(x,ys)) 
\\[2mm]
\rlabel{RSsat:guess:b} && \guess(\nil) & \to \nil
& ~
\rlabel{RSsat:guess:r} && \guess(c \cons cs) & \to \choice(c) \cons \guess(cs)
\\[2mm]
\rlabel{RSsat:choice:a} && \choice(a \cons \nil) & \to a
&~
\rlabel{RSsat:choice:b} && \choice(a \cons b \cons bs) & \to \choice(b \cons bs)
\\
\rlabel{RSsat:choice:a2} && \choice(a \cons b \cons bs) & \to a
\\[2mm]
\rlabel{RSsat:issat} && \issat(cs) & \to \mparbox{1cm}{\issat'(\guess(cs))}
\\
\rlabel{RSsat:issat'} && \issat'(as) & \to \mparbox{1cm}{\mif(\verify(as),as,\unsat) \tpkt}
\end{alignat*}
Atoms are encoded as binary strings 
(built from the constant $\varepsilon$, and unary constructors $\mZ$ and $\mO$), 
the unary constructors ($+$) and ($-$) lift atoms to positive 
and negative literals respectively.
The rules~\rref{RSsat:neg:Z}--\rref{RSsat:eq:ee} 
define negation and equality on this representation of literals.

%
Clauses are lists of literals, 
clause sets are denoted by lists of clauses.
Lists are constructed in the usual way 
using a constant $\nil$ and binary constructor $(\cons)$. 
Call a list of literals consistent, if an atom does not 
occur positively and negatively. This is formalised by
rules~\rref{RSsat:verify:b}--\rref{RSsat:member:r}. 
A clause set $cs$ is satisfiable if and only if 
there exists a list of literals $as$, denoting a satisfying assignment, 
such that $as$ is consistent and contains at least one literal from every clause $c$.
The rules~\rref{RSsat:guess:b}--\rref{RSsat:choice:a2} 
are used to generate a candidate list $as$ that contains for each clause 
one literal.
Using these auxiliary rules, the algorithm is implemented 
by rules~\rref{RSsat:issat} and~\rref{RSsat:issat'}.
Given a clause set $cs$, a candidate list $as$ 
is guessed and its consistency is checked.
If this check succeeds the list $as$ is returned.
\end{example}

It is easy to see that $\RSsat$ is terminating, for example
this can be verified by showing compatibility with the multiset
path order (\MPO\ for short)~\cite{TeReSe}.
%
It is a standard exercise in rewriting to assess the 
complexity of rewrite systems via an analysis of termination techniques
and it is a well-known result that \MPO\ induces 
primitive recursive derivational and runtime complexity~\cite{H92,B95,MW03}.
Furthermore, \MPO\ \emph{characterises} the 
\emph{primitive recursive functions} ($\mathcal{PR}$ for short)~\cite{CW97}:
any function computed by an $\MPO$-terminating TRS is primitive recursive, vice versa, 
any primitive recursive function can be stated as an $\MPO$-terminating TRS.
However, from these results we can only conclude that 
the runtime complexity function of $\RSsat$ is
bounded by a primitive recursive function, which is hardly revealing.
This motivates the quest for a ``polynomial path order'' depicted by
the question mark in Figure~\ref{fig:3}.%
\footnote{
Solid lines indicate a characterisation,
whereas dashed lines indicate an inclusion relationship.}
Such an order should be a restriction of \MPO, but \emph{miniaturises} its properties: it would be expected
that this order induces \emph{polynomial} runtime complexity and
provides a characterisation of the class $\FP$ of \emph{functions computable
in polynomial time}.

\begin{figure}
  \centering
  \begin{tikzpicture}
    \node(1) {\MPO};
    \node(2) [right=of 1] {\emph{prim.\ rec.\ runtime compl.}};
    \node(5) [right=of 2] {$\mathcal{PR}$};
    \node(3) [below=of 1] {\textsf{?}};
    \node(4) [below=of 2, yshift=2pt] {\emph{polynomial runtime compl. }};
    \node(6) [below=of 5] {$\FP$};

    \draw[->,dashed] (1.south) -- (3.north);
    \draw[->,dashed] (2.south) -- (4.north);
    \draw[->,dashed] (5.south) -- (6.north);
    \draw[->] (1.east) -- (2.west);
    \draw[->] (2.east) -- (5.west);
    \draw[->] (3.east) -- (4.west);
    \draw[->] (4.east) -- (6.west);
  \end{tikzpicture}  
  \caption{The Quest for ``Polynomial Path Orders''}
  \label{fig:3}
\end{figure}

In this paper, we propose the \emph{polynomial path order} (\emph{\POPSTAR} for short)
as such a miniaturisation of \MPO: \POPSTAR\ induces
polynomial runtime complexity (for innermost rewriting) and at the
same time yields a characterisation of $\FP$. 
In the design of \POPSTAR\ we have striven for a \emph{maximal} miniaturisation
of \MPO, so that these key features of \POPSTAR\ remain intact. 
Alas, some of the essential properties of \MPO\ cannot be preserved. 
First, \POPSTAR\ can only analyse the \emph{runtime complexity} of TRSs; the derivational
complexity induced by \POPSTAR\ is (at least) double-exponential
(Example~\ref{ex:dc}). Second, the restriction to \emph{innermost}
rewriting is essential (Example~\ref{ex:outermost}) and
finally, our result only holds for \emph{constructor} TRS (Example~\ref{ex:constructor}).
More precisely, we establish the following results.

\begin{enumerate}[labelsep=*,leftmargin=*]
\item \POPSTAR\ induces polynomial innermost runtime complexity on constructor TRSs. That is,
the innermost runtime complexity function for a constructor TRS compatible with \POPSTAR\ is
polynomially bounded (Theorem~\ref{t:popstar}).
\item \POPSTAR\ captures exactly the class $\FP$ on orthogonal constructor TRSs. That is,
any orthogonal constructor TRS compatible with \POPSTAR\ computes a polytime function.
On the other hand, any function in $\FP$ can be implemented by an orthogonal constructor
TRS compatible with \POPSTAR\ (Theorems~\ref{t:icc:soundness} and~\ref{t:icc:completeness}).

\item We extend upon \POPSTAR\ by proposing 
  a generalisation \POPSTARP, admitting the same properties as outlined above, 
  but that allows to handle more general recursion schemes that make
  use of parameter substitution (Theorem~\ref{t:popstarps}).
\item We have implemented the proposed technique in 
the \emph{Tyrolean Complexity Tool} (\TCT\ for short)~\cite{AM13b}.
The experimental evidence obtained indicates the viability of the method.
\end{enumerate}

\noindent By a comparison with the formal definition in Section~\ref{s:popstar} it
is not difficult to verify that $\RSsat$ is compatible with \POPSTAR\
(cf.~Example~\ref{ex:rssat:2}).
This implies that the number of rewrite steps starting from $\issat(cs)$ 
is polynomially bounded in the size of the CNF $cs$. 
This can be automatically verified by $\TCT$ in a fraction of a second.
Due to a suitable adaption of the polynomial invariance theorem~\cite{AM10b}
(cf.~Proposition~\ref{p:invariance})
we can thus \emph{automatically} conclude that $\FSAT$ belongs to $\FNP$.

The termination order \POPSTAR\ gives 
a syntactic account of the principle of \emph{predicative recursion} as proposed 
by Bellantoni and Cook~\cite{BC92}. 
Conclusively any TRS compatible with \POPSTAR\
is called \emph{predicative recursive}.
Analogously \POPSTAR\ can be 
conceived as syntactic account of Leivant's notion of \emph{tiered recurrence}~\cite{Leivant:1990,L91,Leivant93}, cf.~Simmons~\cite{Simmons:1988}. 
We think that \POPSTAR\ is not only of interest from the
viewpoint of automated runtime complexity, but also from the 
viewpoint of \emph{implicit computational complexity} (\emph{ICC} for
short)~\cite{BMR09,DalLago:2011}. In particular \POPSTAR\ is
applicable to verify closure properties of the class of polytime computable function.
Through our extension \POPSTARP, we reobtain Bellantoni's result that predicative
recursion is closed under parameter substitution (cf.~Section~\ref{s:popstarps}).

Preliminary versions of the presented results appeared in~\cite{AM08,AMS08,AM09b}.
The order $\POPSTAR$ has been introduced in~\cite{AM08}, extended to
quasi-precedences in~\cite{AMS08} and the extension $\POPSTARP$
appeared in~\cite{AM09b}.
Apart from the correction of some shortcomings, we extend our earlier work
in the following way:
First, the presented definition of $\POPSTAR$ is more liberal and 
captures the underlying idea of predicative recursion more precisely, 
compare~\cite[Definition~4]{AM08} and Definition~\ref{d:gpop} from Section~\ref{s:popstar}.
Second, our soundness result (cf.~Theorem~\ref{t:icc:soundness} from Section~\ref{s:icc}) 
is new and more general than similar results presented earlier. In particular
it does no longer require typing of constructors nor the intermediate
step of completely defined TRSs, cf.~\cite{AM08}.
Third, the propositional encoding used in our automation of polynomial path orders 
(cf.~Section~\ref{s:impl}) has been completely overhauled.

\subsection{Related Work}

Polynomial complexity analysis is an active research area
in rewriting. Starting from~\cite{MS08} interest in \emph{automated} polynomial
complexity analysis greatly increased over the last years, 
see for example~\cite{HM08,HZMK10,NEG11,HM11,MMNWZ11}.
This is partly due to the incorporation of a dedicated category for complexity 
into the annual termination competition (TERMCOMP).%
\footnote{\url{http://termcomp.uibk.ac.at/}.} 

There are several accounts of predicative analysis of recursion in the (ICC) literature. 
We mention Marion's \emph{light multiset path order} (\emph{\LMPO} for short)~\cite{M03}.
The path order \LMPO\ provides an order-theoretic characterisation of the class 
$\FP$ and can be also consider as a miniaturisation of \MPO\ of sorts: it is a restriction
of \MPO\ and yields an order-theoretic characterisation of a complexity class. 
On the other hand \LMPO\ cannot be used to characterise the (innermost) runtime complexity
of TRSs. This follows from Example~\ref{ex:RS2} below. 
In particular, although $\RSsat$ is compatible with $\LMPO$, from this we
can only conclude that $\FSAT$ is computable on a nondeterministic 
Turing machine in polynomial time. However, this follows by design
as $\FSAT$ is complete for $\FNP$.
\begin{example}
The TRS $\RSbin$ is given by the following rules:
\label{ex:RS2}
\begin{alignat*}{6}
    \rlabel{RS2:a} && \bin(x,\Null) & \to \ms(\Null) 
    & \quad
    \rlabel{RS2:b} && \bin(\Null,\ms(y)) & \to \Null 
    & \quad    
    \rlabel{RS2:c} && \bin(\ms(x),\ms(y)) & \to \mP(\bin(x,\ms(y)),\bin(x,y))
    \tpkt
\end{alignat*}
For a precedence $\qp$ that fulfils $\bin \sp \ms$ and $\bin \sp \mP$ we obtain 
that $\RSbin$ is compatible with $\LMPO$. 
However it is straightforward to verify that the family of terms
$\bin(\ms^n(\Null),\ms^m(\Null))$ admits (innermost) derivations whose length grows exponentially in $n$.
Still the underlying function can be proven polynomial, essentially relying on memoisation techniques~\cite{M03}.
\end{example}

On the other hand Arai and the second author introduced
the \emph{polynomial path order for $\FP$} (\emph{\POPFP} for short). 
In a similar way as \LMPO, \POPFP\ characterises the class $\FP$
and in addition induces innermost polynomial runtime complexity. However in
comparison to \POPSTAR, \POPFP\ severely lacks applicability as it requires
a specific signature of the given rewrite system. For example \POPFP\
is not directly applicable to $\RSsat$: only a special transformation
of the rewrite system $\RSsat$ can be handled. 

Furthermore, a strengthening of our first main theorem to
runtime complexity can be obtained
if one considers polynomial interpretations, where the interpretations of
constructor symbols is restricted. Such restricted polynomial interpretations
are called \emph{additive} in~\cite{BCMT01}. Note that additive polynomial interpretations
also characterise the functions computable in polytime~\cite{BCMT01}.
Similarly, \emph{quasi-interpretations}~\cite{BMM11} provide an elegant way to characterise
time complexity classes through a combination of syntactic (via restrictions of reduction
orders) and semantic (via quasi-interpretations) considerations. 
To date it is unknown whether quasi-interpretations can be used to assess polynomial runtime complexity
of TRSs. 
Unarguable these semantic techniques admit a better intensionality than the
syntactic characterisation provided through the path order \POPSTAR. 
But semantic methods are notoriously 
difficult to implement efficiently in an automated setting. In particular
we are only aware of one accessible implementation of quasi-interpretations,
our own~\cite{AMS08}. Note that these semantic methods are
not tailored for innermost rewriting, in particular Example~\ref{ex:dup} 
given below cannot be handled, while it can be easily handled by \POPSTAR.

Although we consider here only time complexity, related work 
indicates that the overall approach is general enough to reason also about space complexity.
For instance, the \emph{Knuth-Bendix} order~\cite{BN98} can be miniaturised to characterise linear space~\cite{BonfanteMoser:2010}.
Likewise, \emph{sup-interpretations}~\cite{MP:09} provide a 
semantic technique capable of characterising polynomial space. 

In~\cite{BW96}, Beckmann and Weiermann have given a term rewriting characterisation
of the principle of predicative recursion proposed by Bellantoni and Cook. 
Following ideas proposed by Cichon and Weiermann in~\cite{CW97}, Beckmann and Weiermann
thus reobtain Bellantoni's result that predicative recursion is closed
under parameter recursion. 

We have extended our complexity analysis tool \TCT~\cite{AM13b} with polynomial path orders.
We briefly contrast this implementation to related tools for the static resource analysis 
of programs.
Hoffmann~et~al.~\cite{HAH11} provide an automatic multivariate amortised 
cost analysis exploiting typing, which extends earlier results on amortised 
cost analysis~\cite{Tarjan:1985}.
To indicate the applicability of our method we have employed 
a straightforward (and complexity preserving) 
transformation of the RAML programs considered
in~\cite{HAH11,HAH12} into TRSs. Equipped with \POPSTAR\ our complexity
analyser \TCT\ can handle all examples from~\cite{HAH11}.
Albert et al.~\cite{AAGGPRRZ:2009} present an automated complexity tool
for Java${}^\text{\texttrademark}$ Bytecode programs, 
Alias et al.~\cite{ADFG10} give a complexity and
termination analysis for flowchart programs, and
Gulwani~et~al.~\cite{GMC09} as well as Zuleger~et~al.~\cite{ZulegerGSV11} 
provide an automated complexity tool for C programs.
Very recently Hofmann and Rodriguez proposed in~\cite{HR13} an 
automated resource analysis for object-oriented programs via an 
amortised cost analysis.

\subsection{Outline}
The remainder of this paper is organised as follows.
In the next section we recall basic notions and starting points of
this paper.
In Section~\ref{s:popstar} we introduce polynomial path orders.
In the subsequent Sections~\ref{s:pop} and~\ref{s:embed} we show that the 
innermost runtime complexity of predicative recursive TRSs is polynomially bounded.
As essential tool for this we introduce an extended version of \POPFP.
In Section~\ref{s:icc} we present our ramification of polynomial path orders in ICC.
Parameter substitution is incorporated in Section~\ref{s:popstarps}.
Our implementation is detailed in Section~\ref{s:impl} and experimental evidence 
is provided in Section~\ref{s:exps}. 
Finally, we conclude and present future work in Section~\ref{s:conclusion}.


\section{Preliminaries}\label{s:basics}

We denote by $\N$ the set of natural numbers $\{0,1,2,\dots\}$.
Let $R$ be a binary relation.
The transitive closure of $R$ is denoted by $R^+$ and its transitive and reflexive closure by $R^{\ast}$. 
For a binary relation $R$, we frequently write $a \mathrel{R} b$ instead of $(a,b) \in R$. 
Composition of binary relations $R$ and $S$ is denoted by $R \cdot S$,
and defined in the usual way.
For $n\in \N$ we denote by $R^n$ the \emph{$n$-fold composition of $R$}.
The binary relation $R$ is \emph{well-founded} if 
there exists no infinite chain $a_0, a_1, \dots$ with $a_i \mathrel{R} a_{i+1}$
for all $i \in \N$. Moreover, we say that $R$ is well-founded on a set $A$ if 
there exists no such infinite chain with $a_0 \in A$.
The relation $R$ is \emph{finitely branching} if for all elements $a$, the set $\{b \mid a \mathrel{R} b\}$ is finite.

A \emph{proper order} is an irreflexive and transitive binary relation.
A \emph{preorder} is a reflexive and transitive binary relation. 
An \emph{equivalence relation} is reflexive, symmetric and transitive.

A multiset is a collection in which elements are allowed to
occur more than once. We denote by $\msetover(A)$ the set of multisets over $A$
and write $\mset{a_1,\dots,a_n}$ to denote multisets with elements $a_1, \dots,a_n$.
We use $m_1 \uplus m_2$ for the summation and $m_1 \backslash m_2$ for difference on multisets $m_1$ and $m_2$.
The \emph{multiset extension} $\mextension{R}$ \emph{of a relation $R$ on $A$} is the 
relation on $\msetover(A)$ such that $M_1 \mextension{R} M_2$ if there exists 
multisets $X,Y \in \msetover(A)$ satisfying 
\begin{enumerate}[labelsep=*,leftmargin=*]    
    \item $M_2 = (M_1 \backslash X) \uplus Y$, 
    \item $\varnothing \not= X \subseteq M_1$ and 
    \item for all $y \in Y$ there exists an element $x \in X$ such that $x \mathrel{R} y$.
\end{enumerate}
In order to extend this definition to preorders and equivalences, we follow~\cite{Ferreira95}.
Let $\eqi$ denote an equivalence relation over the set $A$ and let ${\succcurlyeq} = {\succ} \cup {\eqi}$ be a binary
relation over $A$ so that $\succ$ and $\eqi$ are \emph{compatible} in the following sense: 
${{\eqi} \cdot {\succ} \cdot {\eqi}} \subseteq {{\succ}}$. 
Let $\eclass{a}$ denote the \emph{equivalence class of $a \in A$} with respect to $\eqi$.
By the compatibility requirement, 
the extension $\sqsupset$ of $\succ$ to equivalence classes 
such that $\eclass{a}\sqsupset\eclass{b}$ if and only if $a \succ b$, 
is well defined.
We define the \emph{strict multiset extension $\mextension{\succ}$ of $\succcurlyeq$} as 
${M_1} \mextension{\succ} {M_2}$ if and only if ${\eclass{M_1}} \mextension{\sqsupset} {\eclass{M_2}}$.
Further, the \emph{weak multiset extension $\mextension{\succcurlyeq}$ of $\succcurlyeq$} is given 
by ${M_1} \mextension{\succcurlyeq} {M_2}$ if and only if ${\eclass{M_1}} \mextension{\sqsupset} {\eclass{M_2}}$
or ${\eclass{M_1}} = {\eclass{M_2}}$ holds. 
Note that if ${\succcurlyeq}$ is a preorder (on $A$) then
$\mextension{\succ}$ is a proper order
and $\mextension{\succcurlyeq}$ a preorder on $\msetover(A)$, cf.~\cite{Ferreira95}. 
Also $\mextension{\succ}$ is well-founded if $\succ$ is well-founded. 

\subsection{Complexity Theory}

Notations are taken from~\cite{Papa}. 
The \emph{function problem} $F_R$ associated with a binary relation $R$
is defined as follows: given $x$ find some $y$ such that $(x,y) \in R$ holds if $y$ exists, 
otherwise return $\m{no}$. 
A binary relation $R$ on words is called \emph{polynomially balanced}
if for all $(x,y) \in R$, the size of $y$ is polynomially bounded in the 
size of $x$.
The relation $R$ is \emph{polytime decidable} if $(x,y) \in R$ 
is decided by a deterministic \emph{Turing machine} (\emph{TM} for short) $M$ operating in polynomial time.
The class $\NP$ is the class of languages $L$ admitting polynomially balanced, polytime decidable 
relations $R_L$~\cite[Chapter~9]{Papa}: $L =  \{x \mid (x,y) \in R_L \text{ for some $y$}\}$.
The class $\FNP$ is the class of function problems associated with the polynomially balanced and polytime decidable
relations $R_L$ as above.
The class of \emph{polytime computable functions} $\FP$ 
is the subclass resulting if we only consider function problems in $\FNP$ that 
can be solved in polynomial time~\cite[Chapter 10]{Papa}.

Recall that a function problem $F$ \emph{reduces} to a function problem $G$ if there 
exist functions $s$ and $r$, both computable in logarithmic space, such that 
for all $x,y$ with $F$ computing $y$ on input $x$, 
$G$ computes on input $s(x)$ the output $z$ with $r(z) = y$.
Note that both $\FNP$ and $\FP$ are closed under reductions.
We note that nondeterministic Turing machines running in polynomial time compute function 
problems from $\FNP$.
\begin{proposition}\label{p:fnp}
  Let $N$ be a nondeterministic Turing machine that computes the function problem $F$ in polynomial time.
  Then $F \in \FNP$.
\end{proposition}
\begin{proof}
  Define the following relation $R$: $(x,y) \in R$ if and only if $y$ 
  is the encoding of an accepting computation of $N$ on input $x$. For this encoding
  it is sufficient to encode a successful sequence of configurations. 
  Since $N$ operates in polynomial time, the length
  of any computation, 
  and also the size of each configuration, is polynomially bounded. 
  It follows that $R$ is polynomially balanced. 
  As it can be checked in linear time that $y$ encodes an accepting run of $N$ on input $x$, 
  $R$ is polytime decidable.
  Hence the function problem $F_R$ that computes an accepting run $y$ of $N$ on input $x$ is in $\FNP$.
  Finally notice that $F$ reduces to $F_R$.
  To see this, employ following reduction:
  the function $s$ is simply the identity function; the logspace computable function 
  $r$ extracts the result of $N$ on input $x$ 
  from the accepting run $y$ computed by $F_R$ on input $x$.
  We conclude the lemma since $\FNP$ is closed under reductions.
\end{proof}

\subsection{Term Rewriting}
We assume at least nodding acquaintance with the basics of term rewriting~\cite{BN98}.
We fix notions and notation that are used in the paper.

Throughout the paper, we fix a countably infinite set of \emph{variables} $\VS$
and a finite \emph{signature} $\FS$ of \emph{function symbols}.
The signature $\FS$ is partitioned into \emph{defined symbols} $\DS$ and 
\emph{constructors} $\CS$.
The set of \emph{values}, \emph{basic terms} and
\emph{terms} is defined according to the grammar
\begin{alignat*}{4}
  & \text{\textsf{(Values)}} & & \quad & \Val & \ni v &~&\defsym~x \mid c(\seq{v}) \\
  & \text{\textsf{(Basic Terms)}} & & & \BASICS & \ni s & &  \defsym~ x \mid f(\seq{v}) \\
  & \text{\textsf{(Terms)}} & & & \TERMS & \ni t & & \defsym~ x \mid c(\seq{t}) \mid f(\seq{t})
  \tpkt
\end{alignat*}
where $x \in \VS$, $c \in \CS$, and $f \in \DS$.

The \emph{arity} of a function symbol $f \in \FS$ is denoted by $\ar(f)$.
We write $s \superterm t$ if $t$ is a \emph{subterm} of $s$, the 
strict part of $\superterm$ is denoted by $\supertermstrict$.
The \emph{size} of a term $t$ is denoted by $\size{t}$ and refers
to the number of occurrences of variables and function symbols contained in $t$.
We denote by $\depth(t)$ the \emph{depth} of $t$ which is
defined as $\depth(t) = 1$ if $t \in \VS$ and 
$\depth(f(\seq{t})) = 1 + \max\{0\} \cup \{\depth(t_i) \mid i = 1,\dots n\}$.

Let $\qp$ be a preorder on the signature $\FS$, called \emph{quasi-precedence}
or simply \emph{precedence}. We always write $\sp$ for the 
induced proper order and 
$\ep$ for the induced equivalence on $\FS$.
We lift the equivalence ${\ep}$ to terms 
modulo argument permutation:
$s \eqi t$ if either $s = t$ or
$s = f(\seq{s})$ and $t = g(\seq{t})$ where $f \ep g$
and for some permutation $\pi$,
$s_i \eqi t_{\pi(i)}$ for all $i \in \{1,\dots,n\}$.
Further we write $s \esuperterm t$ if $t$ is a subterm
of $s$ modulo $\eqi$, i.e.,  $s~{\superterm} \cdot {\eqi}~{t}$.
We denote by $\sigbelow{f}{\FS} \defsym \{g \mid f \sp g\}$ the set of function symbols 
below $f$ in the precedence $\qp$.
The \emph{rank} of a function symbol is inductively defined by
$\rk(f) = \max\{0\} \cup \{1 + \rk(g) \mid f \sp g\}$.

A \emph{rewrite rule} is a pair $(l,r)$ of terms, in notation $l \to r$, 
such that $l$ is not a variable and all variables in $r$ occur also in $l$.
Here $l$ is called the \emph{left-hand}, and $r$ the \emph{right-hand side} of $l \to r$.
A \emph{term rewrite system} (\emph{TRS} for short) $\RS$ over
$\TERMS$ is a set of \emph{rewrite rules}.
In the following, $\RS$ always denotes a TRS.\@
If not mentioned otherwise, we assume $\RS$ to be \emph{finite}.
A relation on $\TERMS$ is a \emph{rewrite relation} if it is
closed under contexts and closed under substitutions. 
The smallest rewrite relation that contains $\RS$ is denoted by
$\rew$. 

A term $s \in \TERMS$ is called a \emph{normal form} if there is no
$t \in \TERMS$ such that $s \rew t$. 
With $\NF(\RS)$ we denote the set of all normal forms of a TRS $\RS$;
if the TRS is clear from context we simply write $\NF$.
Whenever $t$ is a normal form of $\RS$ we write $s \rsn t$ for $s \rss t$.
The \emph{innermost rewrite relation}, denoted as $\irew$, is the restriction
of $\rew$ where the arguments of the redex are in normal form.
The TRS $\RS$ is \emph{terminating} if no infinite rewrite sequence exists.
The TRS $\RS$ has \emph{unique normal forms} if for all 
$s, t_1, t_2 \in \TERMS$ with $s \rsn t_1$ and 
$s \rsn t_2$ we have $t_1 = t_2$.
The TRS $\RS$ is called \emph{confluent} if for all $s, t_1, t_2 \in \TERMS$
with $s \rss t_1$ and $s \rss t_2$ there exists a term $u$ such that
$t_1 \rss u$ and $t_2 \rss u$. An \emph{orthogonal} TRS is 
a \emph{left-linear} and \emph{non-overlapping} TRS. 
Here \emph{left-linear} means that no variable occurs more than once in each left-hand side. 
A TRS is \emph{overlapping}, if some pair of rules $l_1 \to r_1$ and $l_2 \to r_2$ in $\RS$, renamed 
so that variables are disjoint, satisfies:
(i)~a subterm $l_1'$ of $l_1$ is unifiable with $l_2$, i.e., $l_1'\sigma = l_2\sigma$ for some substitution $\sigma$; and
(ii)~if $l_1' = l_1$ then the rules $l_1 \to r_1$ and $l_2 \to r_2$ are not equal up to renaming of variables.
Note that every orthogonal TRS is confluent~\cite{BN98}.
A TRS $\RS$ is a \emph{constructor} TRS if all left-hand sides are basic terms.

\subsection{Rewriting as Computational Model}

We fix \emph{call-by-value} semantics and only consider 
\emph{constructor} TRSs $\RS$. Input and output are taken from the set of 
values $\Val$, and defined symbols $f \in \DS$ denote computed functions. 
More precise, a (finite) \emph{computation} of $f\in \DS$ on input $\seq{v} \in \Val$
is given by \emph{innermost} reductions
$$
  f(\seq{v}) = t_0 \irew t_1 \irew \cdots \irew t_\ell = w \tpkt
$$
If the above computation ends in a value, i.e., $w \in \Val$, 
we also say that $f$ \emph{computes} on input $\seq{v}$ in $\ell$ steps the value $w$.

To account for nondeterministic computation, 
we capture the semantics of $\RS$ by assigning to each $n$-ary defined symbol 
$f \in \DS$ an $n+1$-ary relation $\sem{f}$ that relates
input arguments $\seq{v}$ to computed values $w$.
A \emph{finite} set $\NA$ of \emph{non-accepting patterns} 
is used to distinguish meaningful outputs
$w$ from outputs that should not be considered part of the 
computation~\cite{BonfanteMoser:2010}.
A value $w$ \emph{is accepting} with respect to $\NA$ 
if no $p \in \NA$ and no substitution $\sigma$ exists, such that 
$p\sigma = w$ holds. A typical example of a value that should 
not be accepted is the constant $\unsat$ 
appearing in the TRS $\RSsat$ from Example~\ref{ex:rssat}.
Below function problem are extended to $n+1$-ary relations in the obvious way.
\begin{definition}
\label{d:computation}
Let $\RS$ be a TRS and let 
$\NA$ be a set of non-accepting patterns.
For each $n$-ary symbol $f \in \DS$ the 
\emph{relation $\sem{f} \subseteq {\Val^{n+1}}$ defined by $f$ in $\RS$}
is given by
\begin{equation*}
{(\seq{v},w) \in \sem{f}} \quad\defiff\quad {f(\seq{v}) \irsn[\RS] w} \text{ and $w$ is accepting with respect to $\NA$}\tpkt
\end{equation*}
We say that $\RS$ \emph{computes} the function problems associated with $\sem{f}$.
\end{definition}

Note that if $\RS$ is confluent, then $\sem{f}$ is in fact a (partial) function. 
Following~\cite{HM08,AM10b} we adopt an \emph{unitary cost model} for rewriting, 
where each reduction step accounts for one time unit, cf.~\cite{LM09,LM:2009b}.
Reductions are of course measured in the size of the input. 
\begin{definition}
The \emph{innermost runtime complexity function} $\ofdom{\rc[\RS]}{\N \to \N}$
relates sizes of basic terms $f(\seq{v}) \in \BASICS$ to the maximal 
length of computation. Formally
$$
 \rc[\RS](n) \defsym 
 \max\{\ell \mid \exists s \in \BASICS, \size{s} \leqslant n \text{ and } s  = t_0 \irew t_1 \irew \dots \irew t_\ell\} \tpkt
$$
If clear from context, we sometime drop the qualifier 'innermost' and
simply speak of \emph{runtime complexity} (of the TRS $\RS$).
\end{definition}

The restriction $s \in \BASICS$ accounts for the fact that computations start only from basic terms. 
We sometimes use $\dheight[\RS](s) \defsym \max\{\ell \mid \exists t.~s \irsl{\ell} t\}$
to refer to the \emph{(innermost) derivation height} of a single term $s$.
Note that the runtime complexity function is well-defined if $\RS$ is \emph{terminating}, 
i.e., $\irew$ is well-founded. 
Suppose $\rc$ is asymptotically bounded from above by a linear, quadratic,\dots, 
polynomial function, we simply say that the runtime of $\RS$ is linear, quadratic,\dots, 
or respectively polynomial.
If no confusion can arise, we drop the reference to the TRS $\RS$ and
simple write $\Rc$ instead of $\rc[\RS]$.

By folklore it is known that rewriting can be implemented with only polynomial overhead
if terms grow only polynomial during reductions. This implies that the functions 
computed by specific TRSs of polynomial runtime complexity can be implemented
with polynomial time complexity on a Turing machine (or an alternative computation
model). This observation can be significantly extended as it can be shown that 
the restriction on polynomial growth is not necessary. 

%
In~\cite{AM10,AM10b} it is shown that the unitary cost model is reasonable 
for full rewriting. The deterministic case was established independently in~\cite{LM09,LM:2009b} using essentially the same approach.
By Proposition~\ref{p:fnp} and a suitable adaption of Theorem~6.2 in~\cite{AM10b} to innermost rewriting we obtain the following proposition.

\begin{proposition}
\label{p:invariance}
Let $\RS$ be a rewrite system with polynomial runtime.
Then the function problems associated with $\sem{f}$ defined by $\RS$ are contained in $\FNP$. 
If $\RS$ is confluent, i.e., deterministic, then $\sem{f}$ is a (partial) function contained in $\FP$.
\end{proposition}

Our choice of adopting call-by-value semantics rests
on the observation that the unitary cost model of unrestricted rewriting 
often overestimates the runtime complexity of computed functions. 
This has to do with the unnecessary duplication of redexes.
\begin{example}\label{ex:dup}
  Consider the constructor TRS $\RSdup$ given by the following rules:
  \begin{alignat*}{6}
    \rlabel{dup:1} &~& \m{btree}(0) & \to \m{leaf}  
    \qquad & 
    \rlabel{dup:3} &~& \m{dup}(t) & \to \m{node}(t,t)
    \quad &
    \rlabel{dup:2} &~& \m{btree}(\ms(n)) & \to \m{dup}(\m{btree}(n)) \tpkt 
  \end{alignat*}
  Then for $n \in \N$, $\m{btree}(\ms^n(0))$ computes a binary tree of height $n$
  in a linear number of steps.
  On the other hand, $\RSdup$ gives also rise to a non-innermost reduction
  $$
  \m{btree}(\ms^{n}(0)) 
  \rew \m{dup}(\m{btree}(\ms^{n-1}(0)))
  \rew \m{node}(\m{btree}(\ms^{n-1}(0)), \m{btree}(\ms^{n-1}(0)))
  \rew \dots
  \tkom
  $$
  obtained by preferring $\m{dup}$ over $\m{btree}$.
  The length of the derivation is however exponential in $n$.
\end{example}

By Proposition~\ref{p:invariance} we obtain $\sem{\m{btree}} \in \FP$.
As indicated later, our analysis can automatically classify the function $\sem{\m{btree}}$ as feasible.


\section{The Polynomial Path Order}\label{s:popstar}

We arrive at the formal definition of the \emph{polynomial path order}
(\emph{\POPSTAR} for short). Variants of the definition presented here have been 
presented in earlier conference publications~\cite{AM08,AMS08,AM09,AM09b}. 

As already mentioned in the introduction 
the multiset path order capture the primitive recursive functions.
Hence reduction orders can entail implicit characterisations of
complexity classes. This motivates the quest for a miniaturisation of \MPO\ that
precisely captures the class $\FP$. 
Another motivation for the design of \POPSTAR\ rests on
the observation that term-rewriting characterisations of complexity classes 
may facilitate the study of (low) complexity classes, cf.~\cite{BW96,CW97}. 
Such applications imply the need to craft an order that
induces polynomial innermost runtime complexity.
\POPSTAR\ meets these demands, by providing a syntactic account of the
\emph{predicative analysis} of recursion set forth by Bellantoni and Cook~\cite{BC92}. 
Analogously \POPSTAR\ can be conceived as syntactic account of Leivant's 
notion of \emph{tiered recurrence}~\cite{Leivant:1990,Leivant93}.

For each function $f$, the arguments to $f$ are separated into \emph{normal} 
and \emph{safe} ones.
To highlight this separation, we write $f(\svec{x}{y})$ where arguments $\vec{x}$ to the left 
of the semicolon are normal, the arguments $\vec{y}$ to the right are safe.
Bellantoni and Cook define a class $\B$,
consisting of a small set of initial functions 
and that is closed under \emph{safe composition} and \emph{safe recursion on notation} (\emph{safe recursion} for brevity).

The crucial ingredient in $\B$ is that 
a new function $f$ is defined via safe recursion by the equations:
\begin{equation}
\label{scheme:srn} 
\tag{\ensuremath{\mathrm{SRN}}}
\begin{aligned}
  f(\sn{0,\vec{x}}{\vec{y}}) & = g(\sn{\vec{x}}{\vec{y}}) 
     \\
  f(\sn{2z + i,\vec{x}}{\vec{y}}) & = 
   h_i(\sn{z,\vec{x}}{\vec{y},f(\sn{z,\vec{x}}{\vec{y}})})\quad i \in \set{1,2}
   \tkom
\end{aligned}
\end{equation}
for functions $g,h_1$ and $h_2$ already defined in $\B$.
This definition corresponds to a recursion on the binary representation
of numbers. Consequently the number of recursive calls is linear in the size of the recursive input. 
Unlike primitive recursive functions, the
stepping functions $h_i$ cannot perform recursion on the \emph{impredicative} value $f(\sn{z,\vec{x}}{\vec{y}})$.
This is a consequence of \emph{data tiering}: recursion is performed on \emph{normal} arguments only. 
Dual, the recursive call $f(\sn{z,\vec{x}}{\vec{y}})$ is substituted into a \emph{safe} argument position.

To maintain the separation, safe composition restricts the usual composition operator
so that safe arguments are not substituted into normal argument position.
Precisely, for functions $h$, $\vec{r}$ and $\vec{s}$ already defined in $\B$, 
a function $f$ is defined by safe composition using the equation
\begin{equation}\label{scheme:sc} \tag{\ensuremath{\mathrm{SC}}}
  f(\sn{\vec{x}}{\vec{y}}) = h(\sn{\vec{r}(\sn{\vec{x}}{})}{\vec{s}(\sn{\vec{x}}{\vec{y}})}) \tpkt
\end{equation}
Crucially, the safe arguments $\vec{y}$ are absent in normal arguments to $h$.
The main result from~\cite{BC92} states that the class $\B$ coincides with
the class of polytime computable functions~$\FP$. 

As a first step to capture the notion of predicative recursion in~$\POPSTAR$, 
we introduce the concept of \emph{safe mappings} to
handle the separation of argument positions.
\begin{definition}\label{d:safemapping}
  A \emph{safe mapping} $\safe$ is a function 
  $\ofdom{\safe}{\FS \to 2^\N}$ that associates 
  with every $n$-ary function symbol $f$ the set of \emph{safe argument positions} 
  $\{i_1, \dots , i_m\} \subseteq \{1,\dots,n\}$.
  Argument positions included in $\safe(f)$ are called \emph{safe},
  those not included are 
  called \emph{normal} and collected in $\normal(f)$.
  For $n$-ary constructors $c$ 
  we require that all argument positions are safe, 
  i.e., $\safe(c) = \{1,\dots,n\}$.
\end{definition}

We refine term equivalence so that the safe mapping is taken into account. 

\begin{definition}\label{d:eqis}
  Let ${\qp}$ denote a precedence and $\safe$ a safe mapping.
  We define \emph{safe equivalence} $\eqis$ for terms $s,t \in \TERMS$
  inductively as follows:
  $s \eqis t$ if either $s = t$ or
  $s = f(\seq{s})$, $t = g(\seq{t})$, $f \ep g$
  and there exists a permutation $\pi$ such that for all $i \in \{1,\dots,n\}$, 
  $s_i \eqis t_{\pi(i)}$ and $i \in \safe(f)$ if and only if $\pi(i) \in \safe(g)$.
\end{definition}

To avoid notational overhead, we fix a safe mapping $\safe$ and suppose that for each $k+l$ ary function symbol 
$f$, the first $k$ argument positions are normal, and the remaining 
argument positions are safe, i.e., $\safe(f) = \{k+1,\dots,k+l\}$.
This allows use to write $f(\pseq{s})$ as before.
We require that the precedence $\qp$ adheres the partitioning of 
$\FS$ into defined symbols and constructors in the following sense.

\begin{definition}
A precedence $\qp$ is \emph{admissible} (for \POPSTAR) if $f \ep g$ implies 
that either both $f$ and $g$ are defined symbols, or both are constructors.
\end{definition}

In particular $\eqis$ preserves values, i.e., 
if $s \in \Val$ and $s \eqis t$ then also $t \in \Val$. 
The following definition introduces an auxiliary order $\gsq$, 
the full order $\gpop$ is then presented in Definition~\ref{d:gpop}.

\begin{definition}\label{d:gsq}
  Let ${\qp}$ denote a precedence.  
  Consider terms $s, t \in \TERMS$ such that $s = f(\pseq[k][l]{s})$.
  Then $s \gsq t$ if one of the following alternatives holds:
  \begin{enumerate}[labelsep=*,leftmargin=*]
  \item\label{d:gsq:st} $s_i \geqsq t$ for some $i \in \{1,\dots,k+l\}$ and, 
    if $f \in \DS$ then $i$ is a normal argument position ($i \in \{1,\dots,k\}$);
  \item\label{d:gsq:ia} $f \in \DS$, $t = g(\pseq[m][n]{t})$ where $f \sp g$ 
    and $s \gsq t_i$ for all $i = 1,\dots,m+n$.
  \end{enumerate}
  Here we set ${\geqsq} \defsym {\gsq} \cup {\eqis}$.
\end{definition}

Consider a function $f$ defined by safe composition from $r$ and $s$, 
cf.~scheme~\eqref{scheme:sc}.
The effect of this auxiliary order is to (properly) encompass 
safe composition in the full order $\gpop$.
Note that the auxiliary order can orient $f(\sn{\vec{x}}{\vec{y}}) \gsq r(\sn{\vec{x}}{})$ for defined symbols $f$ and $r$
with $f \sp r$. 
On the other hand, $f(\sn{\vec{x}}{\vec{y}})$ and safe arguments $y_i$ are incomparable,
and thus the orientation of $f(\sn{\vec{x}}{\vec{y}})$ and $s(\sn{\vec{x}}{\vec{y}})$ fails, even if $f \sp s$ is supposed.

\begin{definition}\label{d:gpop}
  Let ${\qp}$ denote a precedence.  
  Consider terms $s, t \in \TERMS$ such that $s = f(\pseq[k][l]{s})$.
  Then $s \gpop t$ if one of the following alternatives holds:
  \begin{enumerate}[labelsep=*,leftmargin=*]
  \item\label{d:gpop:st} $s_i \geqpop t$ for some $i \in \{1,\dots,k+l\}$, or
  \item\label{d:gpop:ia} $f \in \DS$, $t = g(\pseq[m][n]{t})$ where $f \sp g$ 
    and the following conditions hold:
    \begin{itemize}
    \item $s \gsq t_j$ for all normal argument positions $j = 1,\dots,m$;
    \item $s \gpop t_j$ for all safe argument positions $j = m+1,\dots,m+n$;
    \item $t_j \not\in \TA(\sigbelow{f}{\FS},\VS)$ for at most one safe argument position $j \in \{m+1,\dots,m+n\}$;
    \end{itemize}
  \item\label{d:gpop:ep} $f \in \DS$, $t = g(\pseq[m][n]{t})$ where $f \ep g$
    and the following conditions hold:
    \begin{itemize}
    \item $\mset{\seq[k]{s}} \gpopmul \mset{\seq[m]{t}}$;
    \item $\mset{\seq[k+l][k+1]{s}} \geqpopmul \mset{\seq[m+n][m+1]{t}}$.
    \end{itemize}
  \end{enumerate}
  Here ${\geqpop} \defsym {\gpop \cup \eqis}$.
\end{definition}

We use the notation $\caseref{\gsq}{i}$ and respectively $\caseref{\gpop}{i}$ 
to refer to the \nth{$i$} case in Definition~\ref{d:gsq} respectively 
Definition~\ref{d:gpop}.
Note that \POPSTAR\ is not a reduction order, as for example
closure under contexts fails due to the conditions put
upon $\cpop{ep}$. However \POPSTAR\ is a restriction of \MPO\ and
thus a \emph{termination order}.

\begin{remark}
The restrictions put upon \cpop{ia} amount to the fact that \POPSTAR\
allows at most one recursive call per right-hand side.  
\end{remark}

\begin{remark}
The proposed constraints are weaker compared to the corresponding clause given 
in~\cite[Definition 4]{AM08}. 
The early definition from~\cite[Definition~4]{AM08}, 
used the full order $\gpop$ only on one argument of the right-hand side 
(the one that possibly holds the recursive call), 
the remaining arguments were all oriented with the auxiliary order $\gsq$. 
\end{remark}

The case \cpop{ia} accounts for definitions by safe composition \eqref{scheme:sc}.
The final restriction put onto \cpop{ia} is used to prevent multiple recursive calls
as indicated in Example~\ref{ex:RS2}.
The case \cpop{ep} restricts the corresponding case in \MPO\ 
by taking the separation of normal and safe argument positions into account. 
Note that here normal arguments need to decrease. 
This reflects that as in \eqref{scheme:srn} 
recursion is performed on normal argument positions.
We emphasise that one application of \cpop{ia} possibly followed by one application of \cpop{ep}
orients the defining equations in \eqref{scheme:srn}, using
the obvious precedence (using suitable term representation of natural numbers).

Our order-theoretic account of predicative recursion motivates 
following definition.
\begin{definition}
  We call a constructor TRS $\RS$ \emph{predicative recursive}, if
  $\RS$ is compatible with an instance $\gpop$ of $\POPSTAR$ based on an 
  admissible precedence.
\end{definition}

Note that it can be determined in nondeterministic polynomial
time that a constructor TRS is predicative recursive: simply
guess a safe mapping and a precedence and apply the definition
of \POPSTAR. This is in contrast to semantic method 
(like additive polynomial interpretations) whose synthesis is 
undecidable in general.

We clarify Definition~\ref{d:gpop} on several examples. The first
of these examples shows a typical application of predicative recursion.
\begin{example}\label{ex:mult1}
  Consider the TRS $\RSmul$ expressing multiplication in Peano arithmetic.
  \begin{alignat*}{4}
    \rlabel{plus1} &~& +(\sn{0}{y}) & \to y 
    & \qquad
    \rlabel{plus2} &~& +(\sn{\ms(\sn{}{x})}{y}) & \to \ms(\sn{}{+(\sn{x}{y})})
    \\
    \rlabel{times1} &~& \times(\sn{0,y}{}) & \to 0 
    & \qquad
    \rlabel{times2} &~& \times(\sn{\ms(\sn{}{x}),y}{}) & \to +(\sn{y}{\times(\sn{x,y}{})})
  \end{alignat*}
  The TRS $\RSmul$ is predicative recursive, using the precedence
  ${\times}~\sp~{+}~\sp~{\ms}$ and the safe mapping as indicated in the rules:
  The rules \rref{plus1} and \rref{times1} are oriented by \cpop{st}.
  The rule \rref{plus2} is oriented by $\cpop{ia}$ using ${+}~\sp~{\ms}$
  and $+(\sn{\ms(\sn{}{x})}{y}) \cpop{ep} +(\sn{x}{y})$. 
  Note that the latter inequality only holds as the first argument position of addition is normal.
  Similar, the final rule \rref{times2} is oriented by \cpop{ia}, employing
  ${\times}\sp{+}$ together with $\times(\sn{\ms(\sn{}{x}),y}{}) \csq{st} y$ 
  and $\times(\sn{\ms(\sn{}{x}),y}{}) \cpop{ep} \times(\sn{x,y}{})$.
  Note that the latter two inequalities require that 
  the both argument positions of $\times$ are normal, i.e.,\ are 
  used for recursion.
\end{example}  

We re-consider the motivation Example~\ref{ex:rssat} from the introduction.
\begin{example}[Example~\ref{ex:rssat} continued]
\label{ex:rssat:2}
Consider the TRS~$\RSsat$ from the motivating Example~\ref{ex:rssat}, where
the separation into normal and safe arguments for the
defined function symbols is defined as follows: 
$\normal(\mif) = \normal(\mneg) = \varnothing$, $\normal(\meq) = \{1\}$.
We set $\normal(\member) = \{2\}$, 
and for all remaining defined function symbols, we
make all arguments normal, i.e.,
$\normal(\verify) = \normal(\issat) = \normal(\issat')
= \normal(\guess) = \normal(\choice) = \{1\}$.

Then ${\RSsat} \subseteq {\gpop}$ for any admissible precedence
satisfying the following constraints: 
$\guess \sp \choice$, $\verify \sp \mif, \member, \mneg$,
$\issat \sp \issat', \guess$, and 
$\issat' \sp \mif, \verify, \unsat$.
\end{example}

The next example is negative, in the sense that the considered TRSs admits
polynomial runtime complexity, but fails to be compatible with $\POPSTAR$.
\begin{example}[Example~\ref{ex:mult1} continued]
  Consider the TRS $\RSmul$ where the rule \rref{times2} is replaced by the rule
  $$
  \rlabel[\therule{times2}a]{times2a}~\times(\sn{\ms(\sn{}{x}),y}{}) \to +(\sn{\times(\sn{x,y}{})}{y})
  \tpkt
  $$
  The resulting system has polynomial runtime complexity, which can
  be automatically verified with~\TCT. However, the TRS does not follow the rigid 
  scheme of predicative recursion. For this reason, it cannot be handled by \POPSTAR.\@
  Technically, the terms $\times(\sn{\ms(\sn{}{x}),y}{})$ and $\times(\sn{x,y}{})$ 
  are incomparable with respect to $\gsq$ independent on the precedence, and 
  consequently also orientation of left- and right-hand side with 
  \cpop{ia} fails. 
\end{example}

The following examples clarifies the need for data tiering.
\begin{example}[Example~\ref{ex:mult1} continued]
  Consider the extension of $\RSmul$ by the two rules
  \begin{alignat*}{4}
    \rlabel{exp1} &~& \m{exp}(0,y) & \to \ms(\sn{}{0})  \qquad\qquad &
    \rlabel{exp2} &~& \m{exp}(\ms(\sn{}{x}),y) & \to \times(\sn{y,\m{exp}(x,y)}{}) \tkom
  \end{alignat*}
  that express exponentiation $y^x$ in an exponential number of steps. 
  The definition of $\m{exp}$ disregards data tiering as imposed by predicative recursion.
  In particular, since $\times$ admits no safe argument positions it
  cannot serve as a stepping function. Independent on the safe mapping for $\m{exp}$, 
  rule \rref{exp2} cannot be oriented using polynomial path orders.
\end{example}

The following theorem constitutes the first main result of this paper.
\begin{theorem}\label{t:popstar}
  Let $\RS$ be a predicative recursive (constructor) TRS.\@
  Then the innermost derivation height of any basic term 
  $f(\svec{u}{v})$ is bounded by a polynomial in the 
  maximal depth of normal arguments $\vec{u}$.
  The polynomial depends only on $\RS$ and the signature $\FS$.
  In particular, the innermost runtime complexity of $\RS$ is polynomial.
\end{theorem}

\label{popstar:proofplan}
The proof of Theorem~\ref{t:popstar} is involved and requires
a variety of ingredients. We give a short outline.
In Section~\ref{s:pop}, we define \emph{predicative interpretations} $\ints$
that flatten terms to \emph{sequences of terms}, essentially separating 
safe from normal arguments. This allows us to analyse terms independent from safe arguments.
Then we introduce an order $\gpopv[][]$ on sequences of terms, 
that is simpler compared to $\gpop$ and does not rely on the separation
of argument positions.  This \emph{polynomial path order on
sequences} extends the polynomial path order for $\FP$, introduced in~\cite{AM05}.
In Section~\ref{s:embed} we establish a \emph{predicative embedding}
of derivations into $\gpopv[][]$ as depicted in Figure~\ref{fig:2}.

\begin{figure}
\begin{center}
  \begin{tikzpicture}
    \newcommand{\n}{2};
    \newcommand{\dx}{1.8};
    \newcommand{\dy}{0.9};
    \newcommand{\drawterms}[2]{
      \pgfmathparse{#1}
      \node(s_#1) at (\pgfmathresult*\dx,0) {$s_{#2}$};
      \node(is_#1) at (\pgfmathresult*\dx,-\dy) {$\ints(s_{#2})$};
      \draw[->] (s_#1) to node[right] {} (is_#1); 
    }
    \newcommand{\drawrel}[1]{
      \pgfmathparse{#1}
      \node(rew_#1) at (\pgfmathresult*\dx+\dx/2,0) {$\irew$};
      \node(ord_#1) at (\pgfmathresult*\dx+\dx/2,-\dy) {$\gpopv[][]$};
    }

    \foreach \x in {1,...,\n} \drawterms{\x}{\x};

    \pgfmathparse{(\n+1)*\dx};
    \node at (\pgfmathresult,0) {$\dots$};
    \node at (\pgfmathresult,-\dy) {$\dots$};
    \drawterms{\n+2}{\ell};

    \foreach \x in {1,...,\n} \drawrel{\x};
    \drawrel{\n+1};
  \end{tikzpicture}
\end{center}
\caption{Predicative Embedding}
\label{fig:2}
\end{figure}

In Theorem~\ref{t:pop} we show that the length of $\gpopv[][]$ descending sequences
starting from basic terms can be bound appropriately.
One may wonder whether a precise degree on the provided polynomial bound can be obtained
by reasoning based on the depth of recursion and formation of composition rules.
This is not the case as Lemma~\ref{l:imprecise} clarifies, where 
we provide a family of TRS $(\RS_k)_{k \geqslant 1}$, 
all with depth of recursion one and without the use of composition, such that ${\rc[\RS_k](n)} = \Omega(n^k)$.
However, it is possible to design a restriction of \POPSTAR, dubbed
\emph{small polynomial path order} in~\cite{AEM12}, which induces 
a precise degree on the provided polynomial bound.

\begin{lemma}\label{l:imprecise}
  For every $k \geqslant 1$ there exists a TRS $\RS_k$ 
  over constructors $\ms,0$ defining a single defined symbol $\m{f}_k$
  such that $\RS_k \subseteq {\gpop}$ 
  and ${\rc[\RS_k](n)} = \Omega(n^k)$. 
\end{lemma}
\begin{proof}
  Consider the following TRS $\RS_k$ that is compatible with $\gpop$:
  \begin{align*}
    \m{f}_k(\sn{\ms(x_1), x_2,x_3, \dots,x_k}{}) & \to \m{f}_k(\sn{x_1, x_2, x_3, \dots,x_k}{}) \\
    \m{f}_k(\sn{0,\ms(x_2), x_3, \dots, x_k}{}) & \to \m{f}_k(\sn{x_2, x_2, x_3, \dots,x_k}{}) \\
    & \ \,\vdots \\
    \m{f}_k(\sn{0,\dots,0,\ms(x_k)}{}) & \to \m{f}_k(\sn{x_k,\dots,x_k,x_k}{})
  \end{align*}
  We show the stronger claim that for all $n \geqslant 1$ there exist constants $c_k \in \Rat$ such that 
  $$
  \m{f}_k(\ms^n(0), \dots, \ms^n(0)) 
  \rsl[\RS_k]{{\geqslant}c_k \cdot n^k} \m{f}_k(0, \dots, 0)
  $$
  by induction on $k$.
  The base case is trivial, we consider the inductive step. 
  Applying the induction hypothesis, by definition of $\RS_k$ and
  $\RS_{k+1}$ it is easy to see that for all $n \geqslant 1$
  \begin{align*}
  \m{f}_{k+1}(\ms^{n}(0), \dots, \ms^{n}(0), \ms^n(0)) 
  & \rsl[\RS_{k+1}]{{\geqslant}c_k \cdot n^k} \m{f}_{k+1}(0, \dots, 0,\ms^{n}(0)) \\
  & \rew[\RS_{k+1}]\m{f}_{k+1}(\ms^{n-1}(0),\dots,\ms^{n-1}(0),\ms^{n-1}(0)) \tpkt
  \end{align*}
  Using this observation, a simple side induction on $n$ reveals
  $$
  \m{f}_{k+1}(\ms^{n}(0), \dots, \ms^{n}(0), \ms^n(0)) 
  \rsl[\RS_{k+1}]{{\geqslant}f(n)} \m{f}_{k+1}(0, \dots, 0, 0) 
  $$
  where $f(n) = \sum_{i=1}^n c_k \cdot i^k \in \Omega(n^{k+1})$. 
  The exact overhead due to multiset comparisons is further 
  investigated in Example~\ref{ex:multisets}.
\end{proof}

The next three examples stress that the restrictions to 
runtime complexity (Example~\ref{ex:dc}),
innermost reductions (Example~\ref{ex:outermost}), 
as well as constructor TRSs (Example~\ref{ex:constructor}), 
are all essential for the correctness of Theorem~\ref{t:popstar}. 

\begin{example}
\label{ex:dc}
Consider the TRS $\RSdc$ given by the following rules, 
cf.~\cite[Example~1]{HL89}.
\begin{alignat*}{4}
  \rlabel{dc:1} && +(0;y) & \to y 
  & \qquad
  \rlabel{dc:2} && +(\ms(;x);y) & \to \ms(;+(x;y))  
  \\
  \rlabel{dc:3} && \m{d}(0;) & \to 0 
  & \qquad
  \rlabel{dc:4} && \m{d}(\ms(;x);) & \to \ms(;\ms(;\m{d}(x;)))
  \\
  \rlabel{dc:5} && \m{q}(0;) & \to 0 
  & \qquad
  \rlabel{dc:6} && \m{q}(\ms(;x);) & \to +(\ms(;\m{d}(x;));\m{q}(x;))
  \tkom
\end{alignat*}
where we suppose that $0$ and $\ms$ are the only constructors. 
As shown by Hofbauer and Lautemann, $\RSdc$ admits at least 
double-exponentially \emph{derivational complexity}; in particular
it is easy to find a family of terms $t_n$ such that 
$\dheight(t_n) = 2^{2^{\Omega(n)}}$. 

On the other hand $\RSdc$ is compatible with a polynomial path order $\gpop$
as induced by a precedence $\qp$ satisfying $\m{q} \sp \m{d} \sp \ms \sp 0$
and the safe mapping as indicated in the rules. 
As a side-remark we emphasise that the orientability of rule~\rref{dc:6}
induces that \POPSTAR\ properly extends the safe composition scheme~\eqref{scheme:sc}.
\end{example}

\begin{example}[Example~\ref{ex:dup} continued]
\label{ex:outermost}  
  Observe that $\RSdup \subseteq {\gpop}$ with any admissible precedence 
  satisfying $\m{btree} \sp \m{dup}$ for the TRS $\RSdup$ depicted in Example~\ref{ex:dup}. 
  We use $\normal(\m{btree}) = \{1\}$ and $\normal(\m{dup}) = \varnothing$ for
  the defined function symbols.
  Theorem~\ref{t:popstar} thus implies that the (innermost) runtime complexity of $\RSdup$ is 
  polynomial. On the other hand, we already observed that $\RSdup$ admits exponentially long 
  \emph{outermost} reductions.
\end{example}

\begin{example}
\label{ex:constructor}
  Consider the TRS $\RSnc$ given by the rules 
  \begin{alignat*}{4}
    \rlabel{nc:f} && \mf(\sn{n}{}) & \to \mh(\sn{}{\m{gs}(\sn{n}{})})
    & \qquad
    \rlabel{nc:gs1}  && \m{gs}(\sn{0}{}) & \to 0
    \\
    \rlabel{nc:h} && \mh(\sn{}{\mg(\sn{}{n})}) & \to \mc(\sn{}{\mh(\sn{}{n}),\mh(\sn{}{n})})
    & \qquad
    \rlabel{nc:gs2}  && \m{gs}(\sn{\m{s}(\sn{}{n})}{}) & \to \mg(\sn{}{\m{gs}(\sn{n}{})}{})
    \\
    \rlabel{nc:g}  && \mg(\sn{}{\bot}) & \to \mc(\sn{}{\mh({\sn{}{\bot}}),\mh({\sn{}{\bot}})})
    \tkom
  \end{alignat*}
  where we suppose that the only constructors are $\bot$, $0$ and $\ms$. 
  The rule~\rref{nc:g} is used to define the symbol $\mg$, and to properly set up the precedence. 
  The rules~\rref{nc:gs1} and~\rref{nc:gs2} are used to translate a tower $\ms^n(\sn{}{0})$ to $\mg^n(\sn{}{0})$, 
  using rule~\rref{nc:f} we thus obtain a family of reductions 
  $$
  \mf(\sn{\ms^n(\sn{}{0})}{}) 
  \rew[\RSnc] \mh(\sn{}{\m{gs}(\sn{\ms^n(\sn{}{0})}{})}) 
  \rss[\RSnc] \mh(\sn{}{\mg^n(\sn{}{0})}{}) \tkom
  $$
  for $n \in \N$. 
  It is not difficult to see that the derivation height of 
  the final term $\mh(\sn{}{\mg^n(\sn{}{0})})$ grows exponentially in $n$ due to rule \rref{nc:h}, 
  and so the innermost runtime complexity of $\RSnc$ is bounded by an exponential from below. 

  On the other hand, $\RSnc$ is compatible with a polynomial path order $\gpop$
  as induced by a precedence $\qp$ satisfying 
  $\mf\  \sp\ \m{gs}\  \sp\  \mg\  \sp\  \mh\  \sp\  \m{c}$
  and the safe mapping as indicated in the rules.
  Observe that for rule~\rref{nc:h} we exploit that $\mg$ is defined, 
  conclusively one can show
  $\mg(\sn{}{m}) \gpop \mc(\sn{}{\mh(\sn{}{m}),\mh(\sn{}{m})})$
  and therefore $\mh(\sn{}{\mg(\sn{}{m})}) \cpop{st} \mc(\sn{}{\mh(\sn{}{m}),\mh(\sn{}{m})})$ 
  holds. 
  However, due to rule~\rref{nc:g} $\RSnc$ is not a constructor TRS, as demanded by 
  Theorem~\ref{t:popstar}.
\end{example}

Bellantoni and Cook's characterisation, as well as Leivant's work on tiered recurrence,
was originally stated on word algebras. In quite recent work 
by Dal~Lago~et~al.,~\cite{DLMZ:10} the result by Leivant~\cite{Leivant:1990,Leivant93} has been extended to arbitrary free algebras. 
In particular, in~\cite{DLMZ:10} a function $f$ is defined by \emph{general ramified recursion} as 
$$
  f(c(\seq{x}), \vec{y}) = h_c(\seq{x},\vec{y},f(x_1,\vec{y}),\dots,f(x_n,\vec{y})) \tkom
$$
for every $n$-ary constructor $c$, provided a data tiering principle is satisfied.
Due to the linearity condition imposed on $\cpop{ia}$, such a recursion principle 
cannot be expressed in predicative recursive TRSs. Indeed, allowing this form 
of recursion would invalidate Theorem~\ref{t:popstar}. 
\begin{example}[Example~\ref{ex:dup} continued]
  Let $\RSgr$ denote the extension of $\RSdup$ from Example~\ref{ex:dup} by the rules
  \begin{alignat*}{4}
    \rlabel{gr:call}  && \mf(\sn{n}{}) & \to \m{traverse}(\sn{\m{btree}(\sn{n}{})}{})
    \\
    \rlabel{gr:g1}  && \m{traverse}(\sn{\m{leaf}}{}) & \to \m{leaf}
    \\
    \rlabel{nc:g2}  && \m{traverse}(\sn{\m{node}(\sn{}{x,y})}{}) & \to \m{node}(\sn{}{\m{traverse}(\sn{x}{}), \m{traverse}(\sn{y}{})})
    \tpkt
  \end{alignat*}
  The above definition of $\m{f}$ is expressible in the system of~\cite{DLMZ:10}, 
  by simply conceiving the rewrite rules as defining equations.
  In particular it follows that $\sem{f}$ is polytime computable. 
  On the other hand, the runtime complexity of $\RSgr$ is at least exponential
  as a derivation of $\m{f}(\ms^{n}(0))$, $\RSgr$ traverses every node of a binary tree of height $n$. 
\end{example}

Finally we note that the order $\gpop$ is \emph{blind} on constructors, in particular $\gpop$ collapses 
to the subterm relation (modulo equivalence) on values. 
\begin{lemma}\label{l:gpop:val}
  Suppose the precedence underlying $\gpop$ is admissible.
  If $s \gpop t$ and $s \in \Val$ then $t$ is a safe subterm of $s$ 
  (modulo $\eqis$), in particular $t \in \Val$ holds.
\end{lemma}


\section{The Polynomial Path Order on Sequences}\label{s:pop}

Fix again a safe mapping $\safe$ on the signature $\FS$.
We now define the notion of \emph{predicative interpretation} of a term $t$.
Guided by the safe mapping, predicative interpretations map terms to sequences of terms.
We then introduce the \emph{polynomial path order on sequences}, 
intended to orient images of the predicative interpretation as outlined before.

To formalise sequences, we use an auxiliary variadic function symbol
$\listsym$.  Here variadic means that the arity of
$\listsym$ is finite but arbitrary.  We always write $\lseq{t}$
for $\listsym(\seq{t})$, in particular if we write $f(\seq{t})$ then $f \not=\listsym$.
Define the \emph{normalised signature} $\FSn$ as follows:
\begin{equation*}
  \FSn \defsym \bigl\{ \fn \mid f \in \FS, \normal(f) = \{i_1,\dots,i_k\} \text{ and } \ar(\fn) = k \}\bigr\}
\tpkt
\end{equation*}
The predicative interpretation of a term $f(\pseq{s})$ results in a sequence
$\lst{\fn(\seq[k]{a})} \append a_{k+1} \append \cdots \append a_{k+l}$, 
where $\append$ denotes \emph{concatenation} of sequences and the 
sequences $a_i$ are predicative interpretations of the corresponding arguments $s_i$ ($i = 1,\dots,k+l$).
Note that in the interpretations, terms have sequences as arguments. 

\begin{definition}\label{d:sequencedterms}
  The set of \emph{terms with sequence arguments}
  $\TS \subseteq \TA(\FSn \uplus \set{\listsym},\VS)$
  and the set of \emph{sequences} $\LS  \subseteq \TA(\FSn \uplus \set{\listsym},\VS)$ is inductively defined as follows:
  \begin{enumerate}[labelsep=*,leftmargin=*]
  \item $\VS \subseteq \TS$, and
  \item if $\seq{a} \in \LS$ and $f \in \FSn$ then $f(\seq{a}) \in \TS$, and
  \item if $\seq{t} \in \TS$ then $\lseq{t} \in \LS$.
  \end{enumerate}
\end{definition}

We always write $a,b, \dots$, possibly extended by
subscripts, for elements from $\TS$ and $\LS$.  
The restriction of $\TS$ and $\LS$ to ground terms is denoted  
by $\GTS$ and $\GLS$ respectively.
When no confusion can arise from this we call terms with sequence arguments simply terms.
Further, we sometimes abuse set notation and write $b \in \lseq{a}$ if 
$b = a_i$ for some $i \in \{1,\dots,n\}$.
We denote by $a \append b$ the \emph{concatenation} of $a \in \TLS$ and $b \in \TLS$.
To avoid notational overhead we overload concatenation to both terms and sequences.
Let $\tolst(a) \defsym \lst{a}$ if $a \in \TS$ and $\tolst(a) \defsym a$ if $a \in \LS$.
We set $a \append b \defsym \lst{a_1~\cdots~a_n~b_1~\cdots~b_m}$ 
where $\tolst(a) = \lseq{a}$ and $\tolst(b) = \lseq[m]{b}$.
As concatenation is associative we drop parenthesis at will.
We define the \emph{length} over $\TLS$ as $\len(a) \defsym n$ where $\tolst(a) = \lseq{a}$.
The \emph{sequence width} $\width$ (or \emph{width} for short) 
of an element $a \in \TLS$ is given recursively by
\begin{equation*}
  \width(a) \defsym
  \begin{cases}
    1 & \text{if $a$ is a variable, }\\
    \max \{1,\width(a_1),\dots,\width(a_n)\} & \text{if $a = f(\seq{a})$ with $f \in \FSn$, and}\\
    \sum_{i=1}^n \width(a_i)
    & \text{if $a = \lseq{a}$}
    \tpkt
  \end{cases}
\end{equation*}
In the following we tacitly employ $\len(a) \leqslant \width(a)$ and $\width(a \append b) = \width(a) + \width(b)$ for all $a,b \in \TLS$.
We define the \emph{norm} of $t \in \TERMS$ in correspondence to the depth of $t$, but 
disregard normal argument positions.
\begin{equation*}
  \norm{t} \defsym
  \begin{cases}
    1 & \text{$t$ is a variable,} \\
    1+ \max\{0\} \cup \{\norm{t_{j}} \mid j = k+1,\dots,k+l\} & \text{$t=f(\pseq[k][l]{t})$}.
  \end{cases}
  \tpkt
\end{equation*}
Note that since all argument positions of constructors are safe,\label{d:normonval}
the norm $\norm{\cdot}$ and depth $\depth(\cdot)$ coincide on values.

Predicative interpretations are given by two mappings $\ints$ and $\intn$:
the interpretation $\ints$ is applied on safe arguments and removes normal forms; 
the interpretation $\intn$ is applied to normal arguments and additionally encodes
the norm of the given term as tally sequence. The latter allows us to  
track the maximal depth of normal forms erased by $\ints$.
Let $\theconst \not \in \FSn$ be a fresh constant.
To encode natural numbers $n\in \N$, 
define its \emph{tally sequence representation} $\natToSeq{n}$ 
as the sequence containing $n$ occurrences of this fresh constant:
$\natToSeq{0} = \nil$ and $\natToSeq{n+1} = \theconst \append \natToSeq{n}$.

\begin{definition}\label{d:pi}
A \emph{predicative interpretation} for a TRS $\RS$
is a pair $(\ints_\RS,\intn_\RS)$ of mappings $\ofdom{\ints_\RS,\intn_\RS}{\TERMS \to \TAL^{\!\ast}\!(\FS \cup \{\theconst\})}$
defined as follows:
\begin{align*}
  \ints_\RS(t) & \defsym
  \begin{cases}
    \nil & \text{ if $t \in \NF(\RS)$,} \\
    \lst{\fn(\intn_\RS(t_1), \dots, \intn_\RS(t_k))} \append \ints(t_{k+1}) \append \cdots \append \ints(t_{k+l}) & \text{ otherwise where ($\star$),}
  \end{cases}\\
  \intn_\RS(t) & \defsym \ints_\RS(t) \append \NM{t} \tpkt
\end{align*}
Here ($\star$) stands for $t = f(\pseq{t})$.

As the rewrite system $\RS$ is usually clear from the context, we drop the references 
to $\RS$ when unambiguous.
\end{definition}

Below we introduce the order $\gpopv[][]$ on sequences $\GTLS$. 
In the next section we then embed innermost $\RS$-steps into this order, 
and use $\gpopv[][]$ to estimate the length of reductions accordingly.
For basic terms $s = f(\pseq{u})$ we obtain
$$
\ints(s) = \lst{\fn(\intn(u_1),\dots,\intn(u_k))} \append \ints(u_{k+1}) \append \cdots \append \ints(u_{k+l}) 
= \lst{\fn(\natToSeq{\depth(u_1)},\dots,\natToSeq{\depth(u_k)})}
\tpkt
$$
Hence the obtained bound depends on depths of normal arguments only. 
To get the reader prepared for the definition of $\gpopv[][]$, 
we exemplify Definition~\ref{d:pi} on a predicative recursive TRS.

\begin{example}\label{ex:pint}
  Consider following TRS $\RS_f$, where we suppose that besides $\m{f}$, 
  also $\m{g}$ and $\m{h}$ are defined symbols:
  \begin{align*}
    \rlabel{RSf:1} \m{f}(\sn{0}{y}) & \to y 
    & \rlabel{RSf:2} \m{f}(\sn{\ms(x)}{y}) & \to \m{g}(\sn{\m{h}(\sn{x}{})}{\m{f}(\sn{x}{y})})
  \end{align*}
  Consider a substitution $\ofdom{\sigma}{\VS \to \NF}$.
  Using that $\ints(v) = \nil$ and $\intn(v) = \natToSeq{\depth(v)}$ for all normal forms $v$, 
  the embedding $\ints(l\sigma) \gpopv[][] \ints(r\sigma)$ of root steps ($l \to r \in \RS_f$)
  results in the following order constraints.
  \begin{align*}
    \lst{\fsn{\m{f}}(\natToSeq{1})} & \gpopv[][] \nil && \text{from rule \rlbl{1}}\\
    \mparbox[r]{45mm}{\lst{\fsn{\m{f}}(\natToSeq{\depth(x\sigma) + 1})}} & \gpopv[][] \mparbox[l]{65mm}{\lst{\fsn{\m{g}}(\intn(\m{h}(\sn{x\sigma}{})))~\fsn{\m{f}}(\natToSeq{\depth(x\sigma)})}}  && \text{from rule \rlbl{2},}
  \end{align*}
  where $\intn(\m{h}(\sn{x\sigma}{})) = \lst{\fsn{\m{h}}(\intn(x\sigma))}\append \NM{\m{h}(\sn{x\sigma}{})}  = \lst{\fsn{\m{h}}(\natToSeq{\depth(x\sigma)})~\theconst}$.
  Consider now a step 
  $$
  s = f(s_1,\dots,s_i,\dots,s_{k+l}) \irew[\RS_f] f(s_1,\dots,t_i,\dots,s_{k+l}) = t \tkom
  $$ 
  below the root, where $s_i \irew[\RS_f] t_i$. Depending on the rewrite position $i$, which is 
  either normal ($i \in \{1,\dots,k\}$) or safe ($i \in \{k+1,\dots,k+l\}$), 
  the predicative embedding introduces one of the following two constraints:
  \begin{align*}
    \ints(s) 
    & = \lst{\fn(\intn(s_1), \dots, \intn(s_i), \dots, \intn(s_k))} \append \ints(s_{k+1}) \append \cdots \append \ints(s_{k+l}) \tag{a}\\
    & \gpopv[][] \lst{\fn(\intn(s_1), \dots, \intn(t_i), \dots, \intn(s_k))} \append \ints(s_{k+1}) \append \cdots \append \ints(s_{k+l}) = \intn(t), \text{ or} \\[1mm]
    \ints(s)
    & = \lst{\fn(\intn(s_1), \dots, \intn(s_k))} \append \ints(s_{k+1}) \append \cdots \append \ints(s_{i}) \append \cdots \append \ints(s_{k+l}) \tag{b} \\
    & \gpopv[][] \lst{\fn(\intn(s_1), \dots, \intn(s_k))} \append \ints(s_{k+1}) \append \cdots \append \ints(t_{i}) \append \cdots \append \ints(s_{k+l}) = \ints(t) \tpkt
  \end{align*}
  To be able to deal with steps below normal argument positions as in $(a)$, we also 
  orient images of $\intn$. This results additionally in following constraints:
  \begin{align*}
    \ints(\m{f}(\sn{0}{y\sigma})) \append \natToSeq{\depth(y\sigma) + 1} 
    & \gpopv[][] \ints(y\sigma) \append \natToSeq{\depth(y\sigma)}
    && \text{from rule \rlbl{1}}\\
    \mparbox[r]{45mm}{\ints(\m{f}(\sn{\ms(x\sigma)}{y\sigma}))\append \natToSeq{\depth(y\sigma) + 1}} 
    & \gpopv[][] \mparbox[l]{65mm}{\ints(\m{g}(\sn{\m{h}(\sn{x\sigma}{})}{\m{f}(\sn{x\sigma}{y\sigma})})) \append \natToSeq{\depth(y\sigma) + 2}}
    && \text{from rule \rlbl{2}}.
  \end{align*}
\end{example}


The \emph{polynomial path order on sequences} (\emph{\POP}~for short), 
denoted by $\gpopv[][]$, constitutes a generalisation of the \emph{path order for $\FP$} 
as put forward in~\cite{AM05}.
Whereas we previously used the notion of safe mapping to 
dictate predicative recursion on compatible TRSs, 
the order on sequences relies on the explicit separation of safe 
arguments as given by predicative interpretations.
Following Buchholz~\cite{B95}, it suffices to present \emph{finite approximations}
$\gpopv[k][l]$ of~$\gpopv[][]$.
The parameters $k \in \N$ and $l \in \N$ are used to controls the width and depth
of right-hand sides.
Fix a precedence $\qp$ on the normalised signature $\FSn$.
We extend term equivalence with respect to $\qp$ to sequences by 
disregarding the order on elements.
\begin{definition}\label{d:eqi}
  We define $a \eqi b$ if $a = b$ or there exists a permutation $\pi$
  such that $a_i \eqi b_{\pi(i)}$ for all $i = 1,\dots,n$, 
  where either 
  (i) $a = \lseq{a}$, $b = \lseq{b}$, or 
  (ii) $a = f(\seq{a})$, $b = g(\seq{b})$ and $f \ep g$.
\end{definition}

In correspondence to $\gpop$, the order $\gpopv[k][l]$ 
is based on an auxiliary order $\gppv[k][l]$.

\begin{definition}\label{d:gppv} 
  Let $k,l \geqslant 1$.
  We define $\gppv[k][l]$ with respect to the precedence $\qp$ inductively as follows:
  \begin{enumerate}[labelsep=*,leftmargin=*]
  \item\label{d:gppv:st}
    $f(\seq{a}) \gppv[k][l] b$ if $a_i \geqppv[k][l] b$  for some $i \in \{1,\dots,n\}$;
  \item\label{d:gppv:ia}
    $f(\seq{a}) \gppv[k][l] g(\seq[m]{b})$ if $f \sp g$ and the following conditions are satisfied:
    \begin{itemize}
    \item $f(\seq{a}) \gppv[k][l-1] b_j$ for all $j = 1,\dots,m$;
    \item $m \leqslant k$;
    \end{itemize}
  \item\label{d:gppv:ialst}
    $f(\seq{a}) \gppv[k][l] \lseq[m]{b}$ if the following conditions are satisfied:
    \begin{itemize}
    \item $f(\seq{a}) \gppv[k][l-1] b_j$ for all $j = 1,\dots,m$;
    \item $m \leqslant \width(f(\seq{a})) + k$;
    \end{itemize}
  \item\label{d:gppv:ms} 
    $\lseq{a} \gppv[k][l] \lseq[m]{b}$ if the following conditions are satisfied:
    \begin{itemize}
    \item $\lseq[m]{b} \eqi c_1 \append \cdots \append c_n$;
    \item $a_i \geqppv[k][l] c_i$ for all $i = 1, \dots, n$;
    \item $a_{i_0} \gppv[k][l] c_{i_0}$ for at least one $i_0 \in \{1, \dots, n\}$;
    \item $m \leqslant \width(\lseq{a}) + k$;
    \end{itemize}
  \end{enumerate}
  Here ${\geqppv[k][l]}$ denotes ${\gppv[k][l]} \cup {\eqi}$.  We write
  $\gppv[k]$ to abbreviate $\gppv[k][k]$.
\end{definition}

We stress that the definition
lacks a case $f(\seq{a}) \gppv[k][l] g(\seq[m]{b})$ where $f \ep g$.
Still the order is sufficient to account for 
terms oriented by the auxiliary order $\gsq$.

\begin{example}[Example~\ref{ex:pint} continued]\label{ex:gppv}
  Reconsider rule $\rref{RSf:2}$ from the TRS $\RS_f$ given in Example~\ref{ex:pint}, where
  in particular $\m{f}(\sn{\ms(x)}{y}) \gsq \m{h}(\sn{x}{})$.
  We show below $\fsn{\m{f}}(\natToSeq{\depth(x\sigma) + 1}) 
  \gppv[1][4] \intn(\m{h}(\sn{\ms(x\sigma)}{}))$ 
  for all substitutions $\ofdom{\sigma}{\VS \to \NF}$. 
  First recall that by the overloading of concatenation, we can write
  \begin{align*}
    \natToSeq{n} & = \lst{\theconst \cdots \theconst} = \theconst \append \cdots \append \theconst \append \nil \append \cdots \append \nil 
  \end{align*}
  with $n$ occurrences of $\theconst$, appending $m$-times the empty sequence $\nil$ 
  for all $n,m \in \N$. 
  Using that $\theconst \eqi \theconst$ and $\theconst \cppv{ialst}[k][l-1] \nil$ 
  for $l \geqslant 2$, 
  we can thus prove $\natToSeq{n+m} \cppv{ms}[k][l] \natToSeq{n}$ whenever $m \geqslant 1$.
  Moreover we have
  \begin{align*}
    \rlbl{1}: && \natToSeq{\depth(x\sigma) + 1} 
    & \cppv{ms}[1][2] \natToSeq{\depth(x\sigma)}  
    && \text{as $\depth(x\sigma) + 1 > \depth(x\sigma)$} 
    \\
    \rlbl{2}: && \fsn{\m{f}}(\natToSeq{\depth(x\sigma) + 1}) 
    & \cppv{ia}[1][3] \fsn{\m{h}}(\natToSeq{\depth(x\sigma)}) 
    && \text{using $\fsn{\m{f}} \sp \fsn{\m{h}}$ and \rlbl{1}} 
    \\
    \rlbl{3}: && \fsn{\m{f}}(\natToSeq{\depth(x\sigma) + 1}) 
    & \cppv{ialst}[1][4] \lst{\fsn{\m{h}}(\natToSeq{\depth(x\sigma)})~\theconst} 
    && \text{by \rlbl{2} and $\fsn{\m{f}}(\dots) \cppv{ia}[1][3] \theconst$} \\ 
    &&& = \intn(\m{h}(\sn{x\sigma}{})) 
    \tpkt
  \end{align*}
\end{example}

\noindent We arrive at the definition of the full order $\gpopv[k][l]$.
\begin{definition}\label{d:gpopv} 
  Let $k,l \geqslant 1$.
  We define $\gpopv[k][l]$ inductively as the least extension of $\gppv[k][l]$ such that:
  \begin{enumerate}[labelsep=*,leftmargin=*]
  \item\label{d:gpopv:st}
    $f(\seq{a}) \gpopv[k][l] b$ if $a_i \geqpopv[k][l] b$ for some $i \in \{1,\dots,n\}$;
  \item\label{d:gpopv:ep}
    $f(\seq{a}) \gpopv[k][l] g(\seq[m]{b})$ if $f \ep g$ 
    and following conditions are satisfied: 
    \begin{itemize}
    \item $\mset{\seq{a}}~\mextension{\gpopv[k][l]}~\mset{\seq[m]{b}}$;
    \item $m \leqslant k$;
    \end{itemize}
  \item\label{d:gpopv:ialst} 
    $f(\seq{a}) \gpopv[k][l] \lseq[m]{b}$ 
    and following conditions are satisfied: 
    \begin{itemize}
    \item $f(\seq{a}) \gpopv[k][l-1] b_{j_0}$ for at most one $j_0 \in\{1,\dots,m\}$;
    \item $f(\seq{a}) \gppv[k][l-1] b_j$ for all $j \neq j_0$;
    \item $m \leqslant \width(f(\seq{a})) + k$;
    \end{itemize}
  \item\label{d:gpopv:ms} 
    $\lseq{a} \gpopv[k][l]  \lseq[m]{b}$
    and following conditions are satisfied: 
    \begin{itemize}
    \item $\lseq[m]{b} \eqi c_1 \append \cdots \append c_n$;
    \item $a_i \geqpopv[k][l] c_i$ for all $i = 1,\dots, n$; 
    \item $a_{i_0} \gpopv[k][l] c_{i_0}$ for at least one $i_0 \in \{1,\dots, n\}$; 
    \item $m \leqslant \width(\lseq{a}) + k$;
    \end{itemize}
  \end{enumerate}
  Here ${\geqpopv[k][l]}$ denotes ${\gpopv[k][l]} \cup {\eqi}$. 
  We write $\gpopv[k]$ to abbreviate $\gpopv[k][k]$.
\end{definition}

The polynomial path order on sequences forms a restriction of the recursive path order with 
multiset status, where the variadic symbol $\listsym$ is implicitly 
ranked lowest in the precedence. As a consequence the order
it is well-founded~\cite{Ferreira95}.
The use of the auxiliary order in $\cpopv{ialst}[k][l]$ accounts 
for our restriction that predicative recursive TRSs admit at most one 
recursive call per right-hand side.
Observe $t \cpopv{ialst} \nil$ for all terms $t \not \in \VS$, 
consequently $\lseq{t} \cpopv{ms} \nil$  if at least one term $t_i$ is ground.

\begin{example}[Example~\ref{ex:gppv} continued]\label{ex:gpopv} 
  We continue with the orientation of root steps from
  the TRS $\RS_f$ depicted in Example~\ref{ex:pint}
  for substitutions $\ofdom{\sigma}{\VS \to \NF}$.
  Consider the more involved case $\m{f}(\sn{\ms(x\sigma)}{y\sigma}) \irew[\RS_f] \m{g}(\sn{\m{h}(\sn{x\sigma}{})}{\m{f}(\sn{x\sigma}{y\sigma})})$ due to rule \rref{RSf:2}.
  Note that in the orientation below we
  use $\cpopv{ep}$ to orient the recursive call (proof step \rlbl{5}), 
  and $\cppv{ia}$ for the remaining elements (proof step \rlbl{6}).
  \begin{align*}
    \rlbl{4}: && \natToSeq{\depth(x\sigma) + 1} 
    & \cpopv{ms}[1][2] \natToSeq{\depth(x\sigma)}  
    && \text{as in Example~\ref{ex:gppv}} 
    \\
    \rlbl{5}: && \fsn{\m{f}}(\natToSeq{\depth(x\sigma) + 1})
    & \cpopv{ep}[1][3] \fsn{\m{f}}(\natToSeq{\depth(x\sigma)})
    && \text{using \rlbl{4}} 
    \\
    \rlbl{6}: && \fsn{\m{f}}(\natToSeq{\depth(x\sigma) + 1})
    & \cppv{ia}[1][5] \fsn{\m{g}}(\intn(\m{h}(\sn{x\sigma}{})))
    && \text{using $\fsn{\m{f}} \sp \fsn{\m{g}}$ and \rlbl{3}} 
    \\
    \rlbl{7}: && \fsn{\m{f}}(\natToSeq{\depth(x\sigma) + 1})
    & \cpopv{ialst}[1][6] \lst{\fsn{\m{g}}(\intn(\m{h}(\sn{x\sigma}{})))~\fsn{\m{f}}(\natToSeq{\depth(x\sigma)})}
    && \text{using \rlbl{5} and \rlbl{6}}
    \\
    &&& = \ints(\m{g}(\sn{\m{h}(\sn{x\sigma}{})}{\m{f}(\sn{x\sigma}{y\sigma})}))
    \\
    \rlbl{8}: && \ints(\m{f}(\sn{\ms(x\sigma)}{y\sigma})) & = 
    \lst{\fsn{\m{f}}(\natToSeq{\depth(x\sigma) + 1})} \\
    &&& \cpopv{ialst}[1][6] \ints(\m{g}(\sn{\m{h}(\sn{x\sigma}{})}{\m{f}(\sn{x\sigma}{y\sigma})}))
    && \text{using \rlbl{7}}
    \\
    \rlbl{9}: && \fsn{\m{f}}(\natToSeq{\depth(x\sigma) + 1})
    & \cpopv{ialst}[2][6] \lst{\fsn{\m{g}}(\intn(\m{h}(\sn{x\sigma}{})))~\fsn{\m{f}}(\natToSeq{\depth(x\sigma)})~\theconst}
    && \text{by \rlbl{5}, \rlbl{6} and $\fsn{\m{f}}(\dots) \gpopv[1][1] \theconst$}
    \\
    \rlbl{10}: && \intn(\m{f}(\sn{\ms(x\sigma)}{y\sigma})) & = 
    \lst{\fsn{\m{f}}(\natToSeq{\depth(x\sigma) + 1})}\append \natToSeq{\depth(y\sigma) + 1} \\
    &&& \cpopv{ialst}[2][6] \ints(\m{g}(\sn{\m{h}(\sn{x\sigma}{})}{\m{f}(\sn{x\sigma}{y\sigma})})) \append \natToSeq{\depth(y\sigma) + 2}
    && \text{using \rlbl{9} and $\theconst \eqi \theconst$}
    \\
    &&& = \intn(\m{g}(\sn{\m{h}(\sn{x\sigma}{})}{\m{f}(\sn{x\sigma}{y\sigma})})) \tpkt
  \end{align*}
  For the last orientation we employ that the width of the left-hand side 
  is at least $\depth(y\sigma) + 2$, and the length of the right hand side is 
  $\depth(y\sigma) + 4$, as required we have 
  $$
  \len(\intn(\m{g}(\sn{\m{h}(\sn{x\sigma}{})}{\m{f}(\sn{x\sigma}{y\sigma})})))
  \leqslant \width(\intn(\m{f}(\sn{\ms(x\sigma)}{y\sigma}))) + 2 \tpkt
  $$
\end{example}

Observe that as in the above example, 
the parameter $l$ in $\gpopv[k][l]$ controls 
the depth of the proof tree of $a \gpopv[k][l] b$. 
Since leafs of such proof trees hold either due to case $\cpopv{st}$
or the absence of arguments in the right-hand side, it follows that 
the depth of $b$ is bounded linearly 
in $l$ and the depth of $a$.
From the example it should also be clear how the parameter $k$ 
controls the length of right-hand sides, 
compare steps~\rlbl{9} and~\rlbl{10} where we had to increase 
the parameter $k$.

In the example we obtained that the predicative embedding of 
root steps $l\sigma \irew[\RS_f] r\sigma$ 
of predicative recursive TRS $\RS_f$ is possible for $k = 2$, 
independent on the considered substitution $\sigma$. 
The next lemma clarifies that such a global $k$ can always be found, 
and depends on the right-hand sides only. 

\begin{lemma}\label{l:int:len}
  Let $s = f(\pseq{s}) \in \Tb$, $t \in \TERMS$, and $\sigma$ be a normalising substitution. Then
  \begin{enumerate}[labelsep=*,leftmargin=*]
  \item\label{l:int:len:S} $\len(\ints(t\sigma)) \leqslant \size{t}$; and
  \item\label{l:int:len:gsq} if $s \gsq t$  then $\len(\intn(t\sigma)) \leqslant \max \{ \norm{s_1\sigma}, \dots, \norm{s_k\sigma} \} + 2\cdot\size{t}$; and
  \item\label{l:int:len:gpop} if $s \gpop t$ then $\len(\intn(t\sigma)) \leqslant \max \{ \norm{s_1\sigma}, \dots, \norm{s_{k+l}\sigma}\} + 2\cdot\size{t}$.
  \end{enumerate}
\end{lemma}

As a consequence of the above lemma we obtain: if $l\sigma \irew r\sigma$ is a root step of 
a predicative TRS $\RS$, then 
$\len(\intq(r\sigma)) \leqslant \width(\intq(l\sigma)) + 2 \cdot \size{r}$ 
for $\intq \in \{\ints,\intn\}$. 
In the predicative embedding we instantiate 
$k$ by twice the maximum size of right-hand sides of $\RS$.
The side-conditions imposed on $\gpopv[k][l]$ 
allow us to estimate the length of right-hand sides based on 
the width of left-hand sides and the parameter $k$.
This and other frequently used properties are collected in the next lemma,
whose proof is not difficult.

\begin{lemma}\label{l:approx}
  The following properties hold for all $k \geqslant 1$ and $a,b,c_1,c_2 \in \TLS$.
  \begin{enumerate}[labelsep=*,leftmargin=*]
  \item\label{l:approx:kmon} ${\gppv[l]} \subseteq {\gpopv[l]} \subseteq {\gpopv[k]}$ for all $l \leqslant k$;
  \item\label{l:approx:modeqi} ${\eqi} \cdot {\gpopv[k]} \cdot {\eqi} \subseteq {\gpopv[k]}$;
  \item\label{l:approx:bound} $a \gpopv[k] b$ implies $\len(b) \leqslant \width(a) + k$;
  \item\label{l:approx:subseq} $a \gpopv[k] b$ implies ${c_1 \append a \append c_2} \gpopv[k] {c_1 \append b \append c_2}$.
  \end{enumerate}
\end{lemma}

Following~\cite{AM05} we define a function $\Slow[k]$ that measures the 
$\gpopv[k]$-descending lengths on sequences. To simplify matters, 
we restrict the definition of $\Slow[k]$ to ground sequences.
As images of predicative interpretations are always ground, this suffices for our purposes.
\begin{definition}
We define $\ofdom{\Slow[k]}{\GTLS \to \N}$
as
\begin{equation*}
  \Slow[k](a) \defsym 1 + \max
 \{ \Slow[k](b) \mid b \in \GTLS \text{ and } a \gpopv[k] b
 \}\tpkt
\end{equation*}
\end{definition}

Note that due to Lemma~\eref{l:approx}{modeqi}, $\Slow[k](a) = \Slow[k](b)$
whenever $a \eqi b$. The next lemma confirms that sequences act purely as containers.

\begin{lemma}\label{l:slowsum} 
  For $\lseq{a} \in \GLS$ it holds that $\Slow[k](\lseq{a}) = \sum_{i=1}^n \Slow[k](a_i)$.
\end{lemma}
\begin{proof}
  Let $a = \lseq{a} \in \GLS$.
  We first show $\Slow[k](a) \geqslant \sum_{i=1}^n \Slow[k](a_i)$.
  Let $b,c \in \GTLS$ and consider maximal sequences
  $b \gpopv[k] b_1 \gpopv[k] \cdots \gpopv[k] b_o$ and
  $c \gpopv[k] c_1 \gpopv[k] \cdots \gpopv[k] c_p$.
  Using Lemma~\eref{l:approx}{subseq} repeatedly we get
  $$
  b \append c \gpopv[k] b_1 \append c \gpopv[k] \cdots \gpopv[k] b_o \append c
  \gpopv[k] b_o \append c_1 \gpopv[k]  \cdots \gpopv[k] b_o \append c_p \tkom
  $$
  and thus 
  $\Slow[k](b \append c) \geqslant \Slow[k](b) + \Slow[k](c)$ holds
  for all $b,c \in \GTLS$. We conclude
  $
  \Slow[k](a) =  \Slow[k](a_1 \append \cdots \append a_n) \geqslant \sum_{i=1}^n \Slow[k](a_i)
  $ 
  with a straight forward induction on $n$.

  It remains to verify $\Slow[k](a) \leqslant \sum_{i=1}^n \Slow[k](a_i)$.
  For this we show that $a \gpopv[k] b$ implies $\Slow[k](b) < \sum_{i=1}^n \Slow[k](a_i)$
  by induction on $\Slow[k](a)$.
  Consider the base case $\Slow[k](a) = 0$.
  Since $a$ is ground it follows that $a = \nil$, the claim is trivially satisfied.
  For the inductive step $\Slow[k](a) > 1$, 
  let $a \gpopv[k] b$.
  Since $a$ is a sequence, $a \cpopv{ms} b$.
  Hence $b \eqi b_1 \append \cdots \append b_n$ where $a_i \geqpopv[k] b_i$ 
  and thus $\Slow[k](b_i) \leqslant \Slow[k](a_i)$ for all $i =1,\dots,n$.
  Additionally $a_{i_0} \gpopv[k] b_{i_0}$ and hence $\Slow[k](b_{i_0}) < \Slow[k](a_{i_0})$ 
  for at least one $i_0 \in \{1,\dots,n\}$.
  As in the first half of the proof, one verifies $\Slow[k](b_i) \leqslant \Slow[k](b)$
  for all $i = 1,\dots,n$.
  Note $\Slow[k](b) < \Slow[k](a)$ as $a \gpopv[k] b$, 
  hence induction hypothesis is applicable to $b$ and all $b_i$ ($i = 1,\dots,n$).
  It follows that 
  \begin{align*}
    \Slow[k](b) \leqslant \sum_{c \in b} \Slow[k](c) 
    \leqslant \sum_{i=1}^n \sum_{c \in b_i} \Slow[k](c)
    = \sum_{i=1}^n \Slow[k](b_i) 
    < \sum_{i=1}^n \Slow[k](a_i) \tpkt
  \end{align*}
  This concludes the second part of the proof.
\end{proof}

The central theorem of this section, Theorem~\ref{t:pop}, states that $\Slow[k](f(a_1, \dots, a_n))$
is polynomial in $\sum_{i}^n \Slow[k](a_i)$, 
where the polynomial bound depends only on $k$ and the rank $p$ of $f$.
The proof of this is involved.
To cope with the multiset comparison underlying $\cpopv{ep}[k]$,
we introduce as a first step an \emph{order-preserving} 
extension $\MSlow{n}{k}$ of $\Slow[k]$ to multisets of sequences, 
in the sense that $\MSlow{n}{k}(\seq{a}) > \MSlow{m}{k}(\seq[m]{b})$ holds 
whenever $\mset{\seq[n]{a}} \mextension{\gpopv[k]} \mset{\seq[m]{b}}$ 
(provided $k \geqslant m,n$, cf.~Lemma~\ref{l:slowpoly}).
As the next step toward our goal, we estimate $\Slow[k](f(a_1,\dots,a_n))$ in terms 
of $\MSlow{n}{k}(\seq{a})$
whenever $n \leqslant k$ and $\rk(f) \leqslant k$.
Technically we bind following functions by polynomials $q_{k,p}$. 
For all $k,p \in \N$ with $k \geqslant 1$
we define $\ofdom{\Fpop{k}{p}}{\N \to \N}$ as 
  \begin{multline*}
    \Fpop{k}{p}(m) \defsym \max\{ \Slow(f(\seq{a})) \mid \\
    f(\seq{a}) \in \GTS,
    ~\rk(f) \leqslant p,~n \leqslant k
    \text{ and } \MSlow{n}{k}(a_1,\dots,a_n) \leqslant m \} \tpkt
  \end{multline*}

The definition of $\MSlow{n}{k}$ is defined in terms of an order-preserving homomorphism from $\msetover(\N)$ to $\N$. 
To illustrate the construction carried out below, consider the following example.
\begin{example}\label{ex:multisets}
  Consider multisets $\msetover(\N)$ of size $k$.
  Conceive such multisets $\mset{m_1,\dots,m_k}$ as 
  natural 
  numbers written in base-$c$ (with $c>m_i$ for all $i = 1,\dots,k$), 
  where digits $m_1,\dots,m_k$ are sorted from left to right in decreasing order.
  Then one can formulate chains $M_1 \mextension{>} M_2 \mextension{>} \dots$ 
  that can be understood as decreasing counters 
  which however wrap from
  $\{m_1,\dots, m_i+1,0,\dots,0\}$ to $\{m_1,\dots, m_i,m_i,\dots,m_i\}$. 
  Compare the TRS $\RS_k$ defined in Lemma~\ref{l:imprecise} that models such counters.\@
  Using the correspondence, it is
  easy to prove that the length of a chain of this form
  starting from $\mset{c-1,\dots,c-1}$
  is given by
  \begin{align*}
  \sum_{m_1=0}^{c-1} \sum_{m_2=0}^{m_1} \dots \sum_{m_k=0}^{m_{k-1}} m_k = \Omega(c^{k+1}). 
  \end{align*}
  The inclusion follows by $k$-times application of the 
  Faulhaber's formula~\cite{K93}, which states that 
  for all $n,l \in \N$, 
  $\sum_{i=1}^n i^l = p_{l+1}(n)$ for some polynomial $p_{l+1}$ of degree $l+1$. 
\end{example}

The above example gives a polynomial lower bound on the number of $\mextension{>}$ descending 
sequences on multisets $\msetover(\N)$ of size $k$.
We now prove that this lower bound also serves as an asymptotic upper bound, 
for all multisets of natural number of length \emph{up to} $k$. 
For $k\geqslant n \in \N$ and $c \in \N$ we define the family of functions $\ofdom{\homo{n}{k}{c}}{\N^l \to \N}$ 
such that 
\begin{align*}
  \homo{n}{k}{c}(\seq[n]{m}) = \sum_{i=1}^n \sort{n}(\seq[n]{m},i) \cdot c^{(k - i)} \tpkt
\end{align*}
Here $\sort{n}(\seq[n]{m},i)$ denote the \nth{$i$}\ element of $\seq[n]{m}$ sorted in descending order, i.e., 
$\sort{n}(\seq[n]{m},i) \defsym m_{\pi(i)}$ 
for $i =1,\dots,n$ and some permutation $\pi$ such that $m_{\pi(i)} \geqslant m_{\pi(i+1)}$ ($i \in \{1,\dots,n-1\}$).
%
\begin{lemma}\label{l:homo}
  Let $k,n,n' \in \N$ such that $k \geqslant n,n'$. Then
  for all $\seq[n]{m} \in \N$ and $c > \seq[n]{m}$ we obtain:
  \begin{enumerate}[labelsep=*,leftmargin=*]
  \item\label{l:homo:5} 
    $\mset{\seq[n]{m}} \mextension{>} \mset{\seq[n']{m'}}$ implies $\homo{n}{k}{c}(\seq[n]{m}) > \homo{n'}{k}{c}(\seq[n']{m'})$, and
  \item\label{l:homo:4} 
    $c^{k} > \homo{n}{k}{c}(\seq[n]{m})$.
  \end{enumerate}
\end{lemma}

Let $k \in \N$ be fixed and let $M \subseteq \msetover(\N)$ collect all multisets of size up to $k$.
By Lemma~\ref{l:homo} the functions $\homo{l}{k}{\cdot}$ 
gives an order preserving homomorphism from $(M,\mextension{>})$ to $(\N,>)$.
Furthermore this homomorphism is polynomially bounded in its elements.
We extend this homomorphism to multisets $\mset{\seq[n]{a}}$ over $\GTLS$.
%
  Let $k,n \in \N$ such that $k \geqslant n$.
  We define $\ofdom{\MSlow{n}{k}}{\GLS^n \to \N}$ as follows:
  $
  \MSlow{n}{k}(\seq[n]{a}) \defsym \homo{n}{k}{c}(\Slow[k](a_1),\dots,\Slow[k](a_n))
  $,
  where $c = 1 + \max \set{\Slow[k](a_i) \mid i \in \{1,\dots,n\}}$.
%
We obtain:
\begin{lemma}\label{l:slowpoly}
  Let $\seq{a},\seq[m]{b}\in \GTLS$ and let $k \geqslant m,n$. Then
  $$
  \mset{\seq{a}} \mextension{\gpopv[k]} \mset{\seq[m]{b}} 
  \quad \IImp \quad \MSlow{n}{k}(\seq{a}) > \MSlow{m}{k}(\seq[m]{b}) \tpkt
  $$
\end{lemma}
In Theorem~\ref{t:pop} below we prove $\Fpop{k}{p}(m) \leqslant c \cdot {(m+2)}^{d}$
for some constants $c,d \in \N$ depending only on $k$ and $p$. 
Inevitably the proof of Theorem~\ref{t:pop} is technical, 
the reader may to skip the formal proof on the first read.
In the proof of Theorem~\ref{t:pop}, we instantiate the constants $c,d$ by parameters $c_{k,p},d_{k,p} \in \N$, 
which are defined by recursion on $p$ as follows:\label{d:dkp:ckp}
\begin{align*}
 d_{k,p} &\defsym 
 \begin{cases}
  k+1 & \text{if $p = 0$,} \\
  {(d_{k,p-1} \cdot k)}^{k+1}+1 & \text{otherwise;}
 \end{cases}
 \\
 c_{k,p} &\defsym 
 \begin{cases}
  k^k & \text{if $p = 0$,} \\
  {(c_{k,p-1} \cdot k)}^{\sum_{i=1}^k {(k \cdot d_{k,p-1})}^i} & \text{otherwise.}
 \end{cases}
\end{align*}
The theorem is then proven by induction on $p$ and $m$.
Consider term $f(\seq{a})$ with $k \geqslant n$ and $\MSlow{n}{k}(\seq{a}) \leqslant m$.
At the heart of the proof, we show that 
$c_{k,p} \cdot {(m+2)}^{d_{k,p}} > \Slow[k](b)$ for arbitrary $b$ with $f(\seq{a}) \gpopv[k] b$. 
The most involved case is $f(\seq{a}) \cpopv{ialst}[k][l] \lseq[o]{b}$ where for all 
but one $j \in \{1,\dots,o\}$ we have $f(\seq{a}) \gppv[k][l-1] b_j$. 
Here it is important to give a precise analysis of the order $\gppv[k][l]$, 
exploiting the parameters $k$ and $l$. 
To this avail we define for $l,k \geqslant 1$ and $p \in \N$ a family of auxiliary functions 
$\ofdom{g_{l,k,p}}{\N \to \N}$ by
\begin{align*}
  g_{k,l,p}(m) & \defsym 
  \begin{cases}
    k^l \cdot m^l & \text {if $p = 0$,} \\
    m & \text{if $p > 0$ and $l = 1$,} \\
    c_{k,p-1} \cdot (m \cdot g_{k,l-1,p}(m))^{k \cdot d_{k,p-1}} & \text{otherwise.}
  \end{cases}
\end{align*}
Having as premise the induction hypothesis of the main proof, the next lemma 
explains the r\^ole of $\gppv[k][l]$.
\begin{lemma}\label{l:pop:aux}
  Let $f(\seq{a})  \in \GTS$.
  Let $k \geqslant n$ and 
  $m \geqslant \MSlow{n}{k}(\seq{a})$.
  Suppose $\Fpop{k}{p}(m') \leqslant c_{k,p} {(m'+2)}^{d_{k,p}}$ for all $p < \rk(f)$ and $m'$.
  Then for all $b \in \GTLS$, 
  $$
  f(\seq{a}) \gppv[k][l] b \quad \IImp \quad \Slow(b) \leqslant g_{k,l,\rk(f)}(m+2) \tpkt
  $$
\end{lemma}
\begin{proof}

  We prove lemma by induction on $l$.
  The base case $l = 1$ is easy to show, hence 
  assume $l > 1$. 
  Suppose $f(\seq{a}) \gppv[k][l] b$, we continue by case analysis:
  \begin{description}[leftmargin=0.3cm]
  \item[\dcase{$f(\seq{a}) \cppv{st}[k][l] b$}] 
  Then $a_i \gppv[k][l] b$ for some $i \in \{1,\dots,n\}$,
  and consequently $\Slow(b) \leqslant \Slow(a_i)$. 
  Then by definition and assumption we even have $\Slow(a_i) \leqslant \MSlow{n}{k}(\seq{a}) \leqslant m$.

  \item[\dcase{$f(\seq{a}) \cppv{ia}[k][l] b$ where $b = g(\seq[o]{b})}$] 
    Then 
    $f(\seq{a}) \gppv[k][l-1] b_j$ for all $j = 1,\dots,o$.
    Set $m' \defsym \MSlow{o}{k}(\seq[o]{b})$. We have
    \begin{align*}
      m' & < \max\set{\Slow[k](b_j) + 1 \mid j \in \{1,\dots,o\}}^k 
      && \text{by definition and Lemma~\eref{l:homo}{4}} \\
      & \leqslant {(g_{k,l-1,\rk(f)}(m+2)+ 1)}^k
      && \text{applying induction hypothesis.}
    \end{align*}
    As in the considered case $\rk(g) < \rk(f)$ holds, we have
    $\Slow[k](b) \leqslant \Fpop{k}{\rk(g)}(m')$ 
    and so by assumption and arithmetical reasoning we conclude
    \begin{align*}
      \Slow[k](b) 
      & \leqslant c_{k,\rk(g)} \cdot {(m' + 2)}^{d_{k,\rk(f)-1}} && \\
      & < c_{k,\rk(f)-1} \cdot {({(g_{k,l-1,\rk(f)}(m+2)+ 1)}^k + 2)}^{d_{k,\rk(f)-1}} && \text{substituting bound for $m'$} \\
      & \leqslant c_{k,\rk(f)-1} \cdot {((m+2) \cdot g_{k,l-1,\rk(f)}(m+2))}^{k \cdot d_{k,\rk(f)-1}} && \\
      & = g_{k,l,\rk(f)}(m+2) && \text{using $\rk(f) > 0$}\tpkt
    \end{align*}
  \item[\dcase{$f(\seq{a}) \cppv{ialst}[k][l] b$ where $b = \lseq[o]{b}}$]
    Order constraints give $o \leqslant \width(a) + k$
    and $f(\seq{a}) \gppv[k][l-1] b_j$ ($j = 1,\dots,o$).
    Exploiting that $a_i$ is ground, a standard argument shows 
    that $\width(a_i) \leqslant \Slow[k](a_i)$, and consequently 
    $\width(a_i) \leqslant m$.
    Thus
    \begin{align}
      \label{e:bindwidth}
      o \leqslant \width(a) + k 
      = \max\set{1,\width(a_1), \dots,\width(a_n)} + k
      \leqslant m + k \leqslant k \cdot (m + 1) \tpkt
    \end{align}
    Since by Lemma~\ref{l:slowsum} we have $\Slow[k](b) = \sum_{i=1}^o \Slow[k](b_i)$, 
    using \eqref{e:bindwidth} we see
    \begin{align*}
     \Slow[k](b) & \leqslant k \cdot (m + 1) \cdot g_{k,l-1,\rk(f)}(m+2) && \text{by induction hypothesis} \\
     & < g_{k,l,\rk(f)}(m+2) && \text{by case analysis.}
    \end{align*}
  \end{description}
\end{proof}

\begin{theorem}\label{t:pop}
  Let $k \geqslant 1$ and $p \in \N$. There exists constants $c,d \in \N$ 
  (depending only on $k$ and $p$) such that 
  for all $m \in \N$ we have 
  $$
  \Fpop{k}{p}(m) \leqslant c \cdot {(m+2)}^{d} \tpkt
  $$
\end{theorem}
\begin{proof}
  Fix $f(\seq[n]{a}) \in \GTS$ such that $\rk(f) = p$, $k \geqslant n$ and $\MSlow{n}{k}(a_1,\dots,a_n) \leqslant m$.
  To show the theorem, we prove that for all $b$ with $f(\seq[n]{a}) \gpopv[k] b$ 
  we have $\Slow[k](b) < c_{k,p} \cdot {(m+2)}^{d_{k,p}}$ 
  for constants $c_{k,p}$ and $d_{k,p}$ as defined on page~\pageref{d:dkp:ckp}.
  The proof is by induction on the lexicographic combination of $p$ and $m$.
  The base case where $p = 0$ and $m = 0$ is easy to proof, we consider the inductive 
  step. 
  Consider the inductive step. By induction hypothesis 
  we have 
  \begin{align*}
      \Fpop{k}{p'}(m') \leqslant c_{k,p} \cdot {(m'+2)}^{d_{k,p'}} \text{\quad if $p' < p$, or $p' = p$ and $m' < m$.}
  \end{align*}
  For $p' < p$ we will use the induction hypothesis as a premise to Lemma~\ref{l:pop:aux}, 
  for $p' = p$ we use below the consequence
  \begin{align}
    \label{t:pop:ih}
    \Slow[k](g(\seq[o]{b})) < c_{k,p} \cdot (m+1)^{d_{k,p}} 
    \text{\quad if $f \ep g$, $o \leqslant k$ and $\MSlow{o}{k}(b_1,\dots,b_o) < m$} \tpkt
  \end{align}
  We analyse the cases $p = 0$ and $p > 0$ separately. 
  In both cases we perform a side induction on $l$. 
  \smallskip 
  \begin{description}[leftmargin=0.3cm]
  \item[\dcase{$p = 0$}]
    By side induction on $l$ we prove that $\Slow[k](b) < k^k \cdot {(m+1)}^{k+1} + k^l \cdot {(m + 2)}^l$
    for all $b$ with $f(\seq{a}) \gpopv[k][l] b$.
    
    Note that if $f(\seq{a}) \cpopv{st}[k][l] b$ holds, 
    as in the proof of Lemma~\ref{l:pop:aux},
    we even have $\Slow[k](b) \leqslant \MSlow{n}{k}(\seq[n]{a}) \leqslant m$. 

    Consider now $f(\seq{a}) \cpopv{ep}[k][l] b$ where $g(\seq[o]{b})$. 
    The ordering constraints give $o \leqslant k$ and
    $\mset{\seq{a}} \mextension{\gpopv[k][l]} \mset{\seq[o]{b}}$, 
    from Lemma~\ref{l:slowpoly} we thus get $\MSlow{o}{k}(\seq[o]{b}) < \MSlow{n}{k}(\seq[n]{a}) \leqslant m$. 
    Since also $f \ep g$ in this case
    we conclude as we even have
    \begin{align*}
      \Slow[k](g(\seq[o]{b})) 
      & < c_{k,0} \cdot (m + 1)^{d_{k,0}} 
      && \text{by main induction hypothesis}\\
      & = k^k \cdot (m + 1)^{k+1} 
      && \text{by definition of $c_{k,0}$ and $d_{k,0}$.}
    \end{align*}

    Next consider $f(\seq{a}) \cpopv{ialst}[k][l]$ where $\lseq[o]{b}$.
    The order constraints give
    (i) $a \gpopv[k][l-1] b_{j_0}$ for some $j_0 \in\{1,\dots,o\}$,
    (ii) $a \gppv[k][l-1] b_j$ for all $j \neq j_0$, and
    (iii) $o \leqslant \width(a) + k$.
    We have 
    \begin{align*}
      \Slow[k](b_{j_0}) & < k^k \cdot {(m+1)}^{k+1} + k^{l-1} \cdot {(m + 2)}^{l-1} && \text{from (i), using SIH on $l$ } \\
      \Slow[k](b_j) & \leqslant k^{l-1} \cdot {(m + 2)}^{l-1}\text{ for $j \not = j_o$} && \text{from (ii), using Lemma~\ref{l:pop:aux}} \\
      o & \leqslant k \cdot (m + 1) && \text{from (iii).}
    \end{align*}
    For the last inequality, compare Equation~\eqref{e:bindwidth} from Lemma~\ref{l:pop:aux}.
    As $\Slow[k](b) = \sum_{j =1}^o \Slow[k](b_j)$ by Lemma~\ref{l:slowsum}, substituting 
    the above inequalities we get
    \begin{align*}
      \Slow[k](b) & <  k^k \cdot {(m+1)}^{k+1} + k^{l-1} \cdot {(m + 2)}^{l-1} && \text{bound on $\Slow[k](b_{j_0})$}\\
      & \quad + (k \cdot (m+1) - 1) \cdot k^{l-1} \cdot {(m + 2)}^{l-1} && \text{bound on $o$ and $\Slow[k](b_{j})$, $j \not=j_0$}\\ 
      & <  k^k \cdot {(m+1)}^{k+1} + k^{l} \cdot {(m + 2)}^{l} \tpkt
    \end{align*}
    This concludes the final case of the side induction.
    Since ${\gpopv[k]} = {\gpopv[k][k]}$ this preparatory step gives
    \begin{align*}
      \Slow[k](b) < k^k \cdot {(m+1)}^{k+1} + k^k \cdot {(m + 2)}^k 
      \leqslant k^k \cdot {(m + 2)}^{k+1}\tkom
    \end{align*}
    we conclude the case $p = 0$.

    \medskip

  \item[\dcase{$p > 0$}]
    We show first that for all $k \geqslant l$, if $f(\seq{a}) \gpopv[k][l] b$ then
    \begin{equation}
      \label{t:pop:a}
      \Slow[k](b) \leqslant c_{k,p}  \cdot {(m+1)}^{d_{k,p}} + c_{k,p} \cdot {(m + 2)}^{{(k \cdot d_{k,p-1})}^{l+1}} \tpkt
    \end{equation}
    The proof is by induction on $l$. 
    Suppose $f(\seq{a}) \gpopv[k][l] b$. 
    The base case $l = 1$ is trivial, so consider the inductive step $l > 1$.
    As in the case $p = 0$,
    if $f(\seq{a}) \cpopv{st}[k][l] b$ then even $\Slow[k](b) \leqslant m$, 
    and if $f(\seq{a}) \cpopv{ep}[k][l] b$ then even $\Slow[k](b) \leqslant c_{k,p}  \cdot {(m+1)}^{d_{k,p}}$.
    Consider $f(\seq{a}) \cpopv{ialst}[k][l] b$.
    Then $b = \lseq[o]{b}$ with 
    (i) $a \gpopv[k][l-1] b_{j_0}$ for some $j_0 \in \{1,\dots,o\}$,
    (ii) $a \gppv[k][l-1] b_j$ for all $j \neq j_0$, and
    (iii) $o \leqslant \width(a) + k$.
    A standard argument gives
    \begin{align*}
      g_{k,l,p}(n) \leqslant c_{k,p-1}^{\sum_{i=1}^{l-1} {(k \cdot d_{k,p-1})}^i} \cdot n^{\sum_{i=1}^{l} {(k \cdot d_{k,p-1})}^i} \tkom
    \end{align*}
    for all $n \in \N$, thus
    \begin{align*}
      ~\Slow[k](b_{j_0}) & <  c_{k,p} \cdot {(m+1)}^{d_{k,p}} + c_{k,p} \cdot {(m + 2)}^{{(k \cdot d_{k,p-1})}^{l}}
      && \text{from (i), using SIH on $l$ } \\
      \Slow[k](b_j) & \leqslant c_{k,p-1}^{\sum_{i=1}^{l-2} {(k \cdot d_{k,p-1})}^i} \cdot (m+2)^{\sum_{i=1}^{l-1} {(k \cdot d_{k,p-1})}^i} \text{ for $j \not = j_o$} 
      && \text{from (ii), using Lemma~\ref{l:pop:aux}} \\
      o & \leqslant k \cdot (m + 1) && \text{from (iii).}
    \end{align*}
    Using Lemma~\ref{l:slowsum} and substituting 
    the above inequalities we get
    \begin{align*}
      \quad\Slow[k](b) 
      & \leqslant c_{k,p} \cdot {(m+1)}^{d_{k,p}} + c_{k,p} \cdot {(m + 2)}^{{(k \cdot d_{k,p-1})}^{l}} 
      && \text{bound on $\Slow[k](b_{j_0})$}\\
      & ~~ + k \cdot (m+1) \cdot c_{k,p-1}^{\sum_{i=1}^{l-2} {(k \cdot d_{k,p-1})}^i} \cdot {(m+2)}^{\sum_{i=1}^{l-1} {(k \cdot d_{k,p-1})}^i}
      && \text{bound on $\Slow[k](b_{j})$, $j \not=j_0$}\\
      & < c_{k,p} \cdot {(m+1)}^{d_{k,p}} + c_{k,p} \cdot {(m + 2)}^{{(k \cdot d_{k,p-1})}^{l}} \\
      & ~~ + c_{k,p} \cdot (m + 1) \cdot {(m+2)}^{\sum_{i=1}^{l-1} {(k \cdot d_{k,p-1})}^i}
      && \text{as $k \cdot c_{k,p-1}^{\sum_{i=1}^{l-2} {(k \cdot d_{k,p-1})}^i} < c_{k,p}$}\\
      & \leqslant c_{k,p} \cdot {(m+1)}^{d_{k,p}} + c_{k,p} \cdot {(m+2)}^{\sum_{i=0}^{l} {(k \cdot d_{k,p-1})}^i} \\
      & \leqslant c_{k,p} \cdot {(m+1)}^{d_{k,p}} + c_{k,p} \cdot {(m+2)}^{{(k \cdot d_{k,p-1})}^{l+1}} \tkom 
    \end{align*}
    as desired, we conclude Equation~\eqref{t:pop:a}.
    From this preparatory step, ${\gpopv[k]} = {\gpopv[k][k]}$ 
    and ${(k \cdot d_{k,p-1})}^{k+1} < {(k \cdot d_{k,p-1})}^{k+1} + 1 = d_{k,p}$ we finally get
    \begin{align*}
      \Slow[k](b) & \leqslant c_{k,p+1}  \cdot {(m+1)}^{d_{k,p}} + c_{k,p} \cdot {(m + 2)}^{{(k \cdot d_{k,p-1})}^{k+1}} \\
      & = c_{k,p} \cdot ({(m+1)}^{d_{k,p}} + {(m + 2)}^{{(k \cdot d_{k,p-1})}^{k+1}}) 
      < c_{k,p} \cdot {(m + 2)}^{d_{k,p}}\tkom
    \end{align*}
    and conclude also this case.\qedhere
  \end{description}
\end{proof}

\noindent As a consequence, the number of $\gpopv[k]$-descents on 
basic terms interpreted with predicative interpretation $\ints$ 
is polynomial in sum of depths of normal arguments.

\begin{corollary}\label{c:pop}
  Let $k \geqslant 1$ and consider 
  $f \in \DS$ with $m \leqslant k$ normal arguments.
  There exists a constant $d \in \N$ depending only on $k$ and 
  the rank of $f$ such that:
  $$\Slow[k](\ints(f(\pseq[m][n]{u}))) = \bigO\bigl({(\textstyle{\max}_{i=1}^m \depth(u_i))}^{d}\bigr)$$
  for all $\seq[m+n]{u} \in \Val$.
\end{corollary}
\proof
  Let $s = f(\pseq[m][n]{u})$ be as given by the corollary. 
  Recall that since arguments of $s$ are values, we have $\norm{u_i} = \depth(u_i)$ as indicated on page~\pageref{d:normonval}, 
  and further $\ints(u_i) = \nil$ holds for all $i = 1,\dots,m+n$.
  Thus
  \begin{equation*}
    \ints(s) = \lst{\fn(\natToSeq{\depth(u_1)}, \dots, \natToSeq{\depth(u_m)})} \tpkt
  \end{equation*}
  As $\Slow[k](\theconst)$ is constant, say $\Slow[k](\theconst) = c$, by Lemma~\ref{l:slowsum} we see that
  $\Slow[k](\natToSeq{\depth(u_i)}) = c \cdot \depth(u_i)$. 
  We conclude as
  \begin{align*}
    \Slow[k](\ints(s)) =
    & \Slow[k](\fn(\natToSeq{\depth(u_1)}, \dots, \natToSeq{\depth(u_m)})) 
    && \text{by Lemma~\ref{l:slowsum}} 
    \\
    & \leqslant \Fpop{k}{\rk(f)}\bigl(\MSlow{l}{k}(\natToSeq{\depth(u_1)}, \dots, \natToSeq{\depth(u_m)})\bigr) 
    && \text{by assumption $m \leqslant k$}
    \\
    & \leqslant \Fpop{k}{\rk(f)}\Bigl({\bigl(1+ \textstyle{\max}_{i=1}^m \Slow[k](\natToSeq{\depth(u_i)})\bigr)}^k\Bigr) 
    && \text{by Lemma~\eref{l:homo}{4}} 
    \\
    & \leqslant \Fpop{k}{\rk(f)}\Bigl({\bigl(c \cdot (1 + \textstyle{\max}_{i=1}^m \depth(u_i))\bigr)}^k\Bigr)
    && \text{using $\Slow[k](\natToSeq{\depth(u_i)}) \leqslant c \cdot \depth(u_i)$}
    \\
    & \in \bigO\Bigl({\bigl(c \cdot (1 + \textstyle{\max}_{i=1}^m \depth(u_i))\bigr)}^{k+{d_{k,\rk(f)}}}\Bigr)
    && \text{by Theorem~\ref{t:pop}} 
    \\
    & = \bigO\bigl(({\textstyle{\max}_{i=1}^m \depth(u_i)})^{k+d_{k,\rk(f)}}\bigr) \tpkt
    \rlap{\hbox to 217 pt{\hfill\qEd}}
  \end{align*}


\section{Predicative Embedding}\label{s:embed}

Fix a predicative recursive TRS $\RS$ and signature $\FS$, 
and let $\gpop$ be the polynomial path order underlying $\RS$ 
based on the (admissible) precedence $\qp$.
We denote by $\qp$ also the induced precedence on $\FSn$ 
given by: $\fn \ep \gn$ if $f \ep g$ and $\fn \sp \gn$ if $f \sp g$.
Further, we set $f \sp \theconst$ for all $\fn \in \FSn$.
We denote by $\gpopv[\ell]$ (and respectively $\gppv[\ell]$) the approximation 
given in Definition~\ref{d:gpopv} (respectively Definition~\ref{d:gppv}) with underlying precedence $\qp$.

In this section, we establish the embedding of $\irew$ into $\gpopv[\ell]$ as
outlined in the proof plan on page~\pageref{popstar:proofplan}; in the sequel
$\ell$ is set to twice the maximum size of right-hand sides of $\RS$.
Lemma~\ref{l:embed:root} below proves the embedding of root steps 
for the case $l \gpop r$. In Lemma~\ref{l:embed:ctx} we then show that the embedding is closed under contexts.
The next auxiliary lemma connects the auxiliary orders $\gsq$ and $\gppv[k][l]$ (compare Example~\ref{ex:gppv}).

\begin{lemma}\label{l:embedgsq:root}
  Suppose $s = f(\pseq{s}) \in \BASICS$, $t \in \TERMS$ and $\ofdom{\sigma}{\VS \to \NF}$. 
  Then for predicative interpretation $\intq \in \{\ints, \intn\}$
  we have 
  $$
   s \gsq t \quad \IImp \quad \fn(\intn(s_1\sigma), \dots, \intn(s_k\sigma)) \gppv[2\cdot\size{t}] \intq(t\sigma) \tpkt
  $$
\end{lemma}
\begin{proof}
The proof proceeds by induction on the definition of $\gsq$ and makes
use of Lemmas~\ref{l:gpop:val},~\ref{l:int:len} and~\ref{l:approx}.
\end{proof}

\begin{lemma}\label{l:embed:root}
  Suppose $s = f(\pseq{s}) \in \BASICS$, $t \in \TERMS$ and $\ofdom{\sigma}{\VS \to \NF}$. 
  Then for predicative interpretation $\intq \in \{\ints, \intn\}$
  we have 
  $$
   s \gpop t \quad \IImp \quad \intq(s\sigma) \gpopv[2\cdot\size{t}] \intq(t\sigma) \tpkt
  $$
\end{lemma}
\begin{proof}
  Let $s$, $t$, $\sigma$ be as given in the lemma.
  We prove the stronger assertions
  \begin{enumerate}[labelsep=*,leftmargin=*]
  \item \label{l:er1} $\fn(\intn(s_1\sigma), \dots, \intn(s_k\sigma)) \gpopv[2\cdot\size{t}] \ints(t\sigma)$, 
  \item \label{l:er2} $\fn(\intn(s_1\sigma), \dots, \intn(s_k\sigma)) \gppv[2\cdot\size{t}] \ints(t\sigma)$ if $t \in \termsbelow$, and
  \item \label{l:er3} $\fn(\intn(s_1\sigma), \dots, \intn(s_k\sigma)) \append \NM{s\sigma} \gpopv[2\cdot\size{t}] \intn(t\sigma)$.
  \end{enumerate}
As $\ints(s) = \lst{\fn(\intn(s_1\sigma), \dots, \intn(s_k\sigma)}$, 
Property~\ref{l:er1} and Lemma~\eref{l:approx}{bound}
yield  $\ints(s\sigma) \gpopv[2\cdot\size{t}] \ints(s\sigma)$.
Furthermore
$\intn(s) = \ints(\fn(\intn(s_1\sigma), \dots, \intn(s_k\sigma)) \append \NM{s\sigma}$. 
Hence Property~\ref{l:er3} immediately
yields $\intn(s\sigma) \gpopv[2\cdot\size{t}] \intn(s\sigma)$.

We continue with the proof of the assertions by induction on $\gpop$ 
and set $u \defsym \fn(\intn(s_1\sigma), \dots, \intn(s_k\sigma))$.
  \begin{description}[leftmargin=0.3cm]
  \item[{\dcase{$s \cpop{st} t$}}]
Due to Lemma~\ref{l:gpop:val} (employing ${\gsq} \subseteq {\gpop}$) 
we obtain that $t\sigma$ is a safe subterm of $s_i\sigma$
and $t\sigma \in \NF$.
The latter implies $\ints(t\sigma) = \nil$ and thus
Properties~\ref{l:er1} and~\ref{l:er2} follow.
    For Property~\ref{l:er3}, observe that 
    $\len(\NM{t\sigma}) 
    = \norm{t\sigma} \leqslant \norm{s_i\sigma} 
    \leqslant \width(u \append \NM{s\sigma})$.
Here the last inequality follows by a simple case distinction on $i$.
%
    From this and $u \cppv{ia}[2\cdot\size{t} - 1] \theconst$ we get
    \begin{equation*}
      u \append \NM{s\sigma} \cpopv{ms}[2\cdot\size{t}] \NM{t\sigma} = \intn(t\sigma)
      \tpkt
    \end{equation*}

  \item[\dcase{$s \cpop{ia} t$}]
    The assumption gives $t = g(\pseq[m][n]{t})$ where $f \sp g$
    and further $s \gsq t_i$ for all normal argument positions $i = 1,\dots,m$, and
    $s \gpop t_i$ for all safe argument positions $i = m+1,\dots,m+n$, of $g$. 
    Additionally $t_{i_0} \not\in \termsbelow$ for at most one argument position $i_0$.
    Set $v \defsym \gn(\intn(t_1\sigma), \dots, \intn(t_m\sigma))$ and 
    let $\ints(t_i\sigma) = \lst{v_{i,1}~\cdots~v_{i,j_i}}$ for all safe argument positions 
    $i = m+1,\dots,m+n$. Hence, we obtain:
    \begin{equation*}
    \ints(t\sigma) =  \lst{\gn(\intn(t_1\sigma), \dots, \intn(t_m\sigma))~v_{m+1,1}~\cdots~v_{m+1,j_{m+1}}\quad\cdots\quad v_{m+n,1}~\cdots~v_{m+n,j_{m+n}}} \tpkt
    \end{equation*}
    Applying Lemma~\ref{l:embedgsq:root} on all normal arguments of $t$, we see
    \begin{equation}
      \label{e:c3:0}
      u \gppv[2\cdot\size{t} - 1] \gn(\intn(t_1\sigma), \dots, \intn(t_m\sigma)) = v
      \tkom
    \end{equation}
    from the assumptions $\fn \sp \gn$ and $s \gsq t_i$ for all $i = 1,\dots,m$.
    Since $s \gpop t_{i_0}$ by assumption, induction hypothesis on $i_0$ gives 
    %
    ${u} \gpopv[2\cdot \size{t_{i_0}}] {\ints(t_{i_0}\sigma)} = 
    {\lst{v_{i_0,1}, \dots, v_{i_0,j_{i_0}}}}$.
    %
    We obtain:
    \begin{align}
      u & \gpopv[2\cdot\size{t}-1] {v_{i_0,j_0}} 
      && \text{for some $j_0 \in \{1,\dots,j_{i_0}\}$}\label{e:c3:1}\\
      v & \gppv[2\cdot\size{t}-1] {v_{i_0,j}} 
      && \text{for all $j = 1,\dots,j_{i_0}$, $j \not=j_0$.}\label{e:c3:2}
    \end{align}
    Induction hypothesis on safe argument positions $i$ gives:
    \begin{align}
      u & \gppv[2\cdot\size{t}-1] {v_{i,j}} 
      && \text{for all $i = m+1,\dots,m+n$, $i \not=i_0$ and $j = 1,\dots,j_{i}$.} \label{e:c3:3}
    \end{align}
Due to Lemma~\eref{l:int:len}{S}, $\len(\ints(t\sigma)) \leqslant \size{t}$. 
Hence property~~\ref{l:er1} follows by $\cpopv{ialst}[2\cdot\size{t}]$
using equations \eqref{e:c3:0}--\eqref{e:c3:3}.
Likewise, Property~\ref{l:er3} follows by an additional use of
$u \cppv{ia}[2\cdot\size{t}-1] \theconst$ and
    \begin{align*}
      \len(\intn(t\sigma)) 
      & \leqslant 2\cdot\size{t} + \max \{ \norm{s_1\sigma}, \dots, \norm{s_{k+l}\sigma}\}
      \\
      & \leqslant 2\cdot\size{t} + \width(\fn(\intn(s_1\sigma), \dots, \intn(s_k\sigma)) \append \NM{s\sigma}) \tpkt
    \end{align*}
Here the first inequality follows by Lemma~\eref{l:int:len}{gpop}.
For Property~\ref{l:er2} we proceed as above, but strengthen inequality~\eqref{e:c3:1} 
to $u \gppv[2\cdot\size{t}-1] {v_{i_0,j_0}}$.


  \item[\dcase{$s \cpop{ep} t$}]
    Then $t = g(\pseq[m][n]{t})$ where $f \ep g$.
    Further, the assumption gives
    $\mset{\seq[k]{s}} \gpopmul \mset{\seq[m]{t}}$
    and $\mset{\seq[k+l][k+1]{s}} \geqpopmul \mset{\seq[m+n][m+1]{t}}$.
    Hence $t \not \in \termsbelow$ and Property~\ref{l:er2} is vacuously 
    satisfied. 
    We prove Properties~\ref{l:er1} and~\ref{l:er3}.
    Using $s_i \in \Val$ for all normal argument positions $i = 1,\dots,m$
    and employing Lemma~\ref{l:gpop:val} we see that
    $\mset{\seq[k]{s}} \gpopmul \mset{\seq[m]{t}}$ implies
    \begin{equation*}
      \mset{\intn(s_1\sigma), \dots, \intn(s_k\sigma)} \mextension{\gpopv[2\cdot\size{t}-1]} \mset{\intn(t_1\sigma), \dots, \intn(t_m\sigma)}
      \tpkt
    \end{equation*}
    Hence due to 
    $\fn \ep \gn$ and $m \leqslant \size{t} \leqslant 2\cdot\size{t} - 1$
    we obtain:
    \begin{equation}
      \label{eq:root:ep}
      \fn(\intn(s_1\sigma), \dots, \intn(s_k\sigma)) \cpopv{ep}[2\cdot\size{t} -1] \gn(\intn(t_1\sigma), \dots, \intn(t_m\sigma))
      \tpkt
    \end{equation}
Assumption $\mset{\seq[k+l][k+1]{s}} \geqpopmul \mset{\seq[m+n][m+1]{t}}$ together with
    $s_i \in \Val$ for all $i = k+1,\dots,k+l$ gives
    $t_j \in \Val$. As a consequence we have $\ints(t_j\sigma) = \nil$ for all
    $j=m+1,\dots,m+n$ and we obtain:
    $$
    \fn(\intn(s_1\sigma), \dots, \intn(s_k\sigma)) \cpopv{ialst}[2\cdot\size{t}] 
    \lst{\gn(\intn(t_1\sigma), \dots, \intn(t_m\sigma))} = \ints(t\sigma) 
    $$
    which concludes the argument for property~\ref{l:er1}.
    For property~\ref{l:er3}, 
    we see that the order constraints on safe arguments 
    imply $\norm{s\sigma} \geqslant \norm{t\sigma}$. 
    Thus $\NM{s\sigma} \geqpopv[2\cdot\size{t}] \NM{t\sigma}$, 
    using this and Equation~\eqref{eq:root:ep} we obtain 
    $$\fn(\intn(s_1\sigma), \dots, \intn(s_k\sigma)) \append \NM{s\sigma} 
    \cpopv{ms}[2\cdot\size{t}]  \gn(\intn(t_1\sigma), \dots, \intn(t_m\sigma)) \append \NM{t\sigma} = \intn(t\sigma)
    \tkom
    $$
    by Lemma~\eref{l:approx}{kmon} and Lemma~\eref{l:approx}{subseq}.\qedhere
  \end{description}
\end{proof}

\begin{lemma}
  \label{l:embed:ctx}
  Let $\ell \geqslant \max\{\ar(\fn) \mid \fn \in \FSn \}$ and $s,t \in \TERMS$. Then for $\intq \in \{\intn,\ints\}$, 
  $$
  \intq(s) \gpopv[\ell] \intq(t) \quad \IImp \quad \intq(C[s]) \gpopv[\ell] \intq(C[t]) \tpkt
  $$
\end{lemma}
\begin{proof}
We proceed by induction on the context $C$. 
It suffices to consider the inductive step.
Consider terms $s = f(s_1,\dots,s_i, \dots, s_{k+l})$ and 
$t = f(s_1,\dots,t_i, \dots, s_{k+l})$. We restrict our attention
the predicative interpretation $\intn$ and show $\intn(f(s_1,\dots,s_i, \dots, s_{k+l})) \gpopv[\ell] \intn(f(s_1,\dots,t_i, \dots, s_{k+l}))$, whenever 
$\intn(s_i) \gpopv[\ell] \intn(t_i)$. 

Recall that $\intn(s) = \ints(s) \append \NM{s}$ and $\intn(t) = \ints(t) \append \NM{t}$.
If $\norm{s} \geqslant \norm{t}$ then $\intn(s) \gpopv[l] \intn(t)$ follows 
from $\ints(s) \gpopv[l] \ints(t)$ and Lemma~\eref{l:approx}{subseq}.
Hence suppose $\norm{s} < \norm{t}$.
We consider only the case where $t \not \in \NF$.
The assumption $\norm{s} < \norm{t}$ 
implies that $i$ is a safe argument position of $f$. 
Hence we obtain:
\begin{align*}
        \intn(s) & = \lst{\fn(\intn(s_1), \dots, \intn(s_k))} 
                         \append \ints(s_{k+1}) \append \cdots \append \ints(s_i) \append \cdots \append \ints(s_{k+l}) \append \NM{s} 
\intertext{and} 
        \intn(t) & = \lst{\fn(\intn(s_1), \dots, \intn(s_k))} 
                         \append \ints(s_{k+1}) \append \cdots \append \ints(t_i) \append \cdots \append \ints(s_{k+l})  \append \NM{t} 
       \tpkt
\end{align*}
By definition we have $\norm{s_i} < \norm{s}$. This together
with $\norm{s} < \norm{t}$ yields $\norm{t} = \norm{t_i} + 1$ 
by the shape of $s$ and $t$.
Using Lemma~\eref{l:approx}{subseq} and the assumption $\intn(s_i) \gpopv[\ell] \intn(t_i)$
we obtain:
\begin{equation*}
        \ints(s_i) \append \NM{s_i} \append \theconst \gpopv[\ell] \ints(t_i) \append \NM{t_i} \append \theconst 
        \tpkt
\end{equation*}
From this we have $\ints(s_i) \append \NM{s} \gpopv[\ell] \ints(t_i) \append \NM{t}$
and thus due to Lemma~\eref{l:approx}{modeqi}
and Lemma~\eref{l:approx}{subseq} we obtain 
$\intn(s) \gpopv[\ell] \intn(t)$.\qedhere
\end{proof}

We have established our first main result.
\begin{proof}[Proof of Theorem~\ref{t:popstar}]
  Let $\RS$ be a predicative recursive TRS and fix an 
  arbitrary basic term $s = f(\pseq[m][n]{u})$. 
  Set the parameter $\ell$ as follows:
  \begin{equation*}
    \ell \defsym \max\{\ar(\fn) \mid \fn \in \FSn \} \cup \{2\cdot\size{r} \mid {l \to r} \in \RS\} \tpkt
  \end{equation*}
  As $\FSn$ and $\RS$ are finite, $\ell$ is well-defined.
  Consider a maximal $\RS$-derivation
  $$
  \mparbox[c]{1cm}{s} 
  \mparbox[c]{1cm}{\irew[\RS]} \mparbox[c]{1cm}{s_1} 
  \mparbox[c]{1cm}{\irew[\RS]} \mparbox[c]{1cm}{s_2} 
  \mparbox[c]{1cm}{\irew[\RS]} \mparbox[c]{1cm}{\cdots}
  \mparbox[c]{1cm}{\irew[\RS]} \mparbox[c]{1cm}{s_k\tkom} 
  $$
  starting from an arbitrary term $s$, that is, $k = \dheight[\RS](s)$.
  Using Lemma~\ref{l:embed:root} together with Lemma~\ref{l:embed:ctx} $k$-times 
  we get 
  $$
   \mparbox[c]{1cm}{\ints(s)} 
   \mparbox[c]{1cm}{\gpopv[\ell]} \mparbox[c]{1cm}{\ints(s_1)} 
   \mparbox[c]{1cm}{\gpopv[\ell]} \mparbox[c]{1cm}{\ints(s_2)}
   \mparbox[c]{1cm}{\gpopv[\ell]} \mparbox[c]{1cm}{\cdots}
   \mparbox[c]{1cm}{\gpopv[\ell]} \mparbox[c]{1cm}{\ints(s_k)}
   \tpkt
  $$
  As a consequence, we have $k \leqslant \Slow[\ell](\ints(s))$ by definition
of $\Slow[\ell]$ and thus:
\begin{equation*}
\dheight(s, \irew) \leqslant \Slow[\ell](\ints(s)) 
= \bigO\bigl((\max_{i=1}^m \depth(u_i))^{d}\bigr)
\tkom
\end{equation*}
where the asymptotic estimation follows by Corollary~\ref{c:pop}.
Note that the degree $d$ depends only on~$\ell$.
\end{proof}


\section{An Order-Theoretic Characterisation of the Polytime Functions}\label{s:icc}

We now present the application of polynomial path orders
in the context of \emph{implicit computational complexity}.
As by-product of Proposition~\ref{p:invariance} and Theorem~\ref{t:popstar} we immediately obtain
that $\POPSTAR$ is \emph{sound} for $\FNP$ respectively $\FP$.
\begin{theorem}\label{t:icc:soundness}
  Let $\RS$ be a predicative recursive (constructor) TRS.\@ 
  For every relation $\sem{f}$ defined by $\RS$, 
  the function problem $F_f$ associated with $\sem{f}$ is in $\FNP$.
  Moreover, if $\RS$ is confluent then $\sem{f} \in \FP$.
\end{theorem}

Although it is decidable whether a TRS $\RS$ is predicative recursive (we 
present a sound and complete automation in Section~\ref{s:exps}), 
confluence is undecidable in general. To get a decidable result for $\FP$, 
one can replace confluence by an decidable criteria, for instance orthogonality. 

We will now also establish that \POPSTAR\ is \emph{complete} for $\FP$, that is, 
every function $f \in \FP$ is computed by some confluent (even orthogonal) 
predicative recursive TRS.\@
For this we employ Beckmann and Weiermann's \emph{term rewriting characterisation} of the 
Bellantoni and Cook's class $\B$.

\begin{definition}{\cite[Definition~2.2]{BW96}}
\label{d:Rb}
For each $k,l \in \N$ the set of function symbols $\Fb^{k,l}$ 
with $k$ normal and $l$ safe argument positions is the 
least set of function symbols such that
\begin{enumerate}[labelsep=*,leftmargin=*]
\item $\epsilon \in \Fb^{0,0}$, $\mS_1,\mS_2 \in \Fb^{0,1}$, $\m{P} \in \Fb^{0,1}$, $\m{C} \in \Fb^{0,4}$
  and $\m{I}^{k,l}_j, \m{O}^{k,l} \in \Fb^{k,l}$, where $j = 1, \dots, k+l$; 
\item if $\vec{r} = \seq[m]{r} \in \Fb^{k,0}$, $\vec{s} = \seq[n]{s} \in \Fb^{k,l}$
  and $h \in \Fb^{m,n}$ then $\m{SC}[h,\vec{r}, \vec{s}] \in \Fb^{k,l}$; 
\item if $g \in \Fb^{k,l}$ and $h_1,h_2 \in \Fb^{k+1,l+1}$ then $\m{SRN}[g,h_1,h_2] \in \Fb^{k+1,l}$; 
\end{enumerate}
The \emph{predicative signature} is given by $\Fb \defsym \bigcup_{k,l \in \N} \Fb^{k,l}$.
Only the constant $\epsilon$ and \emph{dyadic successors} $\mS_1,\mS_2$, 
which serve the purpose of encoding natural numbers in binary, 
are constructors. The remaining symbols from $\Fb$ are defined symbols.
\end{definition}

In Figure~\ref{fig:1} we recall from~\cite[Definition~2.7]{BW96} 
the (infinite) schema of rewrite rules $R_\B$ that 
form a term rewriting characterisation of the class $\B$.
Here we let $k,l$ range over $\N$ and set
$\vec{x} \defsym \seq[k]{x}$ and 
$\vec{y} \defsym \seq[l]{y}$ for $k$ respectively $l$ distinct variables. 

\begin{figure}[ht]
\begin{tabular}{@{\quad}r@{~}c@{~}l@{\hspace{-10mm}}r}
  \multicolumn{4}{@{}l}{\textbf{Initial Functions}}\\[2mm]
  $\m{P}(\sn{}{\epsilon})$ & $\to$ & $\epsilon$ \\[1mm]
  $\m{P}(\sn{}{\mS_i(\sn{}{x})})$ & $\to$ & $x$ &  for $i=1,2$ \\[1mm]
  $\m{I}^{k,l}_j(\svec{x}{y})$ & $\to$ & $ x_j$ & for all $j = 1,\dots,k$ \\[1mm]
  $\m{I}^{k,l}_j(\svec{x}{y})$ & $\to$ & $y_{j-k}$ & for all $j =k+1, \dots, l+k$ \\[1mm]
  $\m{C}(\sn{}{\epsilon, y, z_1, z_2})$ & $\to$ & $y$ \\[1mm]
  $\m{C}(\sn{}{\mS_i(\sn{}{x}), y, z_1, z_2})$ & $\to$ & $z_i$ & for $i = 1, 2$ \\[1mm]
  $\m{O}(\svec{x}{y})$ & $\to$ & $\epsilon$ 
  \\[3mm]
  \multicolumn{4}{@{}l}{\textbf{Safe Composition} ($\m{SC}$)} \\[2mm]
    $\m{SC}[h,\vec{r}, \vec{s}](\svec{x}{y})$ & $\to$ & $h(\sn{\vec{r}(\sn{\vec{x}}{})}{\vec{s}(\svec{x}{y})})$
  \\[3mm]
  \multicolumn{4}{@{}l}{\textbf{Safe Recursion on Notation} ($\m{SRN}$)} \\[2mm]
    $\m{SRN}[g,h_1,h_2](\sn{\epsilon, \vec x}{\vec y})$ & $\to$ & $ g(\svec{x}{y})$ \\[1mm]
    $\m{SRN}[g,h_1,h_2](\sn{\mS_i (\sn{}{z}), \vec x}{\vec y}) 
    $ & $\to$ & $h_i(\sn{z, \vec x}{\vec y, \m{SRN}[g,h_1,h_2](\sn{z, \vec x}{\vec y})})$ & for $i = 1, 2$
\end{tabular}
\caption{Term Rewriting Characterisation of the Class~$\B$}
\label{fig:1}
\end{figure}

\begin{remark}
We emphasise that the system $R_\B$ is called \emph{infeasible} in~\cite{BW96}.
Indeed $R_\B$ admits an exponential lower bound on the derivation height if
one considers full rewriting. This is induced by duplicating redexes as
explained already in Example~\ref{ex:dup} on page~\pageref{ex:dup}. 
However, this should rather be understood as a miss-configuration 
of the evaluation strategy, rather than a defect of the rewrite system. 
Indeed, in our completeness argument below, we exploit that
$R_\B$ is predicative recursive, thus the \emph{innermost} runtime complexity
is polynomial, as expected.
\end{remark}

We emphasise that the above rules are all orthogonal and 
the following proposition verifies that $R_\B$ generates only polytime computable functions.
\begin{proposition}\label{prop:Rf}{\cite[Lemma~5.2]{BW96}}
  Let $f \in \FP$. There exists a finite restriction $\RS_f \subsetneq R_\B$
  such that $\RS_f$ computes $f$.
\end{proposition}

We arrive at our completeness result.
\begin{theorem}\label{t:icc:completeness}
  For every $f \in \FP$ there exists a finite, orthogonal, and 
  predicative recursive (constructor) TRS $\RS_f$ that computes $f$.
\end{theorem}
\begin{proof}
  Take the finite TRS $\RS_f \subsetneq R_\B$ from Proposition~\ref{prop:Rf} that computes $f$.
  Obviously $\RS_f$ is orthogonal hence confluent. 

  It remains to verify that $\RS_f$ is compatible with some instance $\gpop$.
  To define $\gpop$ we use the separation of normal from safe argument positions
  as indicated in the rules.
  To define the precedence underlying $\gpop$, we 
  define a mapping $\lh$ from the signature of $\Fb$ 
  into the natural numbers as follows:
  \begin{itemize}
  \item $\lh(f) \defsym 0$ if $f$ is one of $\epsilon$, $\mS_0$, $\mS_1$, $\m{C}$, $\m{P}$, $\m{I}^{k,l}_j$ or $\m{O}^{k,l}$;
  \item $\lh(\m{SC}[h,\vec{r}, \vec{s}]) \defsym 1 + \lh(h) + \sum_{r \in \vec{r}} \lh(r) + \sum_{s \in \vec{s}}  \lh(s)$;
  \item $\lh(\m{SRN}[g,h_1,h_2]) \defsym 1 + \lh(g) + \lh(h_1) + \lh(h_2)$. 
  \end{itemize}
  Finally for each pair of function symbol $f$ and $g$ occurring in $\RS_f$, we set
  $f \sp g$ if $\lh(f) > \lh(g)$.
  Then $\sp$ defines an admissible precedence. 
 
 It is straight forward to verify that $\RS_f \subseteq {\gpop}$ where
  $\gpop$ is based on the precedence~$\sp$ and the safe mapping as indicated in
  Definition~\ref{d:Rb}.
\end{proof}

Observe that compatibility of $R_\B$ with \POPSTAR\ together with 
Theorem~\ref{t:popstar} yields a 
strengthened version of Theorem~4.3 in~\cite{BW96},
as due to our result the innermost derivation height is polynomially
bounded in the depth of the normal arguments only. The latter result can 
be obtained directly, by a simplification of 
the semantic argument given in~\cite[Section~4]{BW96}, see~\cite{AM04}.

By Theorem~\ref{t:icc:soundness} and Theorem~\ref{t:icc:completeness} we obtain a
precise characterisation of the class of polytime computable functions and
thus arrive at the second main result of the paper.

\begin{corollary}\label{c:FP}
  The following class of functions are equivalent:
  \begin{enumerate}[labelsep=*,leftmargin=*]
  \item The class of functions computed by confluent predicative recursive (constructor) TRSs.
  \item The class of polytime computable functions $\FP$. 
  \end{enumerate}
\end{corollary}
We note that it is not decidable whether a rewrite system is confluent. However, 
to get a decidable characterisation we could replace confluence by orthogonality, compare 
Theorem~\ref{t:icc:completeness}.


\section{A Non-Trivial Closure Property of the Polytime Computable Functions}\label{s:popstarps}

Bellantoni~\cite{B:92} already observed that the class $\B$ is closed under
\emph{predicative recursion on notation with parameter substitution} (scheme \eqref{scheme:srnps}). 
Essentially this recursion scheme allows substitution on \emph{safe} argument positions. More precise, 
a new function $f$ is defined by the equations
\begin{equation}
  \label{scheme:srnps} \tag{\ensuremath{\mathsf{SRN_{PS}}}}
  \begin{aligned}
   f(\sn{0,\vec{x}}{\vec{y}}) & = g(\sn{\vec{x}}{\vec{y}}) \\
   f(\sn{2z + i,\vec{x}}{\vec{y}}) & = 
   h_i(\sn{z,\vec{x}}{\vec{y},f(\sn{z,\vec{x}}{\vec{p}(\svec{x}{y})})}) 
   \quad i \in \set{1,2} \tpkt
  \end{aligned}
\end{equation}

Bellantoni's result has been reobtained by Beckmann and 
Weiermann~\cite[Corollary~5.4]{BW96} employing a similar
term rewriting characterisation. In this section, we introduce
the \emph{polynomial path order with parameter substitution} (\emph{\POPSTARP} for short).
\POPSTARP\ provides an order-theoretic characterisation of
predicative recursion with parameter substitution, that again precisely
captures the class $\FP$. Furthermore \POPSTARP\ induces polynomial
innermost runtime complexity. As a consequence, we obtain yet another proof
of Bellantoni's result.

The next definition introduces $\POPSTARP$.
It is a variant of $\POPSTAR$, where 
clause $\cpop{ep}$ has been modified and allows computation at safe argument positions.
\begin{definition}\label{d:gpopps}
  Let ${\qp}$ denote a precedence.  
  Consider terms $s, t \in \TERMS$ such that $s = f(\pseq[k][l]{s})$.
  Then $s \gpopps t$ if one of the following alternatives holds:
  \begin{enumerate}[labelsep=*,leftmargin=*]
  \item\label{d:gpopps:st} $s_i \geqpopps t$ for some $i \in \{1,\dots,k+l\}$, or
  \item\label{d:gpopps:ia} $f \in \DS$, $t = g(\pseq[m][n]{t})$ where $f \sp g$ 
    and the following conditions hold:
    \begin{itemize}
    \item $s \gsq t_j$ for all normal argument positions $j = 1,\dots,m$;
    \item $s \gpopps t_j$ for all safe argument positions $j = m+1,\dots,m+n$;
    \item $t_j \not\in \TA(\sigbelow{f}{\FS},\VS)$ for at most one safe argument position $j \in \{m+1,\dots,m+n\}$;
    \end{itemize}
  \item\label{d:gpopps:ep} $f \in \DS$, $t = g(\pseq[m][n]{t})$ where $f \ep g$
    and the following conditions hold:
    \begin{itemize}
    \item $\mset{\seq[k]{s}} \gpopmulps \mset{\seq[m]{t}}$;
    \item $s \gpopps t_j$ and $t_j \in \TA(\sigbelow{f}{\FS},\VS)$ for all safe argument positions $j = m+1, \dots, m+n$.
    \end{itemize}
  \end{enumerate}
  Here ${\geqpopps} \defsym {\gpopps \cup \eqis}$.
\end{definition}

The next lemma shows that $\POPSTARP$ extends the analytic power of $\POPSTAR$.
\begin{lemma}\label{l:psextends}
  For any underlying admissible precedence $\qp$, ${\gpop} \subseteq {\gpopps}$.
\end{lemma}
Note that \POPSTARP~is strictly more powerful than \POPSTAR, as witnessed by 
the following example.

\begin{example}\label{ex:rsrev}
Consider the constructor TRS $\RSrev$ defining the reversal of
lists in a tail recursive fashion:
\begin{alignat*}{4}
  \rlabel{RSrev:revt:b} &\;& \mrevt(\sn{\nil}{ys}) & \to ys
  & \qquad
  \rlabel{RSrev:revt:r} &&\mrevt(\sn{\mcs(x, xs)}{ys}) & \to \mrevt(\sn{xs}{\mcs(x,ys)}) 
  \\
  \rlabel{RSrev:rev}&& \mrev(\sn{xs}{}) & \to \mrevt(\sn{xs}{\mnil}) 
  \tpkt
\end{alignat*}
It is not difficult to see that $\RSrev$ is compatible with \POPSTARP,
if we use the precedence $\mrev \sp \mrevt \sp \mnil \ep \mcs$.
Note that orientation of rule~\rref{RSrev:revt:r}
breaks down to $\mcs(x, xs) \gpopps xs$ and 
$\mrevt(\sn{\mcs(x, xs)}{ys}) \gpopps \mcs(x,ys)$.
On the other hand, $\gpop$ fails as the corresponding clause $\cpop{ep}$ 
requires $ys \geqpop \mcs(x,ys)$.
\end{example}

Due to Lemma~\ref{l:psextends}, \POPSTARP\ is complete for the class of
polytime computable functions.
To show that it is sound, we prove that 
\POPSTARP~induces polynomially bounded runtime complexity in the sense of Theorem~\ref{t:popstar}.
The crucial observation is that the embedding of $\irew$ into $\gpopv$ does not break
if we relax compatibility constraints to $\RS \subseteq {\gpopps}$.

\begin{lemma}\label{l:embed:root:ps}
  Suppose $s = f(\pseq{s}) \in \Tb$, $t \in \TERMS$ and $\ofdom{\sigma}{\VS \to \Val}$. 
  Then for predicative interpretation $\intq \in \{\ints, \intn\}$
  we have 
  $$
   s \gpopps t \quad \IImp \quad \intq(s\sigma) \gpopv[2\cdot\size{t}] \intq(t\sigma) \tpkt
  $$
\end{lemma}
\proof
  First one verifies that 
  Lemma~\ref{l:int:len} holds even if we replace 
  $\gpop$ by $\gpopps$. 
  In particular, the assumptions give 
  \begin{equation}
    \label{l:embed:root:ps:len}
    \len(\intn(t\sigma)) \leqslant 2\cdot\size{t} + \max \{ \norm{s_1\sigma}, \dots, 
    \norm{s_{k+l}\sigma}\}
    \tpkt
  \end{equation}
  The proof follows the pattern of the proof of Lemma~\ref{l:embed:root}, i.e.,
  we proceed by induction on $\gpopps$. 

  We cover only the new case $s \cpopps{ep} t$.
  Let $s$, $t$, and $\sigma$ be as given in the lemma.
    Then $t = g(\pseq[m][n]{t})$ where $f \ep g$.
    Further, the assumption gives
    $\mset{\seq[k]{s}} \gpopmul \mset{\seq[m]{t}}$.
    As $t \not \in \termsbelow$ 
    it suffices to verify Property~\ref{l:er1} and Property~\ref{l:er3}
    from Lemma~\ref{l:embed:root}.
    As before, we obtain:
    \begin{equation}
      \label{eq:root:ep:ps}
      \fn(\intn(s_1\sigma), \dots, \intn(s_k\sigma)) \cpopv{ep}[2\cdot\size{t} -1] \gn(\intn(t_1\sigma), \dots, \intn(t_m\sigma)) \tpkt
    \end{equation}
    By assumption $s \gpopps t_j$ and $t_j \in \termsbelow$, induction hypothesis gives 
    \begin{equation}
      \label{eq:ti:ep:ps}
      \fn(\intn(s_1\sigma), \dots, \intn(s_k\sigma)) \gppv[2\cdot\size{t} -1] \ints(t_j\sigma) \tpkt
    \end{equation}
    As $\len(\ints(t\sigma)) \leqslant \size{t}$ by Lemma~\eref{l:int:len}{S}, 
    we obtain $\fn(\intn(s_1\sigma), \dots, \intn(s_k\sigma)) \cpopv{ialst}[2\cdot\size{t}] \ints(t\sigma)$
    from equations~\eqref{eq:root:ep:ps} and~\eqref{eq:ti:ep:ps}.
    Likewise, from this assertion~\ref{l:er3} follows by $\cpopv{ms}$ using additionally 
    $\fn(\intn(s_1\sigma), \dots, \intn(s_k\sigma)) \cppv{ia}[2\cdot\size{t}-1] \theconst$
    and
    \begin{align*}
      \len(\intn(t\sigma)) 
      & \leqslant 2\cdot\size{t} + \max \{ \norm{s_1\sigma}, \dots, \norm{s_{k+l}\sigma} \}
      && \text{by Equation~\eqref{l:embed:root:ps:len}} \\
      & \leqslant 2\cdot\size{t} + \width(\fn(\intn(s_1\sigma), \dots, \intn(s_k\sigma)) \append \NM{s\sigma}) \tpkt\rlap{\hbox to 129 pt{\hfill\qEd}}
    \end{align*}\vspace{-3 pt}

\noindent Following the pattern of the proof of Theorem~\ref{t:popstar}, replacing
the use of Lemma~\ref{l:embed:root} by Lemma~\ref{l:embed:root:ps} we obtain:
\begin{theorem}\label{t:popstarps}
  Let $\RS$ be a constructor TRS compatible with an instance of $\POPSTARP$.
  Then the innermost derivation height of any basic term 
  $f(\svec{u}{v})$ is bounded by a polynomial in the 
  maximal depth of normal arguments $\vec{u}$.
  The polynomial depends only on $\RS$ and the signature $\FS$.
\end{theorem}

As a corollary we get the following variant of~Theorem~5.3 in~\cite{BW96}.
\begin{corollary}
Let $\RS$ be the rewrite system based on the defining equations from Figure~\ref{fig:1} 
and the Schema~\eqref{scheme:srnps}. Then 
the innermost derivation height of any basic term $f(\svec{u}{v})$ 
is bounded by a polynomial in the maximal depth of normal arguments $\vec{u}$.
\end{corollary}
\begin{proof}
By construction there exists an instance $\gpopps$ such that ${\RS} \subseteq \gpopps$.
Thus by the theorem, the result follows.  
\end{proof}

Applying Proposition~\ref{p:invariance} yields that predicative recursion is closed under parameter substitution.
We can even show a stronger result from Theorem~\ref{t:popstarps}.

\begin{corollary}
  Let $\RS$ be a constructor TRS compatible with an instance of $\POPSTARP$.\@ 
  For every relation $\sem{f}$ defined by $\RS$, 
  the function problem $F_f$ associated with $\sem{f}$ is in $\FNP$.
  Moreover, if $\RS$ is confluent than $\sem{f} \in \FP$.
\end{corollary}

By Lemma~\ref{l:psextends}, parameter substitution extends the power of $\POPSTAR$, 
together with Theorem~\ref{t:icc:completeness} that shows completeness of $\POPSTAR$, 
this shows completeness of \POPSTARP. We obtain our third result.\footnote{Again confluence can be replaced by orthogonality, as in Corollary~\ref{c:FP}.}
\begin{corollary}\label{c:fptime:ps}
  The following class of functions are equivalent:
  \begin{enumerate}[labelsep=*,leftmargin=*]
  \item The class of functions computed by confluent constructor TRS compatible with an instance of $\POPSTARP$. 
  \item The class of polytime computable functions $\FP$. 
  \end{enumerate}
\end{corollary}

\section{Automation of Polynomial Path Orders}\label{s:impl}

In this section we present an automation of polynomial path orders,  
for brevity we restrict our efforts to the order $\gpop$.
Consider a constructor TRS $\RS$. Checking whether $\RS$ is predicative 
recursive is equivalent to guessing a precedence $\qp$ and partitioning 
of argument positions so that $\RS \subseteq {\gpop}$ holds for the induces
order $\gpop$. 
As standard for recursive path orders~\cite{ZM07,SFTGACMZ07},
this search can be automated by encoding the constraints imposed by Definition~\ref{d:gpop}
into \emph{propositional logic}. 
To simplify the presentation, we extend the language of propositional 
logic with truth-constants $\top$ and $\bot$ in the obvious way.
In the constraint presented below we employ the following atoms. 

\subsection{Propositional Atoms}
To encode the separation of normal from safe arguments, we
introduce for $f \in \DS$ and $i = 1, \dots, \ar(f)$ the atoms $\esafe{f}{i}$
so that $\esafe{f}{i}$ represents the assertion that the \nth{$i$} argument position of $f$ is safe.
Further we set $\esafe{f}{i} \defsym \top$ for $n$-ary $f \in \CS$ and $i = 1,\dots,n$,
reflecting that argument positions of constructors are always safe. 

Since $\POPSTAR$ is blind on constructors, predicative 
recursive TRSs are even compatible with  $\gpop$
as induced by an admissible precedence where constructors are equivalent and minimal.
For each pair of symbols $f,g \in \DS$, we introduce
propositional atoms $\esp{f}{g}$ and $\eep{f}{g}$ so that $\esp{f}{g}$
represents the assertion $f \sp g$, and likewise 
$\eep{f}{g}$ represents the assertion $f \ep g$.
Overall we define for function symbols $f$ and $g$ the propositional formulas
\begin{equation*}
  \enc{f \sp g} \defsym 
  \begin{cases}
    \top & \text{if $f\in\DS$ and $g\in \CS$,} \\
    \bot & \text{if $f\in\CS$ and $g\in \CS$,} \\
    \esp{f}{g} & \text{otherwise.}
  \end{cases}
  \quad
  \enc{f \ep g} \defsym
  \begin{cases}
    \top & \text{if $f\in\CS$ and $g\in \CS$, } \\
    \bot & \text{if $f\in\DS$ and $g \in \DS$,} \\
    \eep{f}{g} & \text{otherwise.}
  \end{cases}
\end{equation*}

To ensure that the variables $\esp{f}{g}$ and respectively $\eep{f}{g}$ encode a preorder on $\DS$ 
we encode an order preserving homomorphism into the natural order $>$ on $\N$. 
To this extent, to each $f \in \DS$ we associate a natural number $\rk_f$ encoded as binary string
with $\lceil \log_2(\size{\DS}) \rceil$ bits. 
It is straight forward to define Boolean formulas $\enc{\rk_f > \rk_g}$ (respectively $\enc{\rk_f = \rk_g}$) that are satisfiable iff 
the binary numbers $\rk_f$ and $\rk_g$ are decreasing (respectively equal) in the natural order. Using these we set
\begin{equation*}
  \vprec(\DS)
  \defsym \bigwedge_{f,g \in \DS} (\enc{f \sp g} \imp \enc{\rk_f > \rk_g})
  \wedge \bigwedge_{f,g \in \DS} (\enc{f \ep g} \imp \enc{\rk_f = \rk_g})
  \tpkt
\end{equation*}

We say that a propositional assignment $\mu$ \emph{induces} the precedence
$\qp$ if $\mu$ satisfies $\enc{f \sp g}$ when $f \sp g$ and $\enc{f \ep g}$ when $f \ep g$. 
The next lemma verifies that $\vprec$ serves our needs.
\begin{lemma}
  For any assignment $\mu$ that satisfies $\vprec(\DS)$, 
  $\mu$ induces an admissible precedence on $\FS$.
  Vice versa, for any admissible precedence $\qp$ on $\FS$, 
  any valuation $\mu$, satisfying $\mu(\enc{f \sp g})$ iff $f \sp g$ and 
  $\mu(\enc{f \sp g})$ iff $f \ep g$, also satisfies
  the formula $\vprec(\DS)$.
\end{lemma}

\subsection{Order Constraints}
For concrete pairs of terms $s = f(\seq{s})$ and $t$, we define
$$
\enc{s \gpop t} \defsym \enc{s \cpop{st} t} \vee \enc{s \cpop{ia} t} \vee \enc{s \cpop{ep} t}
$$
which enforces the orientation $f(\seq{s}) \gpop t$ using propositional formulations 
of the three clauses in Definition~\ref{d:gpop}. 
To complete the definition for arbitrary left-hand sides, we set $\enc{x \gpop t} \defsym \bot$ for all $x \in \VS$.
Further, weak orientation is given by
$$
\enc{s \geqpop t} \defsym \enc{s \gpop t} \vee \enc{s \eqis t} \tkom
$$
where the constraint $\enc{s \eqis t}$ refers to a formulation of 
Definition~\ref{d:eqis} in propositional logic, defined as follows.
For $s = t$ we simply set $\enc{s \eqis t} \defsym \top$. 
Consider the case $s = f(\seq{s})$ and $t = g(\seq{t})$. 
Then $s \eqis t$ if $f \ep g$ 
and moreover $s_i \eqis t_{\pi(i)}$ for all $i = 1,\dots,n$ and 
some permutation $\pi$ on argument positions that takes the separation of normal and safe positions into account.
To encode $\pi(i) = j$, we use fresh atoms $\pi_{i,j}$ for $i,j=1,\dots,n$. 
The propositional formula $\vperm(\pi,n) \defsym \bigwedge_{i=1}^{n} \eone(\pi_{i,1}, \dots, \pi_{i,n})$
is used to assert that the atoms $\pi_{i,j}$ reflect a permutation on $\{1,\dots,n\}$. 
Here $\eone(\pi_{i,1}, \dots, \pi_{i,n})$ expresses that exactly one of its arguments evaluates to $\top$.
We set
\begin{equation*}
  \enc{s \eqis t} \defsym 
    \enc{f \ep g} \wedge \vperm(\pi,n)
      \wedge~\bigl({\bigwedge_{j=1}^{n} \pi_{i,j} \imp \enc{s_i \eqis t_j} \wedge (\esafe{f}{i} \iff \esafe{g}{j})}\bigr) \tpkt
\end{equation*}
To complete the definition, we set $\enc{s \eqis t} = \bot$ for 
the remaining cases.
\begin{lemma}
  Suppose the assignment $\mu$ induces an admissible precedence $\qp$ and 
  $\mu$ satisfies $\enc{s \eqis t}$. 
  Then $s \eqis t$ with respect to the precedence $\qp$.
  Vice versa, if $s \eqis t$ then $\enc{s \eqis t}$ is satisfiable by assignments $\mu$ 
  that induce the precedence underlying $\eqis$.
\end{lemma}

We now define the encoding for the different cases underlying the definition of $\gpop$.
Assuming that $\enc{s_i \geqpop t}$ enforces $s_i \gpop t$ clause $\cpop{st}$ is expressible
as 
\begin{equation*}
  \enc{f(\seq{s}) \cpop{st} t} \defsym \bigor_{i=1}^n \enc{s_i \geqpop t}
  \tpkt
\end{equation*}
in propositional logic. 
For clause $\cpop{ia}$ we use propositional atoms $\alpha_i$ ($i = 1,\dots,m$)
to mark the unique argument position of $t = g(\seq[m]{t})$ that allows
$t_i \not\in \termsbelow$. 
The propositional formula $\ezeroone(\seq[m]{\alpha})$ expresses that zero or one
$\alpha_i$ valuates to $\top$. 
Further, we introduce the auxiliary constraint 
\begin{equation*}
\enc{g(\seq[m]{t}) \in \termsbelow[f]} \defsym \enc{f \sp g} \wedge \bigwedge_{j=1}^m \enc{t_j \in \termsbelow[f]}
\tpkt
\end{equation*}
and $\enc{x \in \termsbelow[f]} \defsym \top$ for $x \in \VS$.
Using these, clause $\cpop{ia}$ becomes expressible as 
\begin{multline*}
  \enc{f(\seq{s}) \cpop{ia} g(\seq[m]{t})} \defsym
  \enc{f \in \DS}
  \wedge \enc{f \sp g} \\
  \wedge \bigwedge_{j=1}^m (\esafe{g}{j} \imp \enc{s \gpop t_j})
  \wedge \bigwedge_{j=1}^m (\neg \esafe{g}{j} \imp \enc{s \gsq t_j}) \\
  \wedge \ezeroone(\seq[m]{\alpha}) 
  \wedge \bigwedge_{j=1}^m (\neg \alpha_j \imp \enc{t_j \in \termsbelow}) \tpkt
\end{multline*}
Here $\enc{f \in \DS} = \top$ if $f \in \DS$ and otherwise $\enc{f \in \DS} = \bot$.
The propositional formula $\enc{s \gsq t}$ expresses the orientation with the $\gsq$ and is given by
\begin{align*}
  \enc{f(\seq{s}) \gsq t} \defsym \enc{f(\seq{s}) \csq{st} t} \vee \enc{f(\seq{s}) \csq{ia} t}
\end{align*}
and otherwise $\enc{x \gsq t} = \bot$, where
\begin{align*}
  \enc{f(\seq{s}) \csq{st} t} & \defsym \bigor_{i=1}^n ((\enc{s_i \gsq t} \vee \enc{s_i \eqis t}) \wedge (\enc{f \in \DS} \imp \neg \esafe{f}{i})) \\
  \enc{f(\seq{s}) \csq{ia} t} & \defsym 
  \begin{cases}
    \enc{f \in \DS} \wedge \enc{f \sp g} & \text{ if $t = g(\seq[m]{t})$} \\
    \quad \wedge \bigwedge_{j=1}^m \enc{f(\seq{s}) \gsq t_j} \\
    \bot & \text{ if $t \in \VS$}.
  \end{cases}
\end{align*}
This concludes the propositional formulation of clause $\cpop{ia}$.

The main challenge in formulating clause $\cpop{ep}$
is to deal with the encoding of multiset-comparisons. 
We proceed as in~\cite{SK07} and encode the underlying \emph{multiset cover}.
\begin{definition}
Let $\succ_\mul$ denote the multiset extension of a binary relation ${\succcurlyeq} = {\succ} \uplus {\eqi}$.
Then a pair of mapping $(\gamma, \varepsilon)$ 
where $\ofdom{\gamma}{\set{1,\dots,m} \to \set{1,\dots,n}}$ 
and $\ofdom{\varepsilon}{\set{1,\dots,n} \to \set{\top,\bot}}$
is a multiset cover on multisets $\mset{\seq{a}}$ and $\mset{\seq[m]{b}}$
if the following holds for all $j \in \{1,\dots,m\}$:
\begin{enumerate}[labelsep=*,leftmargin=*]\label{d:mscover}
\item\label{d:mscover:1} if $\gamma(j) = i$ then $a_i \succcurlyeq b_j$, in this case we say that $a_i$ \emph{covers} $b_j$; 
\item\label{d:mscover:2} if $\varepsilon(j) = \top$ then $s_{\tau(j)} \eqi t_j$ and $\tau$ is invective on $\{j\}$, 
  i.e., $a_{\tau(j)}$ covers only $b_j$.
\end{enumerate}
The multiset cover $(\gamma, \varepsilon)$ is said to be \emph{strict} if at least one cover is strict, 
i.e., $\varepsilon(j) = \bot$ for some $j \in \{1,\dots,m\}$.
\end{definition}

It is straight forward to verify that multiset covers characterise the multiset extension
of $\succ$ in the following sense.
\begin{lemma}
  We have $\mset{\seq{a}} \mextension{\succcurlyeq} \mset{\seq[m]{b}}$ if and only if there 
  exists a multiset cover $(\gamma, \varepsilon)$ on $\mset{\seq{a}}$ and $\mset{\seq[m]{b}}$.
  Moreover, $\mset{\seq{a}} \mextension{\succ} \mset{\seq[m]{b}}$ if and only if the cover is strict.
\end{lemma}

Consider the orientation $f(\seq{s}) \cpop{ep} g(\seq[m]{t})$. 
Then normal arguments are strictly, and safe arguments weakly decreasing with 
respect to the multiset-extension of $\gpop$. 
Since the partitioning of normal and safe argument is not fixed, 
in the encoding of $\cpop{ep}$ we formalise a multiset-comparison on \emph{all} arguments, 
where the underlying multiset-cover $(\gamma, \varepsilon)$ 
will be restricted so that if $s_i$ covers $t_j$, i.e., $\gamma(i) = j$, 
then both $s_i$ and $t_j$ are safe or respectively normal.
To this extend, for a specific multiset cover $(\gamma, \varepsilon)$ we introduce variables $\gamma_{i,j}$ and 
$\varepsilon_i$, where $\gamma_{i,j} = \top$ represents $\gamma(j) = i$ and
$\varepsilon_i = \top$ denotes $\varepsilon(i) = \top$ ($1 \leqslant i \leqslant n$, $1 \leqslant j \leqslant m$).
We set
\begin{multline*}
  \enc{f(\seq{s}) \cpop{ep} g(\seq[m]{t})} \defsym 
  \enc{f \in \DS}
  \wedge \enc{f \sp g} \\
  \wedge 
  \bigwedge_{i=1}^{n} \bigwedge_{j=1}^{m} \Bigl( \gamma_{i,j} \to \bigl( \varepsilon_i \to \enc{s_i \eqis t_j} \bigr)
                                             \wedge \bigl( \neg \varepsilon_i \to \enc{s_i \gpop t_j} \bigr)
                                             \wedge \bigl( \esafe{f}{i} \iff \esafe{g}{j} \bigr)
                                      \Bigr) \\
  \wedge \bigwedge_{j=1}^m \eone(\gamma_{1,j},\dots,\gamma_{n,j}) 
  \wedge \bigwedge_{i=1}^{n} \bigl(\varepsilon_i \to \eone(\gamma_{i,1},\dots,\gamma_{i,m})\bigr) 
  \wedge  \bigvee_{i=1}^n \bigl( \neg \esafe{f}{i} \wedge \neg \varepsilon_i \bigr)  \tpkt
\end{multline*}
Here the first line establishes the Condition~\eref{d:mscover}{1}, where
$\esafe{f}{i} \iff \esafe{g}{j}$ additionally enforces the separation of normal from safe arguments.
The final line formalises
that $\gamma$ maps $\{1,\dots,m\}$ to $\{1,\dots,n\}$, Condition~\eref{d:mscover}{2}
as well as the strictness condition on normal arguments.
This completes the encoding of $\gpop$.

\begin{lemma}
  Suppose $\mu$ induces an admissible precedence $\qp$ and satisfies $\enc{s \gpop t}$. 
  Then $s \gpop t$ with respect to the precedence $\qp$.
  Vice versa, if $s \gpop t$ then $\enc{s \gpop t}$ is satisfiable assignments $\mu$ 
  that induce the precedence underlying $\gpop$.
\end{lemma}

Putting the constraints together we get the following theorem,
which witnesses the fourth main result of this paper.
\begin{theorem}
  Let $\RS$ be a constructor TRS.\@ 
  The propositional formula 
  $$
  \vpredrec(\RS) \defsym \vprec(\DS) \wedge \bigwedge_{{l \to r} \in \RS} \enc{l \gpop r}
  \tkom
  $$
  is satisfiable if and only if $\RS$ is predicative recursive.
\end{theorem}

We have implemented this reduction to \SAT~in our complexity analyser \TCT.\@
As underlying \SAT-solver we employ the open source solver \minisat~\cite{ES03}.

\subsection{Efficiency Considerations}
The \SAT-solver \minisat\ requires its input in CNF.\@
For a concise translation of $\vpredrec(\RS)$ to CNF 
we use the approach of Plaisted and Greenbaum~\cite{PG86} that 
gives an equisatisfiable CNF linear in size.
Our implementation also eliminates redundancies resulting from 
multiple comparisons of the same pair of term $s, t$ by 
replacing subformulas $\enc{s \gpop t}$ with unique 
propositional atoms $\delta_{s,t}$. Since $\enc{s \gpop t}$ 
occurs only in positive contexts, it suffices to 
add $\delta_{s,t} \imp \enc{s \gpop t}$, resulting in an equisatisfiable formula.
Also during construction of $\vpredrec(\RS)$ our implementation
performs immediate simplifications under Boolean laws.


\section{Experimental Assessment}\label{s:exps}

In this section we present an empirical evaluation of polynomial path orders.
We selected two testbeds: Testbed~\textsf{TC} constitutes 
of 597 terminating constructor TRSs, obtained
by restricting the innermost runtime complexity problemset 
from the \emph{Termination Problem Database}\footnote{The TPDB is available online~\url{http://termcomp.uibk.ac.at/}.} 
(\emph{TPDB} for short), version 8.0,
to known to be terminating constructor TRSs.
Termination is checked against the data available from the termination competition.
Testbed~\textsf{TCO}, containing 290 examples, results from restricting Testbed~\textsf{TC} to 
orthogonal systems.
Unarguably the TPDB is an imperfect choice as examples were collected primarily to 
assess the strength of termination provers, but it is at the moment the only 
extensive source of TRSs. 

Experiments were conducted with $\TCT$ version 1.9.1,%
\footnote{Available from \url{http://cl-informatik.uibk.ac.at/software/tct/}.}
on a laptop with 4Gb of RAM and Intel${}^\text{\textregistered}$ Core${}^\text{\texttrademark}$ i7--2620M CPU (2.7GHz, quad-core).
We assess the strength of $\POPSTAR$ and $\POPSTARP$ in comparison to its predecessors $\MPO$ and $\LMPO$.\@
The implementation of $\MPO$, $\LMPO$ and $\POPSTARP$ follows the line of polynomial path orders 
as explained in Section~\ref{s:impl}.%
\footnote{As far as we know our implementation of $\LMPO$ in~\TCT\ is
the only implementation currently available.}
We contrast these syntactic techniques to \emph{interpretations}
as implemented in our complexity tool $\TCT$.\@
The last column show result of constructor restricted 
matrix interpretations~\cite{MMNWZ11} (dimension $1$ and $3$)
as well as polynomial interpretations~\cite{BCMT01} (degree $2$ and $3$), 
run in parallel on the quad-core processor.
We employ interpretations in their default configuration of \TCT, 
noteworthy coefficients (respectively entries in coefficients) 
range between $0$ and $7$, and we also make use of the \emph{usable argument positions} 
criterion~\cite{HM11} that weakens monotonicity constraints.
Table~\ref{tbl:exp1}
shows totals on systems that can respectively cannot be handled.%
\footnote{Full evidence available at \url{http://cl-informatik.uibk.ac.at/software/tct/experiments/popstar}.}
To the right of each entry we annotate the average execution time, in seconds.

\newcommand{\tm}[1]{\parbox[b]{9mm}{\bf{\tiny{$\backslash$#1}}}}
\renewcommand{\c}[1]{\parbox[b]{9mm}{{\hfill\small{#1}}}}
\begin{table}[h]
  \centering
  \begin{tabular}{l@{}l@{\quad}cccc@{\quad}c}
    \hline
    \TOP & 
    & \MPO 
    & \LMPO
    & \POPSTAR
    & \POPSTARP
    & interpretations
    \BOT
    \\
    \hline
    \textbf{TC} \TOP 
    & \textsf{compatible}
    & \c{76}\tm{0.33} 
    & \c{57}\tm{0.20} 
    & \c{43}\tm{0.18} 
    & \c{56}\tm{0.19} 
    & \c{139}\tm{2.77} 
    \\
    & \TOP \textsf{incompatible}
    & \c{521}\tm{0.58} 
    & \c{540}\tm{0.47} 
    & \c{554}\tm{0.42} 
    & \c{541}\tm{0.43} 
    & \c{272}\tm{6.47} 
    \\
    & \TOP\BOT \textsf{timeout}
    & --- 
    & --- 
    & --- 
    & --- 
    & \c{186}\tm{25.0} 
    \\
    \hline
    \textbf{TCO} \TOP 
    & \textsf{compatible}
    & \c{40}\tm{0.29} 
    & \c{29}\tm{0.16} 
    & \c{24}\tm{0.14} 
    & \c{29}\tm{0.15} 
    & \c{75}\tm{2.81} 
    \\
    & \TOP \textsf{incompatible}
    & \c{250}\tm{0.33} 
    & \c{261}\tm{0.27} 
    & \c{266}\tm{0.26} 
    & \c{261}\tm{0.27} 
    & \c{133}\tm{6.12} 
    \\
    & \TOP\BOT \textsf{timeout}
    & --- 
    & --- 
    & --- 
    & --- 
    & \c{82}\tm{25.0} 
    \\
    \hline
  \end{tabular}
\caption{Empirical Evaluation, comparing syntactic to semantic techniques.}
\label{tbl:exp1}
\end{table}

It is immediate that syntactic techniques cannot compete with the expressive 
power of interpretations. 
In Testbed~\textsf{TC} there are in fact only three examples 
compatible with \POPSTARP\ where \TCT~could not find interpretations.
There are additionally four examples compatible with \LMPO\ but not so with interpretations, 
including the TRS $\RSbin$ from Example~\ref{ex:RS2}. 
All but one (noteworthy the merge-sort algorithm from Steinbach and K\"uhlers collection 
\cite[Example~2.43]{SK90}) 
of these do in fact admit exponential runtime complexity, 
thus a~priori they are not compatible to the restricted interpretations.
%

We emphasise that parameter substitution significantly increases the strength of 
\POPSTAR, 13 examples are provable by \POPSTARP\ but neither by \POPSTAR\ nor \LMPO.\@
\LMPO\ could benefit from parameter substitution, 
we conjecture that the resulting order is still sound for $\FP$. 

In sum on Testbed~\textsf{TCO}, containing only orthogonal TRSs, 
in total 75 systems (26\% of the testbed)
can be verified to encode polytime computable functions, 35 (12\% of the testbed)
can be verified polytime computable by only syntactic techniques. 
It should be noted that not all examples appearing in our collection encode polytime computable 
functions, the total amount of such systems is unknown. 

It seems that Table~\ref{tbl:exp1} clearly shows the weakness of polynomial
path orders (even with parameter substitution) for automated polynomial
runtime complexity. However, remark the average execution times provided.
\POPSTARP\ succeeds on average 14 times faster than polynomial and 
matrix interpretations. Here the difficulty of implementing interpretations efficiently 
is also reflected in the total number of timeouts.
Furthermore note that a competitive complexity
analyser cannot be based on direct techniques alone. 
Instead, our complexity analyser \TCT\ recursively decomposes complexity problems using 
various complexity preserving transformation techniques~\cite{AM13}, 
discarding those problems that can be handled by basic techniques as 
contrasted in Table~\ref{tbl:exp1}.
Certificates are only obtained
if finally all subproblems can be discarded,
above all it is crucial that subproblems can be discarded 
quickly. Due to the efficiency of syntactic methods, these 
can be safely preposed to semantic techniques, thus speeding up the overall procedure.


\section{Conclusion and Future Work}\label{s:conclusion}

This paper is concerned with the complexity analysis of
constructor term rewrite systems and its ramification in implicit
computational complexity.

We have proposed a path order with multiset status, 
the polynomial path order \POPSTAR. 
The order $\POPSTAR$ is a syntactical restriction of multiset path orders, 
with the distinctive feature that the innermost runtime complexity
of compatible TRSs lies in $O(n^d)$ for some $d$.
Based on $\POPSTAR$, we delineate a class of rewrite systems, dubbed
systems of predicative recursion, 
so that the class of functions computed by these systems
corresponds to $\FP$, the class of polytime computable functions.
We have shown that an extension of $\POPSTAR$, the order $\POPSTARP$
that also accounts for parameter substitution, 
increases the intensionality of $\POPSTAR$.

From the viewpoint of implicit computational complexity we
have provided new implicit characterisations of the
class of polytime functions. More precisely, \POPSTAR\ and \POPSTARP\ 
are sound for the class of function problems $\FNP$ and are readily
applicable to obtain exact characterisations of the polytime
computable functions. As an easy corollary, we have given
an alternative proof of Bellantoni's result that the polytime
computable functions are closed under parameter substitution. 

From the viewpoint of (automated) runtime complexity analysis
we have proposed two new syntactic techniques to establish
polynomial innermost runtime complexity. 
In contrast to semantic techniques polynomial path orders 
are partly lacking in intensionality but greatly surpluses
in verification time. Note that in our complexity prover \TCT, we do not intend to replace 
semantic techniques, but rather prepose them by \POPSTARP, in 
order to improve \TCT\ both in analytic power and speed. 

In runtime complexity analysis one is in particular interested 
in obtaining asymptotically tight bounds. 
Although we could estimate the degree of the witnessing
bounding function for \POPSTAR\ and \POPSTARP, 
such a bound would be a gross overestimation.
This is partly due to the underlying multiset extension.

Very recently, together with Eguchi we have proposed a
simplification of the polynomial path orders studied here: 
the \emph{small polynomial path orders} (\emph{\POPSTARS} for short).
This termination order entails a finer control on the runtime
complexity: for any rewrite system compatible with \POPSTARS\
that employs recursion upto depth $d$, the innermost runtime complexity 
is polynomially bounded of degree $d$. This bound is tight,
see~\cite{AEM12}.
This becomes possible, as the underlying scheme of safe composition
is restricted to so-called \emph{weak safe composition}.


\section*{Acknowledgement}

We are in particular thankful to Nao Hirokawa for fruitful discussions.
Furthermore the second author would like to thank Toshiyasu Arai for having
introduced him to the topic of predicative recursion.
Finally, we are indebted to the annonymous reviewers for their constructive criticism.

\bibliographystyle{plainnat}

\begin{thebibliography}{59}
\providecommand{\natexlab}[1]{#1}
\providecommand{\url}[1]{\texttt{#1}}
\expandafter\ifx\csname urlstyle\endcsname\relax
  \providecommand{\doi}[1]{doi: #1}\else
  \providecommand{\doi}{doi: \begingroup \urlstyle{rm}\Url}\fi

\bibitem[Albert et~al.(2009)Albert, Arenas, Genaim, G{\'o}mez-Zamalloa, Puebla,
  Ram\'{\i}rez, Rom{\'a}n, and Zanardini]{AAGGPRRZ:2009}
E.~Albert, P.~Arenas, S.~Genaim, M.~G{\'o}mez-Zamalloa, G.~Puebla,
  D.~Ram\'{\i}rez, G.~Rom{\'a}n, and D.~Zanardini.
\newblock {Termination and Cost Analysis with COSTA and its User Interfaces}.
\newblock \emph{Electronic Notes in Theoretical Computer Science}, 258\penalty0
  (1):\penalty0 109--121, 2009.

\bibitem[Alias et~al.(2010)Alias, Darte, Feautrier, and Gonnord]{ADFG10}
C.~Alias, A.~Darte, P.~Feautrier, and L.~Gonnord.
\newblock {Multi-dimensional Rankings, Program Termination, and Complexity
  Bounds of Flowchart Programs}.
\newblock In \emph{Proc.\ of \nth{17} SAS}, volume 6337 of \emph{Lecture Notes
  in Computer Science}, pages 117--133, 2010.

\bibitem[Arai and Moser(2004)]{AM04}
T.~Arai and G.~Moser.
\newblock {A Note on a Term Rewriting Characterization of PTIME}.
\newblock In \emph{Proc.\ of \nth{7} WST}, pages 10--13. number AIB-2004-07 of
  Aachener Informatik-Berichte, 2004.
\newblock Extended abstract.

\bibitem[Arai and Moser(2005)]{AM05}
T.~Arai and G.~Moser.
\newblock {Proofs of Termination of Rewrite Systems for Polytime Functions}.
\newblock In \emph{{Proc.\ of \nth{25} FSTTCS}}, volume 3821 of \emph{Lecture
  Notes in Computer Science}, pages 529--540. Springer Verlag, 2005.

\bibitem[Avanzini and Moser(2008)]{AM08}
M.~Avanzini and G.~Moser.
\newblock {Complexity Analysis by Rewriting}.
\newblock In \emph{Proc.\ of \nth{9} FLOPS}, volume 4989 of \emph{Lecture Notes
  in Computer Science}, pages 130--146. Springer Verlag, 2008.

\bibitem[Avanzini and Moser(2009{\natexlab{a}})]{AM09}
M.~Avanzini and G.~Moser.
\newblock {Dependency Pairs and Polynomial Path Orders}.
\newblock In \emph{Proc.\ of \nth{20} RTA}, volume 5595 of \emph{Lecture Notes
  in Computer Science}, pages 48--62. Springer Verlag, 2009{\natexlab{a}}.

\bibitem[Avanzini and Moser(2009{\natexlab{b}})]{AM09b}
M.~Avanzini and G.~Moser.
\newblock {Polynomial Path Orders and the Rules of Predicative Recursion with
  Parameter Substitution}.
\newblock In \emph{Proc.\ of \nth{10} WST}, pages 16--20, 2009{\natexlab{b}}.

\bibitem[Avanzini and Moser(2010{\natexlab{a}})]{AM10}
M.~Avanzini and G.~Moser.
\newblock {Complexity Analysis by Graph Rewriting}.
\newblock In \emph{Proc.\ of \nth{10} FLOPS}, volume 6009 of \emph{Lecture
  Notes in Computer Science}, pages 257--271. Springer Verlag,
  2010{\natexlab{a}}.

\bibitem[Avanzini and Moser(2010{\natexlab{b}})]{AM10b}
M.~Avanzini and G.~Moser.
\newblock {Closing the Gap Between Runtime Complexity and Polytime
  Computability}.
\newblock In \emph{Proc.\ of \nst{21} RTA}, volume~6 of \emph{Leibniz
  International Proceedings in Informatics}, pages 33--48, 2010{\natexlab{b}}.

\bibitem[Avanzini and Moser(2013{\natexlab{a}})]{AM13}
M.~Avanzini and G.~Moser.
\newblock {A Combination Framework for Complexity}.
\newblock In \emph{Proc.\ 24th RTA}, volume~21, pages 55--70. Leibniz
  International Proceedings in Informatics, 2013{\natexlab{a}}.

\bibitem[Avanzini and Moser(2013{\natexlab{b}})]{AM13b}
M.~Avanzini and G.~Moser.
\newblock {Tyrolean Complexity Tool}: {F}eatures and usage.
\newblock In \emph{Proc.\ 24th RTA}, Leibniz International Proceedings in
  Informatics, pages 71--80, 2013{\natexlab{b}}.
\newblock 21.

\bibitem[Avanzini et~al.(2008)Avanzini, Moser, and Schnabl]{AMS08}
M.~Avanzini, G.~Moser, and A.~Schnabl.
\newblock {Automated Implicit Computational Complexity Analysis (System
  Description)}.
\newblock In \emph{Proc.\ of \nth{4} IJCAR}, volume 5195 of \emph{Lecture Notes
  in Computer Science}, pages 132--139. Springer Verlag, 2008.

\bibitem[Avanzini et~al.(2012)Avanzini, Eguchi, and Moser]{AEM12}
M.~Avanzini, N.~Eguchi, and G.~Moser.
\newblock {A New Order-theoretic Characterisation of the Polytime Computable
  Functions}.
\newblock In \emph{Proc.\ of \nth{10} APLAS}, volume 7705 of \emph{Lecture
  Notes in Computer Science}, pages 280--295, 2012.

\bibitem[Baader and Nipkow(1998)]{BN98}
F.~Baader and T.~Nipkow.
\newblock \emph{{Term Rewriting and All That}}.
\newblock Cambridge University Press, 1998.

\bibitem[Baillot et~al.(2009)Baillot, Marion, and Rocca]{BMR09}
P.~Baillot, J.-Y. Marion, and S.~Ronchi~Della Rocca.
\newblock {Guest Editorial: Special Issue on Implicit Computational
  Complexity}.
\newblock \emph{ACM Transactions on Computational Logic}, 10\penalty0 (4),
  2009.

\bibitem[Beckmann and Weiermann(1996)]{BW96}
A.~Beckmann and A.~Weiermann.
\newblock {A Term Rewriting Characterization Of the Polytime Functions and
  Related Complexity Classes}.
\newblock \emph{Archive for Mathematical Logic}, 36:\penalty0 11--30, 1996.

\bibitem[Bellantoni(1992)]{B:92}
S.~Bellantoni.
\newblock \emph{Predicative Recursion and Computational Complexity}.
\newblock PhD thesis, University of Torronto, Faculty for Computer Science,
  1992.

\bibitem[Bellantoni and Cook(1992)]{BC92}
S.~Bellantoni and S.~Cook.
\newblock {A new Recursion-Theoretic Characterization of the Polytime
  Functions}.
\newblock \emph{Computational Complexity}, 2\penalty0 (2):\penalty0 97--110,
  1992.

\bibitem[Boas(1990)]{Boas:TCS:90}
P.~Van~Emde Boas.
\newblock {Machine Models and Simulation}.
\newblock In \emph{Handbook of Theoretical Computer Science, Volume A:
  Algorithms and Complexity (A)}, pages 1--66. The MIT Press, 1990.

\bibitem[Bonfante and Moser(2010)]{BonfanteMoser:2010}
G.~Bonfante and G.~Moser.
\newblock {Characterising Space Complexity Classes via {Knuth-Bendix} Orders}.
\newblock In \emph{Proc.\ of \nth{17} LPAR}, volume 6397 of \emph{Lecture Notes
  in Computer Science}, pages 142--156, 2010.

\bibitem[Bonfante et~al.(2001)Bonfante, Cichon, Marion, and Touzet]{BCMT01}
G.~Bonfante, A.~Cichon, J.-Y. Marion, and H.~Touzet.
\newblock {Algorithms with Polynomial Interpretation Termination Proof}.
\newblock \emph{Journal of Functional Programming}, 11\penalty0 (1):\penalty0
  33--53, 2001.

\bibitem[Bonfante et~al.(2011)Bonfante, Marion, and Moyen]{BMM11}
G.~Bonfante, J.-Y. Marion, and J.-Y. Moyen.
\newblock {Quasi-interpretations: A Way to Control Resources}.
\newblock \emph{Theoretical Computer Science}, 412\penalty0 (25), 2011.

\bibitem[Buchholz(1995)]{B95}
W.~Buchholz.
\newblock {Proof-theoretical Analysis of Termination Proofs}.
\newblock \emph{Annals of Pure and Applied Logic}, 75:\penalty0 57--65, 1995.

\bibitem[Cichon and Weiermann(1997)]{CW97}
E.~A. Cichon and A.~Weiermann.
\newblock {Term Rewriting Theory for the Primitive Recursive Functions}.
\newblock \emph{Annals of Pure and Applied Logic}, 83\penalty0 (3):\penalty0
  199--223, 1997.

\bibitem[{Dal Lago} and Martini(2009{\natexlab{a}})]{LM09}
U.~{Dal Lago} and S.~Martini.
\newblock {On {C}onstructor {R}ewrite {S}ystems and the {L}ambda-{C}alculus}.
\newblock In \emph{Proc.\ of \nth{36} ICALP}, volume 5556 of \emph{Lecture
  Notes in Computer Science}, pages 163--174. Springer Verlag,
  2009{\natexlab{a}}.

\bibitem[{Dal Lago} and Martini(2009{\natexlab{b}})]{LM:2009b}
U.~{Dal Lago} and S.~Martini.
\newblock {Derivational Complexity is an Invariant Cost Model}.
\newblock In \emph{Proc.\ of \nst{1} FOPARA}, 2009{\natexlab{b}}.

\bibitem[E{\'e}n and S{\"o}rensson(2003)]{ES03}
Niklas E{\'e}n and Niklas S{\"o}rensson.
\newblock {An Extensible SAT-solver}.
\newblock In \emph{Proc.\ of \nth{6} SAT}, volume 2919 of \emph{Lecture Notes
  in Computer Science}, pages 502--518. Springer Verlag, 2003.

\bibitem[Ferreira(1995)]{Ferreira95}
M.~C.~F. Ferreira.
\newblock \emph{{Termination of Term Rewriting}}.
\newblock PhD thesis, University of Utrecht, November 1995.
\newblock Well-foundedness, Totality and Transformations.

\bibitem[Gulwani et~al.(2009)Gulwani, Mehra, and Chilimbi]{GMC09}
S.~Gulwani, K.K. Mehra, and T.M. Chilimbi.
\newblock {SPEED: Precise and Efficient Static Estimation of Program
  Computational Complexity}.
\newblock In \emph{Proc.\ of \nth{36} POPL}, pages 127--139. Association for
  Computing Machinery, 2009.

\bibitem[Hirokawa and Moser(2008)]{HM08}
N.~Hirokawa and G.~Moser.
\newblock {Automated Complexity Analysis Based on the Dependency Pair Method}.
\newblock In \emph{Proc.\ of \nth{4} IJCAR}, volume 5195 of \emph{Lecture Notes
  in Artificial Inteligence}, pages 364--380. Springer Verlag, 2008.

\bibitem[Hirokawa and Moser(2011)]{HM11}
N.~Hirokawa and G.~Moser.
\newblock {Automated Complexity Analysis Based on the Dependency Pair Method}.
\newblock \emph{CoRR}, abs/1102.3129, 2011.
\newblock submitted.

\bibitem[Hofbauer(1992)]{H92}
D.~Hofbauer.
\newblock {Termination Proofs by Multiset Path Orderings Imply Primitive
  Recursive Derivation Lengths}.
\newblock \emph{Theoretical Computer Science}, 105:\penalty0 129--140, 1992.

\bibitem[Hofbauer and Lautemann(1989)]{HL89}
D.~Hofbauer and C.~Lautemann.
\newblock {Termination Proofs and the Length of Derivations}.
\newblock In \emph{Proc.\ of \nrd{3} RTA}, volume 355 of \emph{Lecture Notes in
  Computer Science}, pages 167--177. Springer Verlag, 1989.

\bibitem[Hoffmann et~al.(2011)Hoffmann, Aehlig, and Hofmann]{HAH11}
J.~Hoffmann, K.~Aehlig, and M.~Hofmann.
\newblock {Multivariate Amortized Resource Analysis}.
\newblock In \emph{Proc.\ of \nth{38} POPL}, pages 357--370. Association for
  Computing Machinery, 2011.

\bibitem[Hoffmann et~al.(2012)Hoffmann, Aehlig, and Hofmann]{HAH12}
J.~Hoffmann, K.~Aehlig, and M.~Hofmann.
\newblock {Resource Aware ML}.
\newblock In \emph{CAV}, volume 7358 of \emph{Lecture Notes in Computer
  Science}, pages 781--786, 2012.

\bibitem[Hofmann and Rodriguez(2013)]{HR13}
M.~Hofmann and D.~Rodriguez.
\newblock {Automatic Type Inference for Amortised Heap-Space Analysis}.
\newblock In \emph{Proc.\ of \nnd{22} ESOP}, volume 7792 of \emph{Lecture Notes
  in Computer Science}, pages 593--613, 2013.

\bibitem[Knuth(1993)]{K93}
D.~E. Knuth.
\newblock {Johann Faulhaber and Sums of Powers}.
\newblock \emph{MC}, 203:\penalty0 277--294, 1993.

\bibitem[Lago(2011)]{DalLago:2011}
U.~Dal Lago.
\newblock {A Short Introduction to Implicit Computational Complexity}.
\newblock In \emph{Lectures on Logic and Computation - ESSLLI 2010 Copenhagen,
  Denmark, August 2010, ESSLLI 2011, Ljubljana, Slovenia, August 2011, Selected
  Lecture Notes}, volume 7388 of \emph{Lecture Notes in Computer Science},
  pages 89--109, 2011.

\bibitem[Lago et~al.(2010)Lago, Martini, and Zorzi]{DLMZ:10}
U.~Dal Lago, S.~Martini, and M.~Zorzi.
\newblock {General Ramified Recurrence is Sound for Polynomial Time}.
\newblock In \emph{Proc.\ of DICE2010}, volume~23 of \emph{Electronic
  Proceedings in Theoretical Computer Science}, pages 47--62, 2010.

\bibitem[Leivant(1990)]{Leivant:1990}
D.~Leivant.
\newblock Subrecursion and lambda representation over free algebras
  (preliminary summary).
\newblock In \emph{Feasible mathematics (Ithaca, NY, 1989)}, Progr. Comput.
  Sci. Appl. Logic, pages 281--291. Birkhäuser Boston, 1990.

\bibitem[Leivant(1991)]{L91}
D.~Leivant.
\newblock {A Foundational Delineation of Computational Feasiblity}.
\newblock In \emph{Proc.\ of \nth{6} LICS}, pages 2--11. IEEE Computer Society,
  1991.

\bibitem[Leivant(1993)]{Leivant93}
D.~Leivant.
\newblock {Stratified Functional Programs and Computational Complexity}.
\newblock In \emph{Proc.\ of \nth{20} POPL}, pages 325--333. ACM Press, 1993.

\bibitem[Marion(2003)]{M03}
J.-Y. Marion.
\newblock {Analysing the Implicit Complexity of Programs}.
\newblock \emph{Information and Computation}, 183:\penalty0 2--18, 2003.

\bibitem[Marion and P{\'e}choux(2009)]{MP:09}
J.-Y. Marion and R.~P{\'e}choux.
\newblock {Sup-interpretations, a Semantic Method for Static Analysis of
  Program Resources}.
\newblock \emph{ACM Trans. Comput. Log.}, 10\penalty0 (4), 2009.

\bibitem[Middeldorp et~al.(2011)Middeldorp, Moser, Neurauter, Waldmann, and
  Zankl]{MMNWZ11}
A.~Middeldorp, G.~Moser, F.~Neurauter, J.~Waldmann, and H.~Zankl.
\newblock {Joint Spectral Radius Theory for Automated Complexity Analysis of
  Rewrite Systems}.
\newblock In \emph{Proc.\ of \nth{4} CAI}, volume 6472 of \emph{Lecture Notes
  in Computer Science}, pages 1--20. Springer Verlag, 2011.

\bibitem[Moser and Schnabl(2008)]{MS08}
G.~Moser and A.~Schnabl.
\newblock Proving quadratic derivational complexities using context dependent
  interpretations.
\newblock In \emph{Proc.\ of \nth{19} RTA}, volume 5117 of \emph{Lecture Notes
  in Computer Science}, pages 276--290, 2008.

\bibitem[Moser and Weiermann(2003)]{MW03}
G.~Moser and A.~Weiermann.
\newblock Relating derivation lengths with the slow-growing hierarchy directly.
\newblock In \emph{Proc.\ of \nth{14} RTA}, volume 2706 of \emph{Lecture Notes
  in Computer Science}, pages 296--310, 2003.

\bibitem[Noschinski et~al.(2011)Noschinski, Emmes, and Giesl]{NEG11}
L.~Noschinski, F.~Emmes, and J.~Giesl.
\newblock {A Dependency Pair Framework for Innermost Complexity Analysis of
  Term Rewrite Systems}.
\newblock In \emph{Proc.\ of \nrd{23} CADE}, Lecture Notes in Computer Science,
  pages 422--438. Springer Verlag, 2011.

\bibitem[Papadimitriou(1995)]{Papa}
Christos~H. Papadimitriou.
\newblock \emph{{C}omputational {C}omplexity}.
\newblock {A}ddison {W}esley {L}ongman, second edition, 1995.

\bibitem[Plaisted and Greenbaum(1986)]{PG86}
D.~A. Plaisted and S.~Greenbaum.
\newblock {A Structure-Preserving Clause Form Translation}.
\newblock \emph{Journal of Symbolic Computation}, 2\penalty0 (3):\penalty0
  293--304, 1986.

\bibitem[Schneider-Kamp et~al.(2007{\natexlab{a}})Schneider-Kamp, Fuhs,
  Thiemann, Giesl, Annov, Codish, Middeldorp, and Zankl]{SFTGACMZ07}
P.~Schneider-Kamp, C.~Fuhs, R.~Thiemann, J.~Giesl, E.~Annov, M.~Codish,
  A.~Middeldorp, and H.~Zankl.
\newblock Implementing {RPO} and {POLO} {U}sing {SAT}.
\newblock In \emph{DDP}, number 07401 in Leibniz International Proceedings in
  Informatics. Dagstuhl, 2007{\natexlab{a}}.

\bibitem[Schneider-Kamp et~al.(2007{\natexlab{b}})Schneider-Kamp, Thiemann,
  Annov, Codish, and Giesl]{SK07}
P.~Schneider-Kamp, R.~Thiemann, E.~Annov, M.~Codish, and J.~Giesl.
\newblock {Proving Termination Using Recursive Path Orders and SAT Solving}.
\newblock In \emph{Proc.\ of \nth{6} FroCoS}, volume 4720 of \emph{Lecture
  Notes in Computer Science}, pages 267--282. Springer Verlag,
  2007{\natexlab{b}}.

\bibitem[Simmons(1988)]{Simmons:1988}
H.~Simmons.
\newblock {The Realm of Primitive Recursion}.
\newblock \emph{Applied Mathematicas Letters}, 27:\penalty0 177--188, 1988.

\bibitem[Steinbach and K{\"u}hler(1990)]{SK90}
J.~Steinbach and U.~K{\"u}hler.
\newblock {Check your Ordering - Termination Proofs and Open Problems}.
\newblock Technical Report SEKI-Report SR-90-25, University of Kaiserslautern,
  1990.

\bibitem[Tarjan(1985)]{Tarjan:1985}
R.E. Tarjan.
\newblock {Amortized Computational Complexity}.
\newblock \emph{SIAM J.~Alg.\ Disc.\ Meth}, 6\penalty0 (2):\penalty0 306--318,
  1985.

\bibitem[Te{R}e{S}e(2003)]{TeReSe}
Te{R}e{S}e.
\newblock \emph{{Term Rewriting Systems}}, volume~55 of \emph{Cambridge Tracks
  in Theoretical Computer Science}.
\newblock Cambridge University Press, 2003.

\bibitem[Zankl and Korp(2010)]{HZMK10}
H.~Zankl and M.~Korp.
\newblock {Modular Complexity Analysis via Relative Complexity}.
\newblock In \emph{Proc.\ of \nst{21} RTA}, volume~6 of \emph{Leibniz
  International Proceedings in Informatics}, pages 385--400, 2010.

\bibitem[Zankl and Middeldorp(2007)]{ZM07}
H.~Zankl and A.~Middeldorp.
\newblock {Satisfying KBO Constraints}.
\newblock In \emph{Proc.\ of the \nth{18} RTA}, volume 4533 of \emph{Lecture
  Notes in Computer Science}, pages 389--403. Springer Verlag, 2007.

\bibitem[Zuleger et~al.(2011)Zuleger, Gulwani, Sinn, and Veith]{ZulegerGSV11}
F.~Zuleger, S.~Gulwani, M.~Sinn, and H.~Veith.
\newblock {Bound Analysis of Imperative Programs with the Size-Change
  Abstraction}.
\newblock In \emph{Proc.\ of \nth{18} SAS}, volume 6887 of \emph{Lecture Notes
  in Computer Science}, pages 280--297. Springer Verlag, 2011.

\end{thebibliography}

\end{document}